\documentclass[11pt,a4paper]{article}%
\usepackage{amssymb,amsmath, amsfonts}
\usepackage{graphicx,graphics}
\usepackage{mathtools}
\usepackage{bbold}
\usepackage[english]{babel}
\usepackage[utf8]{inputenc}
\usepackage{epsfig,url}
\usepackage{bbm,theorem}
\usepackage{a4wide}
\usepackage{color}
\usepackage{enumerate}
\usepackage{calrsfs}
\usepackage[bookmarks=true]{hyperref}
\usepackage{bookmark}
\usepackage{amsmath}
\usepackage{amsfonts}
\usepackage{amssymb}
\usepackage{graphicx}%
\providecommand{\U}[1]{\protect\rule{.1in}{.1in}}
\DeclareMathAlphabet{\pazocal}{OMS}{zplm}{m}{n}
\newtheorem{theorem}{Theorem}

\newtheorem{claim}[theorem]{Claim}

\newtheorem{corollary}[theorem]{Corollary}

\newtheorem{definition}[theorem]{Definition}

\newtheorem{remark}[theorem]{Remark}

\newenvironment{proof}[1][Proof]{\noindent\textbf{#1.} }{\ \rule{0.5em}{0.5em}}

\newenvironment{proofof}[1][Proof]{\noindent \textbf{#1.} } {\ \rule{0.5em}{0.5em}} %{\ \qed}

\numberwithin{equation}{section}
\numberwithin{theorem}{section}

\newcommand{\R}{{\mathbb R}}

\newcommand{\ep}{\varepsilon}

\newcommand{\beq}{\begin{equation}}
\newcommand{\eeq}{\end{equation}}
\newcommand{\beqs}{\begin{eqnarray}}
\newcommand{\eeqs}{\end{eqnarray}}

\newcommand{\calC}{{\cal C}}

\newcounter{jlisti}

\def\be{\begin{equation}}
\def\ee{\end{equation}}
\def\bea{\begin{eqnarray}}
\def\eea{\end{eqnarray}}
\def\ni{\noindent}
\def\nn{\nonumber}

\begin{document}

\title{On the theory of Lorentz gases with long range interactions}
\author{Alessia Nota \thanks{E-mail: \texttt{nota@iam.uni-bonn.de}} , \ Sergio
Simonella \thanks{E-mail: \texttt{s.simonella@tum.de}}, Juan J. L. Vel\'{a}zquez
\thanks{E-mail:\texttt{velazquez@iam.uni-bonn.de}}, }
\date{\today }
\maketitle

\vspace{1cm}
\noindent
{\bf Abstract.} We construct and study the stochastic force field generated by a Poisson distribution of 
sources at finite density, $x_1,x_2,\cdots$ in $\mathbb{R}^3$ each of them yielding a long range potential $Q_i\Phi(x-x_i)$ with possibly different
charges $Q_i \in \mathbb{R}$. The potential $\Phi$ is assumed to behave typically as $|x|^{-s}$ for large
$|x|$, with $s > 1/2$. We will denote the resulting random field as  ``generalized Holtsmark field''. 
We then consider the dynamics of one tagged particle in such random force fields, in several scaling limits
where the mean free path is much larger than the average distance between the scatterers. 
We estimate the diffusive time scale and identify conditions for the vanishing of correlations. 
These results are used to obtain appropriate kinetic descriptions in terms of a linear Boltzmann or Landau evolution 
equation depending on the specific choices of the interaction potential.

\vspace{0.5cm}\noindent
{\bf Keywords.} Lorentz model, Holtsmark field, kinetic equations, power law forces, collisions, diffusion.
\vspace{0.5cm}

\tableofcontents

\section{Introduction.}

In this paper we study the evolution of a classical Newtonian particle moving
in the three-dimensional space, in the field of forces produced by a random
distribution of fixed sources (scatterers). We will assume that the
force field acting over a particle at the position $x$ due to a scatterer
centered at $x_{j}$ is given by $-Q_j\nabla\Phi\left(x-x_{j}\right)$, $Q_j\in\mathbb{R}$, where the potential
$\Phi(x)$ is radially symmetric and behaves typically as $\left\vert x\right\vert^{-s}$ for large $\left\vert x\right\vert $. 
The case of gravitational or Coulombian scatterers corresponds to potentials
proportional to $\left\vert x\right\vert^{-1}$.

The dynamics of tagged particles in fixed centers of scatterers has been
extensively studied, see \cite{Al78, vB82, Ha74, LL2} for classical surveys on the topic.
Such systems are known generally as Lorentz gases, 
since they were proposed by Lorentz in 1905 to explain the motion of
electrons in metals \cite{L}. 
The scatterers are assumed to be short ranged and, 
in the classical setting, they are modelled as elastic hard spheres. A case of
weak, long range field with diffusion has been studied in \cite{Pi81}. 

The statistical properties of the gravitational field generated by a
Poisson repartition of point masses with finite density are also well known. The associated
distribution, known as Holtsmark field, has been investigated in connection with
several applications in spectroscopy and astronomy \cite{CH,F, H, K}. It is a symmetric stable 
distribution with parameter $3/2$, skewness parameter $0$ and semiexplicit form of the density function.
In the present paper, we intend to study generalizations of such
distribution for a large class of random potentials $\Phi$. The main motivation is to
clarify how the tracer particle dynamics depends on the details of $\Phi$.

It is not an easy task to construct the dynamics of a tagged particle for
arbitrary $\Phi$.  
Nevertheless, it is possible to obtain a rather detailed information in the `kinetic limit'. 
We focus on a family of potentials $\left\{  \Phi\left(x,\varepsilon\right)  \right\}  _{\varepsilon>0}$ 
and denote by $\ell_\varepsilon$ the {\em mean free path}, namely the characteristic length travelled by a particle before
its velocity vector $v$ changes by an amount of the same order of magnitude as the velocity itself. 
We will denote by $d$ the typical distance between scatterers and we will use this distance as unit of length. % ($d=O(1)$ in our setting), 
We will then fix the unit of time in order to make the speed of the tagged particle of order one.
We are interested in potentials for which
\begin{equation}
\lim_{\varepsilon\rightarrow0}\frac{\ell_{\varepsilon}}{d}=\lim_{\varepsilon
\rightarrow0}\ell_{\varepsilon}=\infty\label{KL}\;.%
\end{equation}
In particular, the forces produced are small at distances of order one from the scatterers
(but possibly of order 1 at small distances).  
Since the location of the centers of force is random, the position of
the particle is also a random variable. To describe the evolution, we therefore
compute the probability density $f_{\varepsilon}$ 
of finding the particle at a given point $x$ with a given value of the
velocity $v$ at time $t \geq 0$. If an additional condition concerning the independence of the deflections in
macroscopic time scales $\ell_{\varepsilon}$ holds
as $\varepsilon\to 0$, then on a macroscopic scale of space and time $f_{\varepsilon}$ approximates a function 
$f$ governed by a kinetic equation. More precisely, $f_{\varepsilon}\left(\ell_{\varepsilon}t,\ell_{\varepsilon}x,v\right)\to f(t,x,v)$  which solves one of the following equations:

\smallskip
\noindent
1) a linear Boltzmann equation of the form
\begin{equation}
\left(  \partial_{t}f+v\partial_{x}f\right)  \left(  t,x,v\right)
=\int_{S^{2}}B\left(  v;\omega \right)  \left[  f\left(  t,x,\left\vert
v\right\vert \omega\right)  -f\left(  t,x,v\right)  \right]  d\omega\;, \label{BE}%
\end{equation}
where we denoted by $\partial_{x}$ the three dimensional gradient and $B$ is some nonnegative collision kernel (depending on $\Phi$);

\noindent
2) a linear Landau type equation 
\begin{equation}
\left(  \partial_{t}f+v\partial_{x}f\right)  \left(  x,v,t\right)
=\kappa\Delta_{v_{\bot}}f\left(  x,v,t\right)\;,\label{LE}%
\end{equation}
where $\kappa>0$ depends on $\Phi$.

\smallskip

\noindent 
The form of the equation derived, as well as the collision kernel $B$ and the diffusion coefficient $\kappa$ depend on the specific families of potentials $\left\{  \Phi\left(x,\varepsilon\right)  \right\}  _{\varepsilon>0}$ under consideration (see Section \ref{GenKinEq} for further details).

Additionally, there are families of potentials $\left\{  \Phi\left(
x,\varepsilon\right)  \right\}  _{\varepsilon>0}$ for which the evolution
equation for $f$ contains a sum of both a Boltzmann term as in (\ref{BE}), and other
terms as in (\ref{LE}). In other cases the model describing the dynamics of the tagged particle must take into account the non-negligible correlations between the velocities at points placed within a distance of the order of the mean free path.

Heuristically, (\ref{BE}) describes dynamics in which the main deflections
are binary collisions taking place when the tagged particle approaches within 
a distance $\lambda_{\varepsilon}$ of one of the scatterer centers, with 
$\lambda_{\varepsilon}>0$ converging to $0$ as $\varepsilon\rightarrow0$.
This does not imply that the interaction potential $\Phi$ should be
compactly supported or have a very fast decay within distances of order
$\lambda_{\varepsilon}$ (cf.\,Remark \ref{BoSD}). 
The quantity $\lambda_{\varepsilon}$ is a fundamental length of the
problem and we shall refer to it (when it exists) as {\em collision length} 
associated to $\left\{  \Phi\left(x,\varepsilon\right)  \right\}  _{\varepsilon>0}$.
The kinetic equation (\ref{LE}) characterizes cases in which $\lambda_{\varepsilon}$
is not reached, meaning that
particle deflections are only due to the addition of very small forces produced by 
the scatterers. In turn, the latter process can arise in different ways, depending on the
specific $\Phi$. For instance, at any given time, 
one can have a huge number of scatterers 
producing a relevant (small) force in the tagged particle,
%, i.e. contributing to the deflection of the tagged particle,
or no scatterer or just one scatterer producing a relevant force, in such a way that the 
accumulation of many of these small, binary interactions yields an important deflection 
of the trajectory on the large time scale. We will discuss several of these
possibilities in Section \ref{Examples}.

A mathematical derivation of the kinetic equations (\ref{BE}) and (\ref{LE})
has been provided in cases of compactly supported potentials, 
in the so-called low density and weak coupling limits respectively.
We refer to \cite{DGL, G, KP}  for first results in these directions, to
\cite{S1} (Chapter I.8) and to \cite{S2} for an account of the subsequent literature. 
An alternative way of deriving a linear Landau equation is to consider multiple weak interactions of the tagged particle with the scatterers whose density is intermediate between the Boltzmann-Grad regime and the weak-coupling regime; see for instance \cite{BNP, DR, MN} (and Section \ref{DiffLandau} below).

As shown in
\cite{DP}, it is furthermore possible to derive Boltzmann equations in cases of 
potentials with diverging support. It is however unclear, even at a formal level,
what kinetic behaviour has to be expected for generic $\left\{  \Phi\left(x,\varepsilon\right)  \right\}  _{\varepsilon>0}$ 
and in particular for potentials of the form
\begin{equation}
\Phi\left(  x,\varepsilon\right)  =\frac{\varepsilon^{s}}{\left\vert
x\right\vert ^{s}}\label{SingP}%
\end{equation}
(for which $\lambda_{\varepsilon}=\varepsilon$).

We sketch now the main ideas behind the derivation of (\ref{BE}) and (\ref{LE})
(or combinations of them) for generic $\Phi$. We construct the
generalized Holtsmark field associated to $\left\{  \Phi\left(  x,\varepsilon\right)  \right\}
_{\varepsilon>0}$. 
These are random force fields obtained as the sum of the forces generated by a Poisson distribution of points. 
Their analysis was initiated by Holtsmark in \cite{H}. 

We assume condition (\ref{KL}) and also 
that the particle deflections are independent in the macroscopic time scale. 
We split (in a smooth manner) the potential
$\Phi\left(  x,\varepsilon\right)  $ as
\[
\Phi\left(  x,\varepsilon\right)  =\Phi_{B}\left(  x,\varepsilon\right)
+\Phi_{L}\left(  x,\varepsilon\right)
\]
where $\Phi_{B}\left(  x,\varepsilon\right)  $ is supported in a ball of
radius $M\lambda_{\varepsilon},$ with $M$ of order one but large, and
$\Phi_{L}\left(  x,\varepsilon\right)  $ is supported at distances larger than
$\frac{M\lambda_{\varepsilon}}{2}.$ 
At distances of order $\lambda_{\varepsilon}$, the particle is deflected 
for an amount of order one by the interaction $\Phi_{B}\left(  x,\varepsilon\right)  .$
Since the scatterers are distributed according to a Poisson repartition with finite density,
the time between two such consecutive collisions is of order of the
`Boltzmann-Grad time scale'
\begin{equation}
T_{BG}=\frac{1}{\left(  \lambda_{\varepsilon}\right)  ^{2}}\;. \label{BGT}%
\end{equation}
We compute next the time scale $T_{L}$ (`Landau time scale') in which
the deflections produced by $\Phi_{L}$ become relevant. 
Due to (\ref{KL}), we have
$T_{BG}\rightarrow\infty$ and $T_{L}\rightarrow\infty$ as $\varepsilon
\rightarrow0.$ We shall have then three different possibilities:
(i) if $T_{BG}\ll T_{L}$ as $\varepsilon\rightarrow0$, the
evolution of $f$ will be given by (\ref{BE}); (ii) if $T_{BG}\gg T_{L}$ as
$\varepsilon\rightarrow0$, $f$ will evolve according to
(\ref{LE}); (iii) if $T_{BG}$ and $T_{L}$ are of the same order of
magnitude, $f$ will be driven by a Boltzmann-Landau equation. In the case of the
potentials (\ref{SingP}) it turns out that $T_{BG}\ll T_{L}$ as
$\varepsilon\rightarrow0$ if $s>1$ and $T_{BG}\gg T_{L}$ as $\varepsilon
\rightarrow0$ if $1/2< s \leq 1$. The limitation $s>\frac 1 2$ is necessary to ensure that the random fields are not almost constant over distances of order $d$ (see Remark \ref{rem:s=1/2}).

Our analysis relies on the smallness of the total force field at distances of order one from % $O(1)$ from 
the scatterers. 
For slowly decaying potentials, these fields can be still constructed in the
system of infinitely many scatterers, thanks to the mutual 
cancellations of forces produced by a random, spatially homogeneous distribution 
of sources (notice that in Lorentz gases, contrarily to plasmas, no screening effects occur,
even when the scatterers contain charges with different signs). However, some 
extra conditions must be imposed, depending on the decay. Let us restrict to (\ref{SingP}) for definiteness.
Then, if $s>2$ the random force field can be written as a convergent series of the binary forces
produced by the scatterers. If $\frac{1}{2}<s\leq2$, it is possible to obtain
the random force fields at a given point as the limit of forces produced by
scatterers in finite clouds whose size tends to infinity.
Nevertheless, the average value of the force field at a given point will generally depend on the
geometry of the finite clouds. For $s>1$, a translation invariant field can be still obtained by imposing a symmetry condition on the cloud. 
Finally, if $s \leq 1$, the force fields become spatially inhomogeneous, 
unless we consider ``neutral'' distributions of scatterers (e.g. fields produced by charges $Q_j = \pm 1$ yielding
attractive and repulsive forces and with average zero charge).

In this paper we restrict the analysis to the three-dimensional case. We 
also comment on the analogue of our main results in two dimensions
(see Remarks \ref{rem:2D1}, \ref{rem:2D2}).

The paper is organized as follows. In the first part (Section
\ref{Holtsm}) we construct the generalized
Holtsmark fields. 
In the second part (Sections \ref{GenKinEq} and \ref{Examples}) we study the deflections
experienced by a tagged particle in several families of such distributions and deduce the
resulting kinetic equation for $f$. A discussion on inhomogeneous cases and concluding remarks
are collected in the last two sections (Sections \ref{Nonhomog} and \ref{sec:CR}).

\section{Generalized Holtsmark fields.\label{Holtsm}}

\subsection{Definitions.}

In this section we characterize in a precise mathematical form the random distribution
of forces in which the dynamics of the tagged particle will take place.

We will call {\em scatterer configuration} any countable collection of infinitely many points 
$\left\{x_{n}\right\}  _{n\in\mathbb{N}},$ $x_{n}\in\mathbb{R}^{3}$. We will denote by
$\Omega$ the set of {\em locally finite} scatterer configurations, i.e.\,such that $\#\left\{
x_{n}\right\}  _{n\in\mathbb{N}}\cap K<\infty$ for any compact $K\subset\mathbb{R}^{3}$. 
The associated $\sigma-$algebra $\mathcal{B}$ is generated by the sets 
$
\mathcal{U}_{U,m}=\Big\{\left\{  x_{n}\right\}
_{n\in\mathbb{N}},\ \#\left\{  x_{n}\right\}  _{n\in\mathbb{N}}\cap
U=m\Big\} 
$
for all bounded open $U\subset\mathbb{R}^{3}$ and $m=0,1,2,\cdots$. We
then have a measurable space $\left(  \Omega,\mathcal{B},\nu\right)  $
where $\nu$ is the probability measure associated to the Poisson distribution with
intensity one.

Given any finite set 
\begin{equation}
I=\left\{  Q_{1},Q_{2},Q_{3},...,Q_{L}\right\}, \quad Q_{j}\in\R  \label{def:set_I}%
\end{equation}
we introduce a probability measure $\mu$ in
the set $I$. 
We define the set of charged scatterer configurations $\Omega \otimes I $ as%
\[
\Omega \otimes I =\left\{  \left\{  \left(  x_{n},Q_{n}\right)  \right\}  _{n\in
\mathbb{N}},\ \left\{  x_{n}\right\}  _{n\in\mathbb{N}}\in{\Omega
},\ Q_{n}\in I\right\}\;.
\]
We define the $\sigma-$algebra $\mathcal{B}$ of subsets of $\Omega \otimes I$ as the one
generated by the sets $$\mathcal{U}_{U,m,j}=\Big\{\left\{  \left(  x_{n}%
,Q_{n}\right)  \right\}  _{n\in\mathbb{N}},\ \#\left\{  \left(  x_{n}%
,Q_{n}\right)  :Q_{n}=Q_{j},\ x_{n}\in U\right\}  =m\Big\},$$ where
$U\subset\mathbb{R}^{3}$ and $Q_{j}\in I.$
We then define a probability measure $\nu
\otimes\mu$ by means of%
\begin{equation}
(\nu\otimes\mu)\left(  \bigcap_{j=1}^{L}\left[  \mathcal{U}_{U,n_{j},j}\right]  \right)
=\frac{\exp\left(  -\left\vert U\right\vert \right)  \prod_{j=1}^{L}\left[
\mu\left(  Q_{j}\right)  \left\vert U\right\vert \right]  ^{n_{j}}}%
{\prod_{j=1}^{L}\left(  n_{j}\right)  !} \;.\label{def:measure}%
\end{equation}
Notice that (\ref{def:measure}) defines a probability measure because $\sum_{j=1}%
^{L}\mu\left(  Q_{j}\right)  =1.$ 

\begin{definition}
\label{RandForcField}Suppose that $I$ is a finite set as in \eqref{def:set_I} and let $\mu$ be a
probability measure on $I.$ We consider the measure $\nu\otimes\mu$ on $\Omega\otimes I$.  
We define random force field (associated to $\nu\otimes\mu$) 
the set of measurable mappings $\left\{  F\left(  x\right)
:x\in\mathbb{R}^{3}\right\}  $:
\begin{equation}
F\left(  x\right)  :\Omega\otimes I\rightarrow \mathbb{R}^3\cup \{\infty\}  
\ \ ,\ \ \ \ \omega\rightarrow F\left(  x\right)
\omega\label{S1E2}%
\end{equation}
satisfying $\left(  \nu\otimes\mu\right)  \left(  \left(  F\left(  x\right)
\right)  ^{-1}\left(  \left\{  \infty\right\}  \right)  \right)  =0$ for any
$x\in\mathbb{R}^{3}.$ 
\end{definition}

\noindent
When $I$ has just one element, we shorten the notation $\Omega\otimes I$
by $\Omega$ and $\nu\otimes\mu$ by $\nu$.

Let $\mathcal{B}\left(  \mathbb{R}^{3}\right)$ be the Borel algebra of $\mathbb{R}^{3}$.
We are interested in translation invariant fields, defined as follows:
\begin{definition}
\label{InvTrans}The random force field $\left\{  F\left(
x\right)  :x\in\mathbb{R}^{3}\right\}  $ is invariant under translations if
for any collection of points $x_{1},x_{2},\cdots\in\mathbb{R}^{3}$, for
any collection $A_{1},A_{2},\cdots\in\mathcal{B}\left(
\mathbb{R}^{3}\right)  $ and any $a\in\mathbb{R}^{3}$ the following identity
holds:%
\[
\left(  \nu\otimes\mu\right)  \left(  \bigcap_{k}\left(  F\left(
x_{k}+a\right)  \right)  ^{-1}\left(  A_{k}\right)  \right)  =\left(
\nu\otimes\mu\right)  \left(  \bigcap_{k}\left(  F\left(  x_{k}\right)
\right)  ^{-1}\left(  A_{k}\right)  \right)\;.
\]

\end{definition}

Moreover, we focus on additive force fields generated by the scatterers in configurations $\omega\in\Omega$.
More precisely:
\begin{definition}
\label{Holt}Assign the potential $\Phi\in C^{2}\left(  \mathbb{R}^{3}\setminus\left\{
0\right\}; \mathbb{R}\right)  $ and let $I=\left\{  Q_{1},Q_{2},...,Q_{L}\right\}  $ be
a finite set with probability measure $\mu.$ An element
$\omega\in\Omega\otimes I$ is then characterized by the sequence
$\omega=\left\{  \left(  x_{n},Q_{j_{n}}\right)  \right\}  _{n\in\mathbb{N}}.$
We say that the random force field $\left\{  F\left(  x\right)
:x\in\mathbb{R}^{3}\right\}  $ is a generalized Holtsmark field
associated to $\Phi$ if there exists an open set $U\subset
\mathbb{R}^{3}$ with $0\in U$ such that, for any $x\in\mathbb{R}^{3}$, the
following identity holds:%
\begin{equation}
F\left(  x\right)  \omega=F_{U}\left(  x\right)  \omega=\lim_{R\rightarrow
\infty}F^{(R)}_U(x) \omega\;,
\label{S1E3}%
\end{equation}
where
\be
F^{(R)}_U(x)\,\omega:=-\sum_{x_{n}\in RU}Q_{j_{n}}\nabla\Phi\left(  x-x_{n}\right)
\label{eq:FRUx}
\ee
and the convergence in \eqref{S1E3} takes place in law. 
Namely:
\[
\lim_{R \to \infty}\left(  \nu\otimes\mu\right)\left(  \left\{  \omega:-\sum_{x_{n}\in RU}Q_{j_{n}}\nabla\Phi\left(
x-x_{n}\right)  \in A\right\}  \right) = \left(  \nu\otimes\mu\right)\Big(  \left\{
\omega:F_{U}\left(  x\right)  \omega\in A\right\}  \Big)
\]
for all $A \in \mathcal{B}\left(  \mathbb{R}^{3}\right)$.

\end{definition}

%Notice that we do not need to specify how the limit is taken if $F\left(
%x\right)  \omega=\infty,$ because due to the definition of random force fields
%this happens with zero probability.

\noindent 
Notice that we do not assume the absolute
convergence of the series $\sum_{x_{n}}Q_{j_{n}}\nabla\Phi\left(
x-x_{n}\right)$. As we shall see in the rest of this section, the limits $F\left(  x\right)  \omega$
might in general depend on the choice of the domains $U$. 
%This issue will be
%discussed in detail later for specific potentials $\Phi$. 
Notice also that the
generalized Holtsmark fields 
are not necessarily invariant under translations in the sense of Definition \ref{InvTrans}.
Finally, we mention that the $C^2-$regularity could be relaxed (it will be used in Section \ref{GenKinEq}
to ensure a well defined dynamics for one particle moving in the Holtsmark field). 

We shall refer to the numbers $Q_{j}$ as scatterer {\em charges}. 
Moreover, for brevity, we shall call {\em scatterer distribution} the joint
distribution $\nu\otimes\mu$ of scatterer configurations and scatterer charges on $\Omega\otimes I$,
and indicate by $\mathbb{E}\,[\cdot]$ the corresponding expectation.
In particular, the condition $\sum_{j=1}^{L}Q_{j}\mu\left(Q_{j}\right)  =0$ models a {\em neutral} system.
In the case of Coulomb, we shall call this condition `electroneutrality'.
Neutral systems display special properties. In fact, we will show that 
when the potential $\Phi$ decays slowly (including the Coulombian case),
the neutrality becomes necessary to prove the existence of a translation invariant $F$. 

Finally, in the case of scatterers with only one charge, neutrality
can be replaced by the assumption of a ``background charge'' with opposite sign:
\begin{definition}
\label{HoltBack}Assign the potential $\Phi\in C^{2}\left(  \mathbb{R}^{3}\setminus\left\{
0\right\}; \mathbb{R}\right)  $ such that $\int_{|y|<1}\left\vert \nabla\Phi\left(  y\right)  \right\vert dy<\infty.$ We
will say that the random force field $\left\{  F\left(  x\right)  :x\in
\mathbb{R}^{3}\right\}  $ is a generalized Holtsmark field with
background associated to $\Phi$ if there exists an open set
$U\subset\mathbb{R}^{3}$ with $0\in U$ such that, for any $x\in\mathbb{R}^{3}$,
\eqref{S1E3} holds (in the sense of convergence in law) with $F^{(R)}_{U}=F^{(R), 0}_{U}$ given by
\begin{equation}
F^{(R), 0}_{U}\left(  x\right)  \omega:=-\left[  \sum_{x_{n}\in RU}\nabla\Phi\left(  x-x_{n}\right)  -\int
_{RU}\nabla\Phi\left(  x-y\right)  dy\right]\;.  \label{S1E4}%
\end{equation}
\end{definition}

\noindent
Note that in this case we assume $L=1$ and $Q_1=1$,
although more general cases could be easily included.
The above negative background can be thought as a form of the so-called `Jeans swindle',
which has been often used in cosmological problems to study the stability properties of gravitational 
systems \cite{BT,Ki03}.

\begin{remark}[Historical comment]
The goal of Holtsmark's paper (cf. \cite{H}) was
to describe the broadening of the spectral lines in gases. This broadening of
spectral lines is induced by the electrical fields acting on the molecules of
the gas (Stark effect).  
On the other hand the electrical field acting on each individual molecule is a
random variable which depends on the specific distribution of the surrounding
gas molecules. The properties of such a random field were computed in \cite{H}
 in the cases in which the fields induced by the gas molecules
can be approximated by either point charges (ions), dipoles or quadrupoles.
Combining the properties of the resulting random electrical fields with the
physical laws describing the Stark effect it is possible to prove that the
broadening of the spectral lines scales like a power law of the gas density.
The exponents characterizing these laws are different for ions, dipoles and
quadrupoles (cf. \cite{H}).
\end{remark}

In the rest of this section, we construct several examples of generalized Holtsmark fields,
and study their basic properties.

\subsection{Construction.} \label{sec:costr}

Let $s > 1/2 $. We will consider potentials $\Phi$ in the following classes
\bea
\calC_s &:=& \Big\{ \;\; \Phi\in C^{2}\left(  \mathbb{R}^{3}\setminus\left\{
0\right\}; \mathbb{R}\right)
\;\;\;\ \Big|\ \;\;\;\Phi\left(  x\right)  =\Phi\left(  \left\vert x\right\vert \right)
\nn\\
&&\;\;\;\;  \mbox{and }\;\; \exists\, A \neq 0, \; r>\max(s,2) \mbox{ \;\;s.t., for $|x|\geq 1$,}\nn\\
&&\;\;\;\;  \left\vert \Phi\left(  x\right)  -\frac{A}{\left\vert x\right\vert^s
}\right\vert +\left\vert x\right\vert \left\vert \nabla\left(  \Phi\left(
x\right)  -\frac{A}{\left\vert x\right\vert^s }\right)  \right\vert  
\leq \frac{C}{\left\vert x\right\vert ^{r}}\;\;  \Big\}\;. \label{eq:defCCs}
\eea
where $C>0$ is a given constant.
The potentials in $\calC_s$ decay as $\vert x\vert ^{-s}$ with an explicit error determined
by $C,r$.
%(it would not be difficult to replace $r>s \vee 2$ by $r>s$ in the results below).
The singularity at the origin, possibly strong, is not relevant in the discussions of this section.
The assumption of radial symmetry might be relaxed. The condition $r>\max(s,2)$
is technically helpful in the proofs that follow and it could be also relaxed, at the price of additional
estimates of the remainder.
%could be replaced by $r>s$ in the proofs that follow.

Let us denote by $m^{(R)}\left(\eta_{1},\cdots,\eta_{J};y_1,\cdots,y_J\right)$ the $J-$point
characteristic function of the random field $F^{(R)}_U(x)$, that is: 
\be
m^{(R)}\left(\eta_{1},\cdots,\eta_{J};y_1,\cdots,y_J\right) :=
\mathbb{E}\left[  \exp\left(  i\sum_{k=1}^{J}\eta_{k}\cdot F^{\left(
R\right)  }_U\left(  y_{k}\right)  \right)  \right]
\label{eq:mpf}
\ee
for all $J \geq 1$,  $\eta_{1},\cdots,\eta_{J}\in \mathbb{R}^{3}$ and $y_1,\cdots,y_J \in \mathbb{R}^{3}$.
Notice that we are dropping from the notation the possible dependence on $U$. We will further abbreviate $m^{(R)}\left(\eta_{1},\cdots,\eta_{J}\right)$
in statements where $y_1,\cdots,y_J$ are fixed.
%Correspondingly, we denote by 
%\be
%m\left(\eta_{1},\cdots,\eta_{J}\right) :=
%\mathbb{E}\left[  \exp\left(  i\sum_{k=1}^{J}\eta_{k}\cdot F_U\left(  y_{k}\right)  \right)  \right]
%\ee
%the multipoint characteristic function of the field $F$ when the limit in \eqref{S1E3} exists.

Suppose that $\Phi \in \calC_s$. Let $\xi\in C^{\infty}\left(  \mathbb{R}^{3}\right)$ be an
arbitrary cutoff function satisfying $\xi\left(  y\right)  =1$ for $\left\vert
y\right\vert \geq1,$ $\xi\left(  y\right)  =0$ for $\left\vert y\right\vert
\leq\frac{1}{2}$ and $\xi\left(  y\right)  =\xi\left(  \left\vert y\right\vert
\right)$ when $|\nabla\Phi|$ is non-integrable at the origin, and $\xi \equiv 1$ otherwise.
Then, we will show that (see the proof of the theorem below)% formally
, as $R \to \infty$, $m^{(R)}$ converges pointwise to
\bea
&& m\left(\eta_{1},\cdots,\eta_{J}\right) := \exp\Big(i  \sum_{k=1}^{J}\eta_{k}\cdot M_U(y_k)\Big)
\label{eq:mpm}\\
&&\ \ \ \ \ \ \ \ \ \ \ \ \ \ \ \ \ \ \ \ \ \ \ \cdot\exp\Big(  \sum_{j=1}^{L}\mu\left(  Q_{j}\right)  \int_{\mathbb{R}^{3}%
}\Big[  \exp\Big(  -iQ_{j}\sum_{k=1}^{J}\eta_{k}\cdot\nabla\Phi\left(
y_{k}-y\right) \Big)  \nn\\
&&\ \ \ \ \ \ \ \ \ \ \ \ \ \ \ \ \ \ \ \ \ \ \ \ \ \ \ \ \ \ \ \ \ \ \ \ \ \  -1 +iQ_{j}\sum_{k=1}^{J}\eta_{k}\cdot\nabla\Phi\left(
y_{k}-y\right)  \xi\left(  y_{k}-y\right)  \Big]  dy\Big) \nonumber
\eea
where $ M_U:  \mathbb{R}^{3}\to \mathbb{R}^3\cup \{\infty\} $ is given by 
\be
M_U(x):= -A \lim_{R \to \infty} \sum_{j=1}^{L}\mu\left(  Q_{j}\right)  Q_{j}%
\int_{RU}\nabla\left(\frac{1}{|x-y|^s}\right)  \xi\left( x-y\right) dy
\label{eq:MUx}
\ee
(and by $0$ in the case of the field \eqref{S1E4}).
%The cutoff $\xi$ has been introduced to control, in the above integrals, non-integrable
%singularities of $\ \nabla\Phi $ in a neighbourhood of the origin.

Formula \eqref{eq:mpm} is a generalization of the formulas for the Holtsmark field for
Coulomb potential in \cite{CH}. As we shall explain below, Eq.\,\eqref{eq:MUx} 
determines the (non)invariance properties of the limiting field as well as its (in)dependence 
in the features of the geometry $U$.

The integral in the last two lines of \eqref{eq:mpm} converges absolutely for $s>1/2$. 
Moreover, it is $o(\eta)$ for small values of the $\eta-$variables.
In particular, if $|M_U(x)|< \infty$, then the limit field $F_U$ in \eqref{S1E3} is well defined and
\be
\mathbb{E}\left[ F_U(x)\right] = M_U(x)\;.
\label{eq:MUxexp}
\ee
However, the limit integral in \eqref{eq:MUx} is only conditionally convergent 
when $s\leq 2$, so that the existence and the properties of $F_U$ depend strongly on $\mu, U, s$ through $M_U(x)$.
 The results
are summarized by the following statement.

\begin{theorem} \label{thm:costruzione}
Suppose that $\Phi \in \calC_s$. Let $U \subset \mathbb{R}^{3}$ be a bounded open set with $0 \in U$
and $\partial U \in C^1$. Then, the right-hand side of \eqref{S1E3} converges in law and defines 
a random field $F$ (in the sense of Definition \eqref{S1E2}) in the following cases:
\begin{itemize}
\item[\emph{\textbf{a.1}}] $s > 2$\;;
\item[\emph{\textbf{a.2}}] $\sum_{j=1}^{L}Q_{j}\mu\left(Q_{j}\right)  =0$ (`neutrality') and $s>1/2$; 
\item[\emph{\textbf{a.3}}] $F^{(R)}_{U}=F^{(R), 0}_{U}$ (given by \eqref{S1E4}),
$\int_{|y|<1}\left\vert \nabla\Phi\left(  y\right)  \right\vert dy<\infty$ and $s > 1/2$\;;
\item[\emph{\textbf{b.}}] $1 < s \leq 2$ and $\int_{U \setminus \{|y|<\frac{1}{2}\}}\nabla\left(  \frac{1}{\left\vert y\right\vert ^{s}}\right) dy=0$\,;
\item[\emph{\textbf{c.}}] $s = 1$ and $\int_{U}\nabla\left(  \frac{1}{\left\vert y\right\vert }\right)  dy=0$\,.
\end{itemize}
Moreover, the following properties hold in the specific cases:
\begin{itemize}
\item[ \emph{\textbf{a.1-2-3}}] 
$F$ is independent on the domain $U$, it is invariant under translations in the sense of Definition
\ref{InvTrans}, and it has characteristic function given by \eqref{eq:mpm} independent of $\xi$ and with $M_U = 0$;

\item[\emph{\textbf{b.}}] $F$ is independent on the domain $U$, it is invariant under translations in the sense of Definition
\ref{InvTrans}, and it has characteristic function given by \eqref{eq:mpm} independent of $\xi$ and with $M_U = 0$;
however, if $RU$ is replaced by $RU - R^{s-1}e$ in \eqref{eq:FRUx} with $e\in\mathbb{R}^3$ and $|e|$ small enough, 
then it can be $M_U(x) \neq 0$;

\item[\emph{\textbf{c.}}]
$F = F_U(x)$ depends on the domain $U$, it is not invariant under translations, and it has characteristic function given by \eqref{eq:mpm} independent of $\xi$ and with 
\be
M_U(x) =  -A\sum_{j=1}^{L}\mu\left(  Q_{j}\right)  Q_{j}
\int_{\partial U}\frac{\left(  x\cdot y\right)  n\left(
y\right)  }{\left\vert y\right\vert ^{3}}dS\left(  y\right)
\label{eq:MUxs1}
\ee
where $n(y)$ is the outer normal to $\partial U$ at $y$.
\end{itemize}
\end{theorem}

Some remarks on this result follow.

\begin{remark} 
Note that $s \leq 2$ implies that $\sum_{x_{n}}Q_{j_{n}}\nabla\Phi\left(x-x_{n}\right)$
can be no more than conditionally convergent. In this case, if the
neutrality condition is not satisfied, we need a rather
stringent assumption on the domain (items {\bf b} and {\bf c}),
which can be thought as a geometrical condition on the cloud scatterer distribution. 
To prove the criticality of such a condition, we show (see Section \ref{eq:depgeom} in the proof)
in case {\bf b} that small displacements of the domain $U$ (i.e.\,slight asymmetries of the cloud) can
yield in the limit random force fields with a non-zero component in one particular direction.
\end{remark}

\begin{remark} 
The case $s=1$ is of course particularly relevant, because it corresponds to gravitational and Coulombian forces.
The major distinguishing feature of this case is that, in absence of neutrality,
the limiting force field is not invariant under translations (item {\bf c.}).
Roughly speaking, a density of charges of order one with
only one sign yields a change in the average force of order one over
distances of the order of the scatterer distance. 
%In particular,
%translation invariance is lost for potentials of the form $\frac{1}{\left\vert
%x\right\vert }$ (unless we add a background with opposite charge or we assume
%electorneutrality) and it will not be possible to derive kinetic equations for
%particles moving such random fields.
Additional features of this case are discussed in Section \ref{sec:diffeq} and in the next remark.
\end{remark}

\begin{remark} 
In the case of potentials $\Phi$ given exactly by a
power law, the random variable $F_{U}\left(  x\right)  $ is given by $M_U(x)+\zeta$ (cf.\,\eqref{eq:MUxexp})
where $\zeta=\left(  \zeta_{1},\zeta_{2},\zeta_{3}\right)$
and each of the random variables $\zeta_{i}$ is a multiple of a symmetric stable
Levy distribution \cite{F}. The characteristic exponent is $3/2$ in the case $s=1$. For this particular
potential and in the electroneutral case, we can take the limit of
the cutoff $\xi$ to one and the 1-point characteristic function $m\left(
\eta\right)  $ becomes:%
\[
m\left(  \eta\right)  =\exp\left(  -C_{0}\left\vert \eta\right\vert ^{\frac
{3}{2}}\sum_{j=1}^{L}\mu\left(  Q_{j}\right)  \left\vert Q_{j}\right\vert
^{\frac{3}{2}}\right)  \ ,\ \eta\in\mathbb{R}^{3}%
\]
where the constant $C_{0}$ takes the
value $2\pi\int_{0}^{\infty}\left[  1-\frac{\sin\left(  x\right)  }{x}\right] x^{-5/2}
dx$ (see \cite{CH} for the computation).
Moreover, using the inversion formula for the Fourier transform we can compute the
probability density for the force field $F.$ In particular, elementary
computations allow to obtain the probability density $\tilde{p}\left(
\left\vert F\right\vert \right)  $ of the absolute value $|F|$. It
turns out that $\tilde{p}\left(  \left\vert F\right\vert \right)  $ is a
bounded function which behaves proportionally to $ |F|^{-5/2}$ as $\left\vert F\right\vert \rightarrow\infty.$
In particular, $\mathbb{E}\left[  \left\vert F\right\vert \right]  <\infty.$
\end{remark}

\begin{remark} \label{rem:s=1/2}
Suppose that $s<\frac{1}{2}$ in items {\bf a.2} and {\bf a.3} and assume by definiteness that $\Phi(x)=\frac{\ep}{|x|^s}$, $s<\frac 1 2$. In this case, we obtain a completely different picture of the
resulting random force fields, because this is due
mostly to the particles placed very far away. Assuming
neutrality (necessary to deal with spatially homogeneous force fields),
one can compute the limit of the fields generated by a
cloud of particles contained in $RU,$ where $U$ is the unit ball. One sees that, in order to
obtain force fields of order one, $\varepsilon$ has to be rescaled with $R$ as
$\varepsilon=R^{1-\frac{1}{2s}}.$ Then, the resulting force field obtained as
$R\rightarrow\infty$ is constant in regions where $\left\vert x\right\vert $
is bounded and it is given by a Gaussian distribution. Since the dynamics
of a tagged particle would greatly differ from the one taking place in the
generalized Holtsmark fields obtained in Theorem \ref{thm:costruzione}, we will not
continue with the study of this case in this paper.
\end{remark}

\begin{remark} \label{rem:2D1}
In two dimensions, the critical value separating absolutely and conditionally convergent
cases is $s=1$ (cf.\,\eqref{eq:MUx}) and the Coulombian case would correspond to
a logarithmic potential. One can introduce $\calC_0$ defined as the class of potentials close to $\frac{A}{ \log x}$
for $|x|$ large (as in \eqref{eq:defCCs}). However, a straightforward adaptation 
of Theorem \ref{thm:costruzione} is valid only for $s>0$, with the following identification of cases: {\bf a.1} $s>1$;
{\bf a.2} `neutrality' and $s>0$; {\bf a.3} `existence of background' and $s>0$; {\bf b.}
$0<s\leq 1$ and $\int_{U \setminus \{|y|<\frac{1}{2}\}}\nabla\left(  \frac{1}{\left\vert y\right\vert ^{s}}\right) dy=0$. 
Note that every power law decay leads to spatially
homogeneous generalized Holtsmark fields. If
$s\leq 1$, a nontrivial condition on the geometry of the finite clouds is required.
\end{remark}

\subsection{Proof of Theorem \ref{thm:costruzione}.}

\subsubsection{Convergence.}

It is a classical computation in probability theory \cite{F}. It is enough to prove 
convergence of the $J-$point characteristic function of the field $F^{(R)}_U(x)$, defined
by \eqref{eq:FRUx} or by \eqref{S1E4}.

Let us assume \eqref{eq:FRUx}.
Fix $y_{1},\cdots,y_{J}\in \mathbb{R}^{3}$, $\eta_{1},\cdots,\eta_{J}\in \mathbb{R}^{3}$. By definition \eqref{eq:mpf} 
\be
m^{(R)} = \mathbb{E}\left[ \exp\left(  i\sum_{k=1}^{J}\eta_{k}\cdot F^{\left(
R\right)  }_U\left(  y_{k}\right)  \right)  \right]    =\mathbb{E}\left[\,  \prod_{k=1}^J\prod_{x_{n}\in RU}%
\exp\left(  -iQ_{j_{n}}\eta_k\cdot\nabla\Phi\left(  y_k-x_{n}\right)  \right)\\
\right] 
\ee
and the properties of
the scatterer distribution $\nu\otimes\mu$ imply
\begin{align}
& m^{(R)} =\sum_{M=0}^{\infty}\frac{e^{-\left\vert RU\right\vert }}{M!}\left(
\sum_{j=1}^{L}\mu\left(  Q_{j}\right)  \int_{RU}\exp\left(  -iQ_{j
}\sum_{k=1}^J\eta_k\cdot\nabla\Phi\left( y_k-y\right)  \right)  dy\right) ^{M}\nn\\
&  =\exp\Big(  \sum_{j=1}^{L}\mu\left(  Q_{j}\right)  \int_{RU}\Big[
\exp\Big(  -iQ_{j}\sum_{k=1}^J\eta_k\cdot\nabla\Phi\left( y_k-y\right)  \Big)\nn\\
& \ \ \ \ \ \ \ \ \ \ \ \ \ \ \ \ \ -1+iQ_{j}\sum_{k=1}^J\eta_k\cdot\nabla\Phi\left( y_k-y\right)  \xi\left(  y_k-y\right)
\Big]  dy\Big)  \nn\\
&  \;\;\;\cdot\exp\Big(  -i\sum_{j=1}^{L}\mu\left(  Q_{j}\right)  Q_{j}\int
_{RU}\sum_{k=1}^J\eta_k\cdot\nabla\Phi\left(  y_k-y\right)  \xi\left(  y_k-y\right)  dy\Big)\;,
\label{eq:mRcfirst}
\end{align}
where $\xi$ has been introduced before \eqref{eq:mpm}. 

Note that the term in the square brackets is of order $(1+\left\vert y\right\vert^{2s+2})^{-1}$,
which is integrable in $\mathbb{R}^{3}$ for $s > 1/2$.

In the neutral case (item {\bf a.2}), the last exponential in \eqref{eq:mRcfirst} is trivially equal to $1$. 
Furthermore for $s>2$ (item {\bf a.1}) all the integrals are absolutely convergent 
and the last exponential in \eqref{eq:mRcfirst} converges to $1$ by the dominated convergence
theorem, because of the symmetry of $\Phi$ and $\xi$.

Otherwise (items {\bf b} and {\bf c}) we write
\bea
&&\int_{RU}\eta_k\cdot\nabla\Phi\left(  y_k-y\right)  \xi\left(  y_k-y\right)
dy\label{S8E3}\\
&&=A\int_{RU}\eta_k\cdot\nabla\left(  \frac{1}{\left\vert y_k-y\right\vert ^{s}%
}\right)  \xi\left(  
y_k-y\right)  dy+\int_{RU}\eta_k\cdot\nabla\rho\left(
y_k-y\right)  \xi\left(  y_k-y\right)  dy \nn
\eea
where $A$ is the constant appearing in \eqref{eq:defCCs}
and $\rho\left(  y\right)  :=\Phi\left(  y\right)  -\frac{A}{\left\vert
y\right\vert ^{s}}.$
Using $\Phi \in \calC_s$ as well as the fact that $r>2$ in \eqref{eq:defCCs}, we
obtain that $\left\vert \nabla\rho\left(  y\right)  \right\vert $ is
integrable in $\mathbb{R}^{3}$ and therefore the symmetry of $\Phi$ and $\xi$ imply 
$\int_{RU}\eta_k\cdot\nabla\rho\left(
y_k-y\right)  \xi\left(  y_k-y\right)  dy \to \int
_{\mathbb{R}^{3}}\eta_k\cdot\nabla\rho\left( y\right)  \xi\left(  y\right)
dy=0$ in the limit $R \to \infty$. 

In the case $F^{(R)}_{U}=F^{(R), 0}_{U}$ given by \eqref{S1E4} (item {\bf a.3}), 
the computation is simpler since we do not need to add/subtract
the last exponential in \eqref{eq:mRcfirst}.

By dominated convergence, we conclude that $$\lim_{R \to \infty}m^{(R)} = m$$ 
holds in all the cases considered with $m$ given by \eqref{eq:mpm} 
and $M_U(x)=0$ in cases {\bf a.1-2-3}, while
we are left with the evaluation of $M_U(x)$ defined by the limit \eqref{eq:MUx} in cases  {\bf b} and {\bf c}.

Let us focus on case {\bf b}. Using the hypothesis on $U$, the symmetry of $\xi$
and a change of variables $y \to Ry$ we find
\bea
&&\lim_{R \to \infty}\Big| \int_{RU}\nabla\left(  \frac{1}{\left\vert x-y\right\vert ^{s}%
}\right)  \xi\left(  x-y\right)  dy\Big|\nn\\
&&=\lim_{R \to \infty} \Big|\int_{RU}\left[  \nabla\left(  \frac{1}{\left\vert x-y\right\vert
^{s}}\right)  \xi\left(  x-y\right)  -\nabla\left(  \frac{1}{\left\vert
y\right\vert ^{s}}\right)  \xi\left(  y\right)  \right]  dy\Big|\nonumber\\
&& = \lim_{R \to \infty} \Big|\int_{RU-x} \nabla\left(  \frac{1}{\left\vert y\right\vert
^{s}}\right)  \xi\left(  y\right) dy -\int_{RU}\nabla\left(  \frac{1}{\left\vert
y\right\vert ^{s}}\right)  \xi\left(  y\right)   dy\Big|\nn\\
&& \leq \lim_{R \to \infty} C \frac{R^2}{R^{s+1}} = 0
\label{S8E2}%
\eea
since $s>1$. Therefore $M_U(x)=0$. Notice that, in the estimate, $C$ is a positive constant dependent on $x$ and we used the regularity of $\partial U$.
Moreover, the cutoff has been used only to treat the non-integrable singularity of the case $s=2$.

Similarly, in case {\bf c} we get
\bea
&&\int_{RU}\nabla\left(  \frac{1}{\left\vert x-y\right\vert %
}\right)  \xi\left(  x-y\right)  dy\nn\\
&&= \int_{RU}\nabla\left(  \frac{1}{\left\vert x-y\right\vert }-\frac
{1}{\left\vert y\right\vert }\right)  dy\nn\\
&& = \int_{R\partial
U}\left(  \frac{1}{\left\vert x-y\right\vert }-\frac{1}{\left\vert
y\right\vert }\right)   n\left(  y\right) 
dS\left(  y\right)
\eea
by using the divergence theorem, where $n\left(  y\right)  $ is the outer normal at $y\in\partial U.$ 
Hence, for any given $x$,
\bea
&&\int_{RU}\nabla\left(  \frac{1}{\left\vert x-y\right\vert %
}\right)  \xi\left(  x-y\right)  dy\nn\\
&& = R\int_{\partial
U}\left(  \frac{1}{\left\vert \frac{x}{R}-y\right\vert }-\frac{1}{\left\vert
y\right\vert }\right)   n\left(  y\right) 
dS\left(  y\right) \nn\\
&& = \int_{\partial
U}\frac{\left(  x\cdot y\right)  n\left(  y\right)
}{\left\vert y\right\vert ^{3}}dS\left(  y\right)  +o\left(
1\right) \quad \text{as}\quad R \to \infty
\eea
where $o\left(1\right)$ is the remainder of the Taylor expansion in $x/R$.
This proves that the limit \eqref{eq:MUx} reduces to \eqref{eq:MUxs1} in this case.

All the properties stated in the theorem follow from the explicit computation
of the characteristic function $m$ performed above.
In particular, we readily check that $$m\left(\eta_{1},\cdots,\eta_{J};y_1,\cdots,y_J\right) = \left(\eta_{1},\cdots,\eta_{J};y_1+a,\cdots,y_J+a\right)$$
for any $a \in \mathbb{R}^3$, from which translation invariance of $F$ follows in the cases {\bf a.1-2-3} and {\bf b}.

\subsubsection{Dependence on the geometry.\label{eq:depgeom}} 

In this section we prove the statement concerning small asymmetries of the cloud of scatterers in the case {\bf b}.
For convenience, we restrict to $1<s<2$ and to $L=1$ and consider the following random field
\be
\bar F^{(R)}_U(x)\,\omega:=-\sum_{x_{n}\in RU-R^{s-1}e}\nabla\Phi\left(  x-x_{n}\right)\;.
\ee
where $e \in \mathbb{R}^{3}$ is small (so that $x \in RU-R^{s-1}e$ for any $R$ large enough).

Notice that the displacement of the domain, which is of order $R^{s-1}$ tends
to infinity, since $s>1,$ but it is smaller than the size of the domain. As will become clear, choosing displacements
for the domains of order one brings no changes on the value of
$F_{U}\left(  x\right)$. 

Arguing exactly as in the previous section we obtain $\lim_{R \to \infty}m^{(R)} = m$
with $m$ given by \eqref{eq:mpm} and the limit in \eqref{eq:MUx} to be determined.
The computation in \eqref{S8E2} is now replaced by
\bea
&&  \int_{\left[  RU-R^{s-1}e\right]  }\nabla\left(  \frac{1}{\left\vert x-y\right\vert ^{s}%
}\right) \xi\left(  x-y\right)  dy\nn\\
&& = -s\int_{\left[  RU-R^{s-1}e-x\right]  }\frac{y}{\left\vert y\right\vert
^{s+2}}\xi\left(  y\right)  dy+s\int_{RU}\frac{y}{\left\vert y\right\vert
^{s+2}}\xi\left(  y\right)  dy \nn\\
&&  =-s\int_{\left[  RU-R^{s-1}e-x\right]  \setminus RU}\frac{y}{\left\vert
y\right\vert ^{s+2}}dy+s\int_{RU\setminus\left[  RU-R^{s-1}e-x\right]  }%
\frac{y}{\left\vert y\right\vert ^{s+2}}dy\nn\\
&&  =-sR^{2-s}\left[  \int_{\left[  U-R^{s-2}e-\frac{x}{R}\right]  \setminus
U}\frac{y}{\left\vert y\right\vert ^{s+2}}dy-\int_{U\setminus\left[
U-R^{s-2}e-\frac{x}{R}\right]  }\frac{y}{\left\vert y\right\vert ^{s+2}%
}dy\right] \;.
\eea
Neglecting the contribution $x/R$ and parametrizing the boundary, we obtain
\be
\int_{\left[  RU-R^{s-1}e\right]  }\nabla\left(  \frac{1}{\left\vert x-y\right\vert ^{s}%
}\right) \xi\left(  x-y\right)  dy = s \int_{\partial U}\frac{(e\cdot y) n\left(  y\right)}{\left\vert y\right\vert ^{s+2}}dS\left(  y\right)+ o(1)
\ee
as $R \to \infty$, where $n\left(  y\right)$ is the outer normal at $y\in\partial U.$ Therefore
\bea
M_U(x)&=& -A \lim_{R \to \infty}
\int_{\left[  RU-R^{s-1}e\right] }\nabla\left(\frac{1}{|x-y|^s}\right)  \xi\left( x-y\right) dy\nn\\
&=& -A s \int_{\partial U}\frac{(e\cdot y) n\left(  y\right)}{\left\vert y\right\vert ^{s+2}}dS\left(  y\right)
\eea
which is a nontrivial vector depending on $e$.

\subsection{Case $\Phi\left(  x\right)  =\frac{1}{\left\vert
x\right\vert }$: differential equations.  \label{sec:diffeq}}

In the particular case in which $\Phi$ is the Coulomb or the Newton potential the
random force field $F_U(x)$ described by Theorem \ref{thm:costruzione} satisfies a system of
(Maxwell) differential equations. In this section we derive such equations (Theorem \ref{NewtonEquation}).
Then, we use them to show that electroneutrality is necessary in order to obtain the translation invariance
(Theorem \eqref{eq:NtiCp}).

\begin{theorem}
\label{NewtonEquation}Suppose that $\Phi\left(  x\right)  =\frac{1}{\left\vert
x\right\vert }$ and let $\left\{  F_{U}\left(  x\right)  :x\in\mathbb{R}%
^{3}\right\}  $ be the corresponding random force field constructed by means
of Theorem \ref{thm:costruzione}.

(i) In the cases {\bf a.2} and {\bf c}, for almost every
$\omega\in\Omega\otimes I$ with the form $\omega=\left\{  \left(
x_{n},Q_{j_{n}}\right)  \right\}  _{n\in\mathbb{N}}$ we have that the function
$\psi\left(  x\right) : =F_{U}\left(  x\right)\omega$ is a weak solution of (i.e.\;it satisfies in the sense of
distributions) the equation:%
\begin{equation}
\operatorname{div}\psi=\sum_{n}Q_{j_{n}}\delta\left(  \cdot-x_{n}\right)
\ \ ,\ \ \operatorname{curl}\psi=0\;.\label{S2E5}%
\end{equation}

(ii) In the case {\bf a.3}, for almost
every $\omega\in\Omega$ with the form $\omega=\left\{  x_{n}\right\}
_{n\in\mathbb{N}}$ we have that the function $\psi\left(
x\right)  :=F_{U}\left(  x\right)\omega$ is a weak solution
of (i.e.\;it satisfies in the sense of distributions) the equation:%
\begin{equation}
\operatorname{div}\psi=\sum_{n}\delta\left(  \cdot-x_{n}\right)
-1\ \ ,\ \ \operatorname{curl}\psi=0\;.\label{S2E6}%
\end{equation}

\end{theorem}

\begin{proof} The proof is very similar for the two cases and we perform the computations for (ii) only.

Let $B_a(y)$ be the ball of radius $a$ centered in $y$. For any positive integer $n$ we write
%In order to prove the convergence with probability one of the series defining $F_{U} $ 
\begin{align*}
F_{n}(  x)\omega    & :=\sum_{\left\{  \left\vert x_{j}\right\vert
\leq2^{n}\right\}  }\frac{x-x_{j}}{\left\vert x-x_{j}\right\vert ^{3}}%
-\int_{B_{2^{n}}\left(  0\right)  }\frac{x-y}{\left\vert x-y\right\vert ^{3}%
}dy\\
& =\sum_{\left\{  \left\vert x_{j}\right\vert \leq1\right\}  }\frac{x-x_{j}%
}{\left\vert x-x_{j}\right\vert ^{3}}-\int_{B_{1}\left(  0\right)  }\frac
{x-y}{\left\vert x-y\right\vert ^{3}}dy\\
& +\sum_{\ell=1}^{n}\left[  \sum_{\left\{  2^{\ell-1}<\left\vert
x_{j}\right\vert \leq2^{\ell}\right\}  }\frac{x-x_{j}}{\left\vert
x-x_{j}\right\vert ^{3}}-\int_{B_{2^{\ell}}\left(  0\right)  \setminus
B_{2^{\ell-1}}\left(  0\right)  }\frac{x-y}{\left\vert x-y\right\vert ^{3}%
}dy\right]  \\
& \equiv\left[  \sum_{\left\{  \left\vert x_{j}\right\vert \leq1\right\}  }%
\frac{x-x_{j}}{\left\vert x-x_{j}\right\vert ^{3}}-\int_{B_{1}\left(
0\right)  }\frac{x-y}{\left\vert x-y\right\vert ^{3}}dy\right]  +\sum_{\ell
=1}^{n}f_{\ell}(  x)\omega\;.
\end{align*}
We want to prove that the quantities $|f_{\ell}(  x)|$ converge to zero as $\ell\rightarrow\infty$ so
quickly as to obtain convergence with probability one as $n \to \infty$ for any fixed $x$. 

Set $\Omega_{\ell}=B_{2^{\ell}}\left(  0\right)  \setminus B_{2^{\ell-1}}(0)$ and $I_{\ell}\left(  x\right)  =\int_{\Omega_{\ell}}\frac
{x-y}{\left\vert x-y\right\vert ^{3}}dy$. Note that $I_{\ell}\left(  0\right) = 0$ and (since the gradient is bounded)
$I_{\ell}\left(  x\right)$ is bounded uniformly in compact sets. By straightforward computation we find that the variance is
\bea
\mathbb{E}\left[  \left(  f_{\ell}\left(  x\right)  \right)
^{2}\right]    & =&\sum_{J=0}^{\infty}\frac{1}{J!}e^{-\left\vert \Omega_{\ell
}\right\vert }\int_{\Omega_{\ell}}dx_{1}\int_{\Omega_{\ell}}dx_{2}%
...\int_{\Omega_{\ell}}dx_{J}\nn\\
&& \cdot\left(  \sum_{\left\{  2^{\ell-1}<\left\vert x_{j}\right\vert \leq2^{\ell
}\right\}  }\frac{x-x_{j}}{\left\vert x-x_{j}\right\vert ^{3}}-I_\ell(x)\right)  
\left( \sum_{\left\{  2^{\ell-1}<\left\vert x_{k}\right\vert
\leq2^{\ell}\right\}  }\frac{x-x_{k}}{\left\vert x-x_{k}\right\vert ^{3}}%
-I_\ell(x)\right)\nn\\
&= & \left(  I_{\ell}\left(  x\right)  \right)  ^{2}\sum_{J=0}^{\infty}%
\frac{\left\vert \Omega_{\ell}\right\vert ^{J-2}}{J!}e^{-\left\vert
\Omega_{\ell}\right\vert }\left[  \left(  J\left(  J-1\right)  -2J\left\vert
\Omega_{\ell}\right\vert +\left\vert \Omega_{\ell}\right\vert ^{2}\right)
\right]  \nn\\
&& +\left[  \sum_{J=1}^{\infty}\frac{1}{\left(  J-1\right)  !}e^{-\left\vert
\Omega_{\ell}\right\vert }\left\vert \Omega_{\ell}\right\vert ^{J-1}\right]
\int_{\Omega_{\ell}}\frac{1}{\left\vert x-y\right\vert ^{4}}dy\nn\\
&=&
\int_{\Omega_{\ell}}\frac{1}{\left\vert x-y\right\vert ^{4}}dy\;,
\eea
which is of order $2^{-\ell}$.

If $A_{\ell}:=\left\{  \omega:\left\vert f_{\ell}(x)
\right\vert \geq\frac{1}{\ell^{2}}\right\}$, then by Chebyshev's inequality
$\left\vert A_{\ell}\right\vert \leq C\ell^{4}2^{-\ell}$ for some $C>0$, which is summable over $\ell$.
Borel-Cantelli implies that, with probability one, there are at most a
finite number of values of $\ell$ for which $\omega\in A_{\ell}.$ Thus, the series $\sum_{\ell
\geq 1}f_{\ell}(  x)\omega$ converges absolutely for any given $x$, with probability one.

In the same way, one estimates the gradients
\begin{align*}
%F_{n}\left(  x;\omega\right)    & =\left[  \sum_{\left\{  \left\vert x_{j}\right\vert \leq1\right\}  }%
%\frac{x-x_{j}}{\left\vert x-x_{j}\right\vert ^{3}}-\int_{B_{1}\left(
%0\right)  }\frac{x-y}{\left\vert x-y\right\vert ^{3}}dy\right]  +\sum_{\ell
%=1}^{n}f_{\ell}(  x)\omega  \;,\\
\nabla F_{n}(  x)\omega    & =\nabla\left[  \sum_{\left\{
\left\vert x_{j}\right\vert \leq1\right\}  }\frac{x-x_{j}}{\left\vert
x-x_{j}\right\vert ^{3}}-\int_{B_{1}\left(  0\right)  }\frac{x-y}{\left\vert
x-y\right\vert ^{3}}dy\right]  +\sum_{\ell
=1}^{n}\nabla f_{\ell}(  x)\omega\;.
\end{align*}
We conclude that, with probability 1, $F_{n}(x)$ and $\nabla F_n(x)$ converge uniformly over compact sets after removal of a finite number of singularities. 
We can then pass to the limit in the weak formulation of the equation
\begin{equation}
\operatorname{div}F_{n}=\sum_{\left\{  \left\vert x_{j}\right\vert \leq
2^{n}\right\}  }\delta\left(  \cdot-x_{j}\right)
-1\ \ ,\ \ \operatorname{curl}F_{n}=0\;.
\end{equation}
\end{proof}

\begin{theorem} \label{eq:NtiCp}
Suppose that $\left\{  F\left(  x\right)  :x\in\mathbb{R}^{3}\right\}  $ is a
random force field in the sense of Definition \ref{RandForcField} with
$I=\left\{  Q_{1},Q_{2},...,Q_{L}\right\}  $ and such that for almost every
$\omega\in\Omega\otimes I$ with the form $\omega=\left\{  \left(
x_{n},Q_{j_{n}}\right)  \right\}  _{n\in\mathbb{N}}$ the function $\psi\left(
x\right)  :=F\left(  x\right) \omega $ satisfies (\ref{S2E5}) in the weak formulation.
Suppose that the random force field $\left\{  F\left(  x\right)
:x\in\mathbb{R}^{3}\right\}  $ is invariant under translations and that the
average $\mathbb{E}\left[  \left\vert F\left(  x\right)\right\vert\right]  
$ is finite for any point $x\in\mathbb{R}^{3}$. Then $\sum_{j=1}^{L}Q_{j}%
\mu\left(  Q_{j}\right)  =0.$
\end{theorem}

\begin{proof}
Let $\varphi\in C_{0}^{\infty}\left(  \mathbb{R}^{3}\right)  $ be a test
function. Let $\overline\psi := \mathbb{E}\left[  \psi\right]$. 
Then, by averaging \eqref{S2E5} with respect to $\nu\otimes\mu$ (cf.\,Definition \ref{RandForcField})
we have
\[
-\int_{\mathbb{R}^{3}}\nabla\varphi\cdot \overline\psi\left(  x\right)  dx=\left(
\sum_{j=1}^{L}\mu\left(  Q_{j}\right)  Q_{j}\right)  \int_{\mathbb{R}^{3}%
}\varphi,\ \ \ \ \ \ \operatorname{curl}\overline\psi=0
\]
where we have used that $\mathbb{E}\left[  \sum_{n}\varphi\left(
x_{n}\right)  \right]  =\int_{\mathbb{R}^{3}}\varphi.$ 
Then, there exists $\phi$ such that $\overline\psi=\nabla\phi$
and $\phi$ is a weak solution of %
\begin{equation}
\Delta\phi=\sum_{j=1}^{L}\mu\left(  Q_{j}\right)  Q_{j}\ \ \text{in\ \ }%
\mathbb{R}^{3}\label{T1E9}\;.%
\end{equation}
On the other hand, since the random field is invariant under
traslations, $\overline\psi$ must be constant. Taking $\phi$ as a
linear function, the left-hand side of (\ref{T1E9}) vanishes and
the theorem follows. 
\end{proof}

\begin{remark}
The Holtsmark field for Newtonian potentials has been studied in connection to astrophysics. Several statistical properties of these random forces can be found in \cite{CH, CH1, CH2, CH3, CH4}.
\end{remark}

%
%
%
%
%*********************************************************
%*********************************************************
%
%*********************************************************
%
%*********************************************************
%
%*********************************************************
%
%
%

\section{Conditions on the potentials yielding as limit equations for $f$
Boltzmann and Landau equations. \label{GenKinEq}}

In the rest of this paper we discuss the dynamics of a tagged particle in some
families of generalized Holtsmark fields as those constructed in Section \ref{sec:costr}. 
We shall consider families of potentials of the form%
\begin{equation} \label{Pot}
\left\{  \Phi\left(  x,\varepsilon\right)  ;\ \varepsilon>0\right\}  
\end{equation}
where $\varepsilon$ is a small parameter tuning the mean free path $\ell_{\varepsilon}$ (cf.\,Introduction).
The latter is defined as the typical length that the
tagged particle must travel in order to have a change in its velocity
comparable to the absolute value of the velocity itself. 
We recall that, in our units, the average distance $d$ between the scatterers is normalized to one,
and the characteristic speed of the tagged particle is of order one. 

We will be interested in the dynamics of the tagged particle in the so called kinetic
limit. One of the assumptions that we need in order to derive such a limit
is %
\begin{equation}
1=d\ll \ell_{\varepsilon} \ \ \text{as\ \ }\varepsilon\rightarrow 0.
\label{KinLim}%
\end{equation}
A second condition is the statistical independence of the
particle deflections experienced over distances of the
order of $\ell_\varepsilon$. This condition will be discussed in more detail
in Section  \ref{GenLandCase}.

As argued in the Introduction, assumption (\ref{KinLim}) can be obtained 
in two different ways. A first possibility is that the deflections are small except
at rare collisions over distances of order $\lambda_{\varepsilon}\ll d.$ If such rare 
deflections are the main cause for the change
of velocity of the tagged particle, we will obtain that the dynamics is
given by a linear Boltzmann equation. A second possibility is that the potentials in (\ref{Pot}) 
are very weak, but the interaction with many scatterers of the background
yields eventually a change of the velocity of order one when the
particle moves over distances $\ell_{\varepsilon}\gg d.$
The force acting over the tagged particle at any given time is a random variable depending on the
(random) scatterer configuration, leading to a diffusive process in the space of velocities.
The dynamics of the tagged particle is then described by a linear Landau 
equation (if the deflections are uncorrelated in a time
scale of order $\ell_{\varepsilon}$).

We make now more precise the concept of collision length (sometimes also
termed `Landau length' in the literature), namely the characteristic distance for which the deflections
experienced by the tagged particle are of order one. 
\begin{definition}
\label{LandLength}We will say that a family of radially symmetric potentials
\eqref{Pot} 
has a well defined collision length $\lambda_{\varepsilon}$ if there exists a positive
function $\left\{  \lambda_{\varepsilon}\right\}  $ such that $\lambda
_{\varepsilon}\rightarrow0$ as $\varepsilon\rightarrow0$ and 
\[
\lim_{\varepsilon\rightarrow0}\Phi\left(  \lambda_{\varepsilon}y,\varepsilon
\right)  =\Psi\left(  y\right)  =\Psi\left(  \left\vert y\right\vert \right)
\ \text{uniformly in compact sets of }y\in\mathbb{R}^{3}\;,%
\]
where $\Psi\in C^{2}\left(  \mathbb{R}^{3}\setminus\left\{  0\right\}
\right)  $ is not identically zero and satisfies%
\[
\lim_{\left\vert y\right\vert \rightarrow\infty}\Psi\left(  y\right)  =0.
\]

In this case, the characteristic time between collisions
(Boltzmann-Grad time scale) is defined by%
\begin{equation}
T_{BG}=\frac{1}{\lambda_{\varepsilon}^{2}}\;. \label{BG}%
\end{equation}

\end{definition}
For instance, families of potentials behaving as in \eqref{SingP}
have a collision length $\lambda_\varepsilon=\varepsilon$. 
On the contrary a family of potentials like
$\Phi\left(  x,\varepsilon\right)  =\varepsilon G\left(  x\right)  ,$
where $G$ is globally bounded, do not have a collision length.

Notice that $T_{BG}\rightarrow\infty$ as $\varepsilon
\rightarrow0.$ In the kinetic regime (\ref{KinLim}),
Boltzmann terms can appear only if the family of potentials in (\ref{Pot}) admits
a collision length. If a family of potentials does not have a collision length we
will set $T_{BG}=\infty,\ \lambda_{\varepsilon}=0.$ 

%\begin{remark}
Later on we will further assume that the potential $\Psi$ yields a well
defined scattering problem between the tagged particle and one single scatterer, in the precise sense
discussed in Section~\ref{ScattPb}.
%\end{remark}

Next, we recall the class of potentials (\ref{Pot}) for which we assume (\ref{KinLim}). 
We will restrict to radially symmetric functions $\Phi$ which are either globally smooth, or singular
at the origin. Moreover, we will be interested in random force fields which are defined in the whole 
space and are spatially homogeneous. As explained in
Section \ref{Holtsm} this requires to assume that there are different types of
charges and a neutrality condition holds, or that a background charge is present, depending on
the long range decay of the potential.
More precisely, the above assumptions are satisfied by 
the generalized Holtsmark fields as constructed in Theorem \ref{thm:costruzione},
items {\bf a.1-2-3} and {\bf b}, by assuming
\be
\Phi(\cdot,\varepsilon) \in \calC_s
\label{eq:PeCs}
\ee
for some $s > 1/2$ (cf.\,\eqref{eq:defCCs}). Clearly the constant $A=A(\varepsilon)$ in \eqref{eq:defCCs}
depends now on $\varepsilon$.

Let $F\left(  x,\varepsilon\right)$ be such a generalized Holtsmark field.
Let $(x(t),v(t))$ be the position and velocity of the
tagged particle moving in the field. For each given scatterer configuration
$\omega \in \Omega \times I$ with the form $\omega=\left\{  \left(
x_{n},Q_{j_{n}}\right)  \right\}  _{n\in\mathbb{N}}$, the evolution
is given by the ODE:%
\begin{equation}
\frac{dx}{dt}=v\ \ ,\ \ \ \frac{dv}{dt}=F\left(  x,\varepsilon\right) \omega  \label{S3E2}%
\end{equation}
with initial data
\begin{equation}
x\left(  0\right)  =x_{0},\ \ v\left(  0\right)  =v_{0} \label{S3E3}%
\end{equation}
for some $x_{0}\in\mathbb{R}^{3},\ v_{0}\in\mathbb{R}^{3}$.
Since the vector fields $F\left(  x,\varepsilon\right)\omega$ are singular at the points $\left\{  x_{n}\right\}$, given $\left(
x_{0},v_{0}\right)  $ we do not have global well posedness of solutions for
all $\omega\in{\Omega}$. 
However, with \eqref{eq:PeCs} we assume to have global existence with
probability one, i.e.\,the fields are locally 
Lipschitz away from the points $\left\{x_{n}\right\}$ and
the tagged particle does not collide with any of the scatterers.

Let us denote by $T^{t}\left(  x_{0},v_{0};\varepsilon;\omega\right)
$ the hamiltonian flow associated to the equations (\ref{S3E2})-(\ref{S3E3}).
By assumption this flow is defined for all $t\in\mathbb{R}$ and a.e.\;$\omega$. 
Suppose that \thinspace$f_{0}\in\mathcal{M}_{+}\left(
\mathbb{R}^{3}\times\mathbb{R}^{3}\right)  ,$ where $\mathcal{M}_{+}$ denotes
the set of nonnegative 
Radon measures. Our goal is to study the asymptotics of
the following quantity as $\varepsilon$ tends to zero:%
\begin{equation}
f_{\ep}\left( \ell_\ep t,\ell_\ep x,v \right)  =\mathbb{E}[f_{0}(T^{-\ell_\ep t}\left(
\ell_\ep x,v;\varepsilon;\cdot\right)  )] \label{eq:exp}%
\end{equation}
where the expectation is taken with respect to the scatterer distribution.
 
In order to check if it is possible to have a kinetic regime described
by a Landau equation, we must examine the contribution to the deflections of
the tagged particle due to the action of the potentials $\Phi\left(
x,\varepsilon\right)  $ at distances larger than the collision length
$\lambda_{\varepsilon}.$ To this end we split $\Phi\left(  x,\varepsilon
\right)  $ as follows. We introduce a cutoff $\eta\in C^{\infty}\left(
\mathbb{R}^{3}\right)  $ such that $\eta\left(  x\right)  =\eta\left(
\left\vert x\right\vert \right)  ,$ $0\leq\eta\leq1,$ $\eta\left(  x\right)
=1$ if $\left\vert x\right\vert \leq1,$ $\eta\left(  x\right)  =0$ if
$\left\vert x\right\vert \geq2.$ We then write%
\begin{equation}
\Phi\left(  x,\varepsilon\right)  =\Phi_{B}\left(  x,\varepsilon\right)
+\Phi_{L}\left(  x,\varepsilon\right)  \ \label{S4E4}%
\end{equation}
with%
\begin{equation}
\Phi_{B}\left(  x,\varepsilon\right)  :=\Phi\left(  x,\varepsilon\right)
\eta\left(  \frac{\left\vert x\right\vert }{M\lambda_{\varepsilon}}\right)
\ \ ,\ \ \Phi_{L}\left(  x,\varepsilon\right)  :=\Phi\left(  x,\varepsilon
\right)  \left[  1-\eta\left(  \frac{\left\vert x\right\vert }{M\lambda
_{\varepsilon}}\right)  \right]\;.  \ \label{S4E5}%
\end{equation}
Here $M>0$ is a large real number which eventually will be sent to
infinity at the end of the argument.  
If the family of potentials does not have a
collision length we just set $\Phi_{L}\left(  x,\varepsilon\right)
=\Phi\left(  x,\varepsilon\right)  .$ In the above definitions
$B$ stands for `Boltzmann' and $L$ for `Landau'. Indeed the potential $\Phi_{B}$ yields the
big deflections experienced by the tagged particle within distances of order
$\lambda_{\varepsilon}$ of one single scatterer and $\Phi_{L}$ accounts for the
deflections induced by the scatterers which remain at distances much larger
than $\lambda_{\varepsilon}$ from the particle trajectory.
Note that, if the potentials $\Phi\left(  x,\varepsilon\right)  $
satisfy the above explained conditions allowing to define spatially homogeneous generalized Holtsmark fields,
then $\Phi_{B}\left(  x,\varepsilon\right)  $ and $\Phi_{L}\left(
x,\varepsilon\right)  $ satisfy similar conditions and we can define random
force fields $\left\{  F_{B}\left(  x,\varepsilon\right)  :x\in\mathbb{R}%
^{3}\right\}  ,\ \left\{  F_{L}\left(  x,\varepsilon\right)  :x\in
\mathbb{R}^{3}\right\}  $ associated respectively to $\Phi
_{B}\left(  x,\varepsilon\right)  $ and $\Phi_{L}\left(  x,\varepsilon\right).$ 

To understand the deflections produced by $\Phi_{L}$ we have to study
the ODE
\begin{equation}
\frac{dx}{dt}=v\ \ ,\ \ \ \frac{dv}{dt}=F_L\left(  x,\varepsilon\right) \omega  \label{S9E1}%
\end{equation}
for $0 \leq t \leq T$, with initial data $x\left(  0\right)  =x_{0}, v\left(  0\right)  =v_{0}$. 
Due to the invariance under translations, we can assume 
$x_0  =0,$ $v_0  =0.$ The time scale $T$ is chosen sufficiently small to guarantee that the deflection
experienced by the tagged particle in the time interval $t\in\left[
0,T\right]  $ is much smaller than $\left\vert v_{0}\right\vert .$ Then, it is
reasonable to use the approximation%
\[
x\left(  t\right)  \simeq v_{0}t\ \ ,\ \ t\in\left[  0,T\right]
\ \ \text{as\ \ }\varepsilon\rightarrow0
\]
whence%
\[
\frac{dv}{dt}\simeq F_{L}\left(  v_{0}t,\varepsilon\right)\omega
\ \ \text{for }t\in\left[  0,T\right]  \ \ \text{as\ \ }\varepsilon
\rightarrow 0
\]
and the change of velocity in $[0,T]$ can be approximated as $\varepsilon\rightarrow0$ by the random
variable%
\begin{equation}
D_{T}\left(  \varepsilon\right)\omega  :=\int_{0}^{T}F_{L}\left(  v_{0}t,\varepsilon\right)\omega\,  dt\; .\label{S9E2}%
\end{equation}
As in Section \ref{sec:costr}, we may study these random variables by computing
the characteristic function:
\be
m_{T}^{(\ep)}\left(  \theta \right)  =\mathbb{E}\left[  \exp\left(
i\theta \cdot D_{T}\left(  \varepsilon\right)\omega \right)  \right]
\ \ ,\ \ \theta\in\mathbb{R}^{3}.  \label{S9E3}
\ee
The following result is a corollary of Theorem \ref{thm:costruzione}.
\begin{corollary}
\label{CharFunctions}
Suppose that $\Phi(\cdot,\ep)\in \calC_s$. Then, we can define spatially homogeneous random
force fields $ F_{L}(\cdot,\ep)$ associated to $\Phi_L(\cdot,\ep)$, by means of Theorem \ref{thm:costruzione}
(items {\bf a.1-2-3} and {\bf b}).
The characteristic function \eqref{S9E3}
is given by
\bea
m_{T}^{(\ep)}\left(  \theta \right)  &=& \exp\Big(  \sum_{j=1}^{L}\mu\left(  Q_{j}\right)  \int_{\mathbb{R}^{3}%
}\Big[  \exp\Big(  -iQ_{j}\theta\cdot\int_0^T dt\, \nabla_x\Phi_L\left(
v_0 t-y,\ep\right) \Big) \label{eq:CFmeTt}  \\
&&\ \ \ \ \ \ \ \ \ \ \ \ \ \ \ \ \ \ \ \ \ \ \ \ \  -1 +iQ_{j}\theta\cdot\int_0^T dt \,\nabla_x\Phi_L\left(
v_0 t-y,\ep\right) \Big]  dy\Big)\;. \nn
\eea
\end{corollary}

We focus now on the magnitude of the deflections due to $\Phi_{L}.$ 
We assume that
$\left\vert \theta\right\vert $ is of order one.
We are interested in time scales $T=T_{\varepsilon}$ for
which $D_{T}\left(  \varepsilon\right)\omega$ is small, which means
$|\int_0^T dt\, \nabla_x\Phi_L\left(
v_0 t-y,\ep\right)| \ll1$ as
$\varepsilon\rightarrow0$ for the range of values of $y \in\mathbb{R}^{3}$
contributing to the integrals in \eqref{eq:CFmeTt}.  
We can then approximate the characteristic function as:
\begin{equation}
m_{T}^{(\ep)}\left(  \theta \right) =\exp\left(  -\frac{1}{2}\sum_{j=1}%
^{L}\mu\left(  Q_{j}\right)   Q_{j}^{2}\int_{\mathbb{R}^{3}%
}\left( \theta\cdot\int_0^T dt \,\nabla_x\Phi_L\left(
v_0 t-y,\ep\right)\right)  ^{2} dy \right)\;.  \label{S4E9}%
\end{equation}
This formula suggests the following way of defining a characteristic time for the deflections.
Setting
\begin{equation}
\sigma\left(  T;\varepsilon\right)  :=\sup_{\left\vert \theta\right\vert
=1}\int_{\mathbb{R}^{3}}dy\left(  \theta\cdot\int_{0}^{T}\nabla_{x}\Phi
_{L}\left(  vt-y,\varepsilon\right)  dt\right)  ^{2}\;, \label{S4E8a}%
\end{equation}
we define the Landau time scale $T_{L}$ as the solution of the equation%
\begin{equation}
\sigma\left(  T_{L};\varepsilon\right)  =1\;. \label{S4E8}%
\end{equation}
Notice that $T_{L}$ is a function of $\varepsilon$ and that we can assume,
without loss of generality, that $|v|=1$ (we will do so in the following). If there is no solution
of (\ref{S4E8}) for small $\varepsilon$ we set $T_{L}=\infty.$

Using the time scales $T_{BG},\ T_{L}$, we reformulate condition \eqref{KinLim} as
\begin{equation}
\ell_\ep = \min\left\{  T_{BG},T_{L}\right\}  \gg1\ \ \text{as\ \ }\varepsilon
\rightarrow 0\label{S9E4}%
\end{equation}
and we deduce whether the evolution of $f := \lim_{\ep \to 0} f_\ep$ is described by means of a Landau or a Boltzmann equation.
In fact the relevant time scale to describe the evolution of $f$ is the shortest among $T_{BG},\ T_{L},$
and the condition \eqref{S9E4} can take place in different ways:
\begin{align}
T_{L}  &  \gg T_{BG}\ \ \text{as\ \ }\varepsilon\rightarrow0\label{S9E5a}\\
T_{L}  &  \ll T_{BG}\ \ \text{as\ \ }\varepsilon\rightarrow0\label{S9E5b}\\
\frac{T_{L}}{T_{BG}}  &  \rightarrow C_{\ast}\in\left(  0,\infty\right)
\ \text{as\ \ }\varepsilon\rightarrow0\;. \label{S9E5c}%
\end{align}
If (\ref{S9E5a}) holds the dynamics of $f$ will be
described by a linear Boltzmann equation.
If (\ref{S9E5b}) takes place we
would obtain that the small deflections produced in the trajectories of the
tagged particle due to the part $\Phi_{L}$ of the potential modify $f\left(
t,x,v\right)  $ faster than the binary encounters with scatterers yielding
deflections of order one. In this case, if in addition the deflections of the tagged particle are uncorrelated in time scales of order $T_L$,
the evolution will be given by a suitable linear Landau equation.
Finally, if
(\ref{S9E5c}) takes place then both processes, binary collisions and collective
small deflections of the tagged particle, are relevant in the evolution of $f$,
and we can have combinations of the above equations.

A technical point must be addressed here. Due to the presence of the
cutoff $M$ in (\ref{S4E4})-(\ref{S4E5}) some care is needed concerning the
precise meaning of (\ref{S4E8}). Indeed, $\Phi_{L}$ yields also contributions due to binary collisions within
distances of order $M\lambda_{\varepsilon}$ from the scatterers. This 
implies that, if (\ref{S9E5a}) holds, we have
$\sigma\left(  T_{BG};\varepsilon\right)  \simeq\delta\left(  M\right)  >0.$
The natural way of giving a precise meaning to the condition (\ref{S9E5a})
will be then the following. The dynamics of $f$ will be
described by the linear Boltzmann equation if we have%
\begin{equation}
\lim\sup_{\varepsilon\rightarrow0}\sigma\left(  T_{BG};\varepsilon\right)
\leq\delta\left(  M\right)  \ \ \text{with\ }\lim_{M\rightarrow\infty}%
\delta\left(  M\right)  =0\;, \label{S9E6}%
\end{equation}
that is, the small deflections due to interactions between the
tagged particle and the scatterers at distances larger than $M\lambda
_{\varepsilon}$ become irrelevant as $M\rightarrow\infty$ in the time scale
$T_{BG}$\;. 

In the rest of this section, we discuss the specific form of the kinetic equations obtained in the
different cases.

\subsection{Kinetic equations describing the evolution of the distribution
function $f:$ the Boltzmann case.}

In this section we describe the evolution of the function $f_{\ep}$ defined in (\ref{eq:exp}) as
$\varepsilon\rightarrow0,$ assuming \eqref{S9E4} and (\ref{S9E6}) (i.e.\,\eqref{S9E5a}). Before doing that, 
we briefly review the associated two-body problem. The following discussion
is classical. For further details we refer to \cite{LL1}.

\subsubsection{Scattering problem in $\Phi_{B}\;.$\label{ScattPb}}

We consider the mechanical problem of the deflection of a
single particle of mass $m=1$ and initial ($t \to -\infty$) velocity $v \neq 0$ moving in a field $\Phi_{B}$, whose centre
(scatterer source) is
at rest. Due to Definition \ref{LandLength}, it is natural to use here $\lambda_{\varepsilon}$ as unit
of length. We define $y=\frac{x}{\lambda_{\varepsilon}}$ and
focus on the scattering problem associated to the potential $\Psi_{B}\left(
y\right) : =\Psi\left(  y\right)  \eta\left(  \frac{y}{M}\right)  .$ We write $V = |v|$ and $r = |y|$.

\begin{figure}[th]
\centering
\includegraphics [scale=0.4]{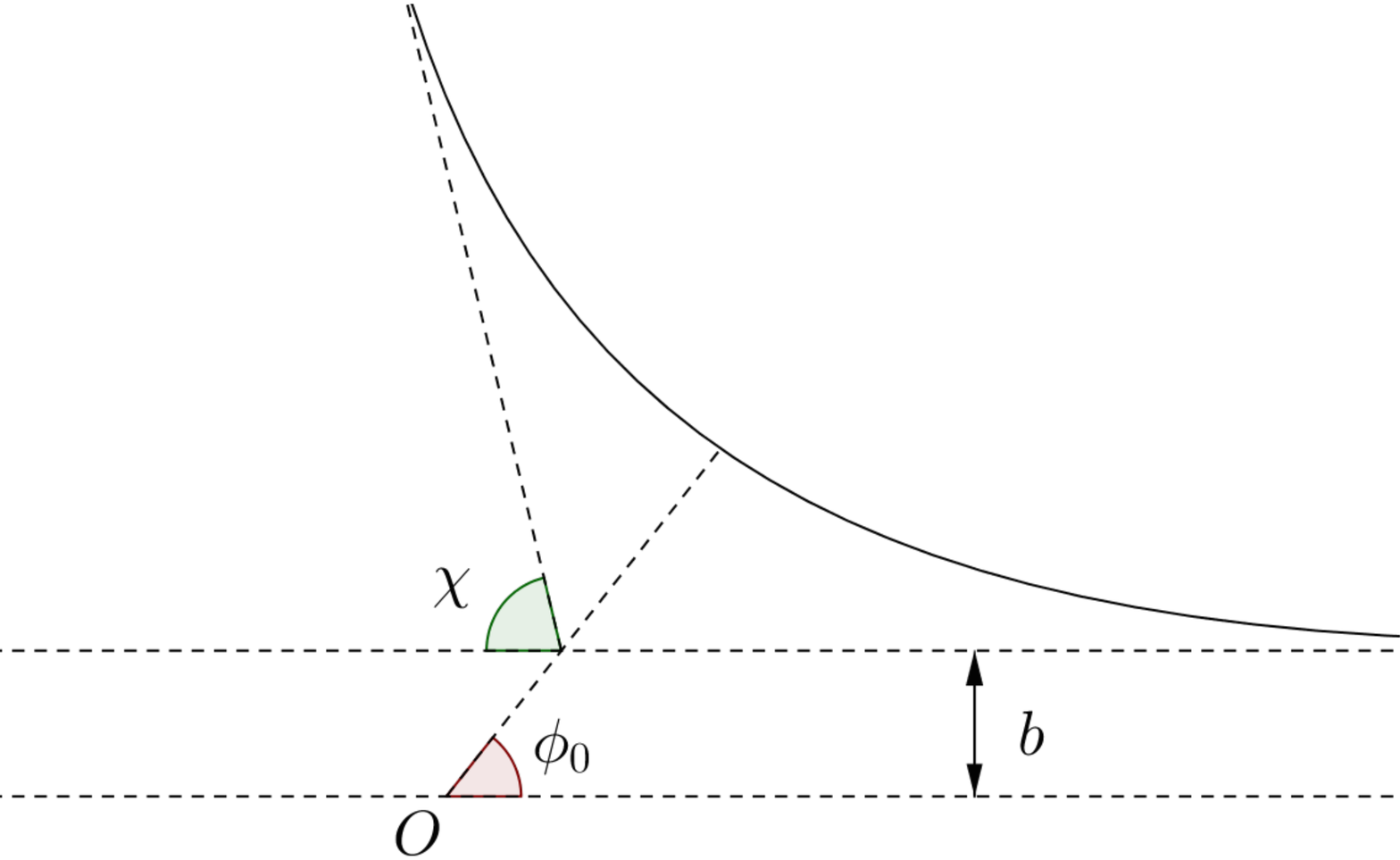}
\caption{
The two-body scattering.
The solution of the two-body problem lies in a plane, which is taken to be the plane
of the page. The scatterer lies in the origin. The scalar $b$ is the impact parameter and  $\chi=\chi(b, V)$ is the scattering angle.   
\label{fig:scattering}
}
\end{figure}

\ni 
The path of the particle in the central field is
symmetrical about a line from the centre to the nearest point in the orbit,
hence the two asymptotes to the orbit make equal angles $\phi_{0}$ with this line (see e.g.\,Figure \ref{fig:scattering}).
The angle of scattering is seen from Fig.\,\ref{fig:scattering} to be%
\begin{equation}
\chi\left(  b,V\right)  =\pi-2\phi_{0}\;.\label{DA}%
\end{equation}

We will say that the scattering problem is well defined for a given value of
$V$ and $b$ if the solution of the equations\
\begin{equation}
\frac{dy}{dt}=v\ ,\ \ \frac{dv}{dt}=-\nabla\Psi_{B}\left(  y\right)
\label{S9E1a}%
\end{equation}
satisfies:
\begin{equation}
\lim_{t\rightarrow-\infty}\left\vert y\left(  t\right)  \right\vert
=\lim_{t\rightarrow\infty}\left\vert y\left(  t\right)  \right\vert
=\infty,\ \lim_{t\rightarrow-\infty}\left\vert v\left(  t\right)  \right\vert
=\lim_{t\rightarrow\infty}\left\vert v\left(  t\right)  \right\vert =V.
\label{P1E6}%
\end{equation}

The effective potential reads%
\begin{equation}
\Psi_{eff}\left(  r\right)   =\Psi_{B}(r)+\frac{b^{2}V^2}{2r^{2}} \label{P1E6a}%
\end{equation}
where $r=\left\vert y\right\vert .$ A sufficient condition for the scattering
problem to be well defined for a given value of $V$ and almost all the values
of $b$ is that the set of nontrivial solutions $r$ of the simultaneous equations%
\be
\frac{\Psi_{B}{}^{\prime}(r)}{V^{2}}=\frac{b^{2}}{r^{3}}\ ,\ \ \frac{\Psi
_{B}(r)}{V^{2}}+\frac{b^{2}}{2r^{2}}=\frac{1}{2}%
\label{eq:sing2bP}
\ee
is nonempty only for a finite set of values $b>0.$ We will assume that
this condition is satisfied for all the families of potentials considered. 
A standard analysis of Newton equations shows that the scattering angle is given by
\begin{equation}
\chi\left(  b,V\right)  =\pi-2\phi_{0}=\pi-2\int_{r_{\ast}}^{+\infty
}\frac{b\,dr}{r^{2}\sqrt{1-2\Psi_{eff}\left(  r\right)/V^2 }} \label{eq:SE_0a}%
\end{equation}
where $r_{\ast}$ is the nearest approach of the particle to the
scatterer (defined as the largest solution to the second equation in \eqref{eq:sing2bP}).

Using spherical coordinates with north pole $\frac{v}{\left\vert v\right\vert }$ and
azimuth angle $\varphi$ characterizing the plane of scattering, 
we define a mapping
\begin{equation}
\left( b, \varphi\right)  \rightarrow\omega=\omega\left(  b,\varphi;v\right)
\in S^{2}, \label{map}%
\end{equation}
where $\omega$ is the unit vector in the direction of the velocity of the
particle as $t \to +\infty$. 
Let $\Sigma\left(  v\right)  \subset S^{2}$ be the image
of this mapping. Due to the symmetry of the potential $\Psi$ the
set $\Sigma\left(  v\right)  $ is invariant under rotations around $\frac
{v}{\left\vert v\right\vert }.$

We do not need to assume that the mapping \eqref{map} is injective
in the variable $b$.
In particular, a point $\omega\in\Sigma\left(
v\right)  $ can have different preimages. We will consider only potentials
$\Psi$ for which the number of these preimages is finite. We can then define a
family of inverse functions%
\[
\omega\rightarrow\left(  b_{j}\left(  \omega\right)  ,\varphi_{j}\left(
\omega\right)  \right)  \ \ ,\ \ j\in J\left(  \omega\right)
\]
where $J\left(  \omega\right)  $ is a set of indexes which characterizes the
number of preimages of $\omega$ for each $\omega\in\Sigma\left(  v\right).$  
We classify the points of $\Sigma\left(  v\right)  $ by means of the
number of preimages, i.e.\,we write
$
\bigcup_{k=1}^{\infty}A_{k}\left(  v\right)  =\Sigma\left(  v\right)
$
with
\begin{equation}
A_{k}\left(  v\right)  :=\left\{  \omega\in\Sigma\left(  v\right)  :\#J\left(
\omega\right)  =k\right\}  \ ,\ \ k=1,2,3,... \label{P3E8}%
\end{equation}

Let $\chi_{A_{k}\left(
v\right)  }\left(  \omega\right)$ be the characteristic function of the set $A_k(v)$.
We define the differential cross-section of the scattering problem as $\frac{1}{|v|}B$
where 
\begin{align}
B\left(  v;
\omega\right)   &  =\sum_{k=1}^{\infty}B_{k}\left( v;\omega\right)  \chi_{A_{k}\left(
v\right)  }\left(  \omega\right)  \ \ \ \ ,\ \ \ \ \omega\in\Sigma\left(  v\right)\;,
\label{P1E8}\\
\frac{1}{|v|}B_{k}\left(  v;\omega\right)   &  =\sum_{j\in J\left(  \omega\right)  }%
\frac{b_{j}}{|\sin\chi|}\Big|\left(  \frac{\partial\chi\left(  b_{j},V\right)  }{\partial b}\right)
^{-1}\Big|\ \ \ \ \text{for }\omega\in A_{k}\left(  v\right)\;.
\label{P1E8a}
\end{align}
Since the dynamics defined by the equations (\ref{S9E1a}) is invariant under time
reversal $\left(  y,v,t\right)  \rightarrow\left(  y,-v,-t\right)$, the following detailed balance condition holds:%
\begin{equation}
B_{k}\left( v; \omega\right)  =B_{k}\left(\left\vert v\right\vert \omega;
 \frac{v}{\left\vert v\right\vert
}\right)\;. \label{P1E7}%
\end{equation}

With a slight abuse, we will use in what follows the same notation for cutoffed and uncutoffed ($M \to \infty$) potential.

\subsubsection{Kinetic equations.} \label{ss:GenBoltzCase}

We focus first on the case of Holtsmark fields with one single charge ($L=1$, $Q_1=1$).

\begin{claim}
\label{BoltGen}Suppose that (\ref{S9E6}) holds. Let us assume that $\Psi$ is
as in Definition \ref{LandLength}. Suppose that the following limit exists:%
\begin{equation}
f_{\ep}\left(  T_{BG}t,T_{BG}x,v\right)  \rightarrow f\left(
t,x,v\right)  \text{ as }\varepsilon\rightarrow0 \;. \label{P1E5a}%
\end{equation}
Then $f$
solves the linear Boltzmann equation%
\begin{equation}
\left(  \partial_{t}f+v\partial_{x}f\right)  \left(  t,x,v\right)
=\int_{S^{2}}B\left(  v;\omega \right)  \left[  f\left(  t,x,\left\vert
v\right\vert \omega\right)  -f\left(  t,x,v\right)  \right]  d\omega\label{P1E5}%
\end{equation}
where $B$ is as in (\ref{P1E8}).
\end{claim}

\begin{proofof}[Justification of the Claim \ref{BoltGen}] 
We introduce a time scale $t_{\ast}$ satisfying $1\ll t_{\ast}\ll T_{BG}.$ We
define the domain $D_{\varepsilon}\left(  vt_{\ast}\right)  \subset
\mathbb{R}^{3}$ as the set swept out by the sphere of radius $M\lambda
_{\varepsilon}$ initially centered at the tagged particle, moving in the
direction of $v$ during the time $t_{\ast}$ (cf.\,\eqref{S4E5}).
The motion of the tagged particle is rectilinear between collisions and is
affected by the interaction $\Phi_B$ if the
domain $D_{\varepsilon}\left(  vt_{\ast}\right)  $ contains one or more
scatterers. Notice that the volume of $D_{\varepsilon}\left(  vt_{\ast}\right)
$ satisfies $\left\vert D_{\varepsilon}\left(  vt_{\ast}\right)  \right\vert
\simeq \pi V M^2\lambda_\varepsilon^{2}t_{\ast}\ll1$ where $V = |v|$. 
Using the properties of the Poisson distribution it follows that the probability of finding one scatterer in
$D_{\varepsilon}\left(  vt_{\ast}\right)  $ is approximately $\left\vert
D_{\varepsilon}\left(  vt_{\ast}\right)  \right\vert $ and the probability of
finding two or more scatterers is proportional to $\left\vert D_{\varepsilon
}\left(  vt_{\ast}\right)  \right\vert ^{2}$ which can be neglected. We
introduce a system of spherical coordinates having $\frac{v}{\left\vert
v\right\vert }$ as north pole and we denote by $\varphi$ the azimuth angle
(as in Section \ref{ScattPb}). Assuming that there
is one scatterer in the domain $D_{\varepsilon}\left(  vt_{\ast}\right)  $,
the conditional probability that the scatterer has azimuth angle in the interval $\left[  \varphi,\varphi
+d\varphi\right]  $ and the impact parameter of the collision is in
the interval $\left[  \bar{b},\bar{b}+d\bar{b}\right]  $  can
be approximated by $\frac{Vt_*\bar{b}d\bar{b}d\varphi}{\left\vert
D_{\varepsilon}\left(  vt_{\ast}\right)  \right\vert }$. 
We obtain deflections in the velocity $v$ of order one if $\bar{b}$ is of
order $\lambda_{\varepsilon},$ therefore it is natural to introduce the change of
variables $\bar{b}=\lambda_{\varepsilon}b.$ We conclude that the probability of a
collision in a time interval of length $t_{\ast}$ with rescaled impact parameter $b$
and azimuth angle $\varphi$ is:
\begin{equation}
\left(  \lambda_{\varepsilon}\right)  ^{2}Vt_{\ast}bdbd\varphi. \label{P1E1}%
\end{equation}

In order to derive the evolution equation for the function $f\left( t,x,v\right)  $ it is convenient to compute the limit behaviour of (cf.\,(\ref{eq:exp}))
\begin{equation}
\psi_{\ep}\left(  t,x,v \right)  :=\mathbb{E}\left[  \psi_{0,\ep}\left(
T^{t}\left(  x,v;\varepsilon;\cdot\right)  \right)  \right]  \label{P3E1}%
\end{equation}
where $\psi_{0,\ep}=\psi_{0,\ep}\left(  x,v\right)  $ is a smooth test function. We
have the following duality formula%
\begin{equation}
\int\int f_{\ep}\left(  t,x,v\right)  \psi_{0,\ep}\left(  x,v\right)
dxdv=\int\int f_{0}(x,v)\psi_{\ep}\left(  t,x,v\right)  dxdv\ \ ,\ \ t>0
\label{P3E2}%
\end{equation}
which follows using (\ref{eq:exp}), the change of variables $T^{-t_{0}}\left(
x,v;\varepsilon;\cdot\right)  =\left(  y,w\right),$ and (\ref{P3E1}). We compute the differential equation
satisfied by the function $\psi\left( t,x,v\right)
=\lim_{\varepsilon\rightarrow0}\psi_{\ep}\left(  T_{BG}t,T_{BG}x,{v}\right)  $ with $\psi_{0}\left( x,v \right)
=\lim_{\varepsilon\rightarrow0}\psi_{0,\ep}\left(  T_{BG}x,v\right)  .$ 
%The computation below shows that the existence of this limit is plausible. 
Supposing that $h>0$ is small but such that $hT_{BG}\gg1$ and using the semigroup property of $T^{t}$ we obtain%
\begin{align}
\psi_{\ep}\left( T_{BG}( t+h),T_{BG} x,v\right)  &=\mathbb{E}\left[  \psi
_{0,\ep}\left(  T^{T_{BG}(t+h)}\left(  T_{BG} x,v;\varepsilon;\cdot\right)
\right)  \right]  \nonumber \\&
=\mathbb{E}\left[  \psi_{0,\ep}\left(  T^{T_{BG}h}%
T^{T_{BG}t}\left( T_{BG} x,v;\varepsilon;\cdot\right)  \right)  \right].
\end{align}
We assume now that the deflection of the
particle during $\left[ T_{BG} t,T_{BG}(t+h)\right]  $ is
independent from its previous evolution in $\left[  0,T_{BG}t\right]$
(in particular, recollisions of the particle with the scatterers happen with small probability). 
To prove this independence would be a crucial step of any rigorous proof of
the Claim \ref{BoltGen} (notice that this implies also the Markovianity of the
limit process). Then 
\begin{equation}
\psi_{\ep}\left(  T_{BG}(t+h),T_{BG}x,v \right)  \simeq \mathbb{E}\left[
\psi_{\ep}\left( T_{BG} t,T^{T_{BG}h}\left(  T_{BG}x,v;\varepsilon;\cdot\right)
\right)  \right] \label{P3E3}%
\end{equation}
for small $\ep>0$.
If the position of the particle is $\left( T_{BG} x,v\right)  $
at the time $T_{BG}t,$ its new position at the time $T_{BG}(t+h)$ is
$T_{BG}x+vT_{BG}h.$ Recalling the expression (\ref{P1E1}) for the probability of a collision
with given impact parameter and azimuth angle, we deduce
\begin{align*}
& \psi_{\ep}\left(  T_{BG}(t+h),T_{BG}x,v \right)    \\
&  \simeq \psi_{\ep}\left( T_{BG}t ,T_{BG}x+vT_{BG}h,v\right)  \left[  1-\left(
\lambda_{\varepsilon}\right)  ^{2}VhT_{BG}\int_{0}^{2\pi}d\varphi\int
_{0}^{M}bdb\right]  +\\
&  +\left(  \lambda_{\varepsilon}\right)  ^{2}VT_{BG}h\int_{0}^{2\pi}%
d\varphi\int_{0}^{M}bdb\, \psi_{\ep}\left(T_{BG} t ,T_{BG} x,\left\vert
v\right\vert \omega\left(  b,\varphi;v\right)  \right)
\end{align*}
where $\omega\left(  b,\varphi;v\right)  $ is
as in (\ref{map}). Here we neglected the probability of having more than one collision  in the time interval $[T_{BG} t , T_{BG}( t +h)]$, since $h$ is sufficiently small. Using $T_{BG}\lambda_\ep^{2}=1$ and a Taylor expansion in $h$, 
we obtain, in the limit $h\rightarrow0$,
\begin{equation}
\frac{\partial\psi\left(  t,x,v\right)  }{\partial\tau}%
=v\frac{\partial\psi\left(  t,x,v\right)  }{\partial x%
}+V\int_{0}^{2\pi}d\varphi\int_{0}^{M}bdb\left[
\psi\left(   t,x,\left\vert v \right\vert \omega\left(
b,\varphi;v\right)  \right)  -\psi\left(   t,x,v\right)
\right]  . \label{P3E4}%
\end{equation}

Finally, we pass to the limit in (\ref{P3E2}).
We set $f(t,x,v) =\lim_{\varepsilon
\rightarrow0}f_\ep\left(  T_{BG}t,T_{BG}x,v\right)  $, $f_{0}(x,v)=f\left(  0,x,v\right) $
and $\bar{\psi}\left(  t,x,v \right)  =\psi\left(  t_{0}-t
,x,v\right)$ for $t_0>0$. We get
\[
\int\int f\left( t_{0},x,v\right)  \bar{\psi}\left( t
_{0},x,v\right)  dxdv=\int\int f_{0}(x,v
)\bar{\psi}_{0}\left(x,v\right)  dxdv%
\ \ \ ,\ t_{0}>0
\]
which implies
\begin{equation}
\partial_{t}\left(  \int\int f\left(t,x,v\right)
\bar{\psi}\left(  t,x,v\right)  dxdv\right)  =0\;.
\label{P3E6}%
\end{equation}
Note that, by (\ref{P3E4}), for $0<t<t_0$ we have%
\begin{align}
\frac{\partial\bar{\psi}\left( t,x,v \right)  }{\partial t}
&  =-v\frac{\partial\bar{\psi}\left(   t,x,v\right)  }%
{\partial x}-V\int_{0}^{2\pi}d\varphi\int_{0}^{M}bdb\left[  \bar{\psi
}\left(  t,x,\left\vert v\right\vert \omega\left(
b,\varphi;v\right)  \right)  -\bar{\psi}\left(  t,x,v\right)  \right]\;,\nonumber\\
\bar{\psi}\left(  t_{0}, x,v\right)   &  =\psi_{0}\left(x,v\right) . \label{P3E7}%
\end{align}
Using
(\ref{P3E6}) and (\ref{P3E7}) and integrating by parts in the term containing
$\partial_{x}$ we obtain%
\begin{align}
0  &  =\int\int\bar{\psi}\left(  t,x,v\right)  \partial_{t}f\left(
t,x,v\right)  dxdv+\int\int\bar{\psi}\left(  t,x,v\right)  v\partial
_{x}f\left(  t,x,v\right)  dxdv-\nonumber\\
&  -\int\int dxdv\int_{0}^{2\pi}d\varphi\left\vert v\right\vert \int_{0}%
^{M}bdb\left[  \bar{\psi}\left(  t,x,\left\vert v\right\vert \omega\left(
b,\varphi;v\right)  \right)  -\bar{\psi}\left( t,x,v\right)  \right]
f\left(  t,x,v\right).   \label{P1E3}%
\end{align}
Performing the change of variables in (\ref{map}) (cf.\,(\ref{P1E8a})) and taking then the limit $M\rightarrow\infty$ we can write
the last integral term in (\ref{P1E3}) as%
\be
-\int\int dxdv\, \bar{\psi}\left(  t,x,v\right)
\sum_{k=1}^{\infty}\int_{A_{k}\left(  v\right)  }\left[  B_{k}\left(\left\vert v\right\vert \omega;
 \frac{v}{\left\vert v\right\vert
}\right)  f\left(
t,x,\left\vert v\right\vert \omega\right)  -B_{k}\left(v;\omega\right)
f\left(  t,x,v\right)  \right]  d\omega\;. \\
\ee
From (\ref{P1E8}), (\ref{P1E7}) and the arbitrariness of $\bar\psi$, we get \eqref{P1E5}.
\end{proofof}

\begin{remark}
The above way of obtaining the Boltzmann equation is reminiscent of the cutoff procedure used in \cite{DP} to derive the Boltzmann equation
rigorously for potentials of the form $|x|^{-s}$ for $s>2$ in two space dimensions.
\end{remark}

\begin{remark}
The condition (\ref{S9E6}) enters in the argument because we assume that the
trajectories of the particles between collisions are rectilinear. This is due
to the fact that the time $T_{L}$ required to produce deflections between collisions is much larger than the scale $T_{BG}$.
\end{remark}

\begin{remark}
If the Holtsmark field in which the particle evolves has different
types of charges we must replace (\ref{P1E5}) by the equation
\begin{equation}
\left(  \partial_{t}f+v\partial_{x}f\right)  \left(  t,x,v\right)
=\sum_{j=1}^{L}\mu\left(  Q_{j}\right)\int_{S^{2}}B\left(  v;\omega;Q_j \right)  \left[  f\left(  t,x,\left\vert
v\right\vert \omega\right)  -f\left(  t,x,v\right)  \right]  d\omega\nn
\end{equation}
where $B\left(  v;\omega;Q_j \right)$ is the scattering kernel obtained computing the
deflections for each type of charge. Notice that the form of $\Psi_{eff}$\ in
(\ref{P1E6a}) yields the following functional dependence for the differential cross-section
$\Sigma = B/|v|$: %
\[
\Sigma\left(  v;\omega;Q_j \right)=\Sigma\left(  \frac{v}{\sqrt{\left\vert Q_{j}\right\vert }};\omega;
\operatorname*{sgn}\left(  Q_{j}\right)\right)
\]
i.e.\,we can reduce the computation of the scattering kernel to just two values
of the charge $\pm1$ and arbitrary particle velocities. Notice that there is
no reason to expect $B$ to take the same value for positive and negative
charges and a given value of the velocity, although it turns out that this
happens in the particular case of Coulomb potentials.
\end{remark}

\subsection{Kinetic equations describing the evolution of the distribution
function $f:$ the Landau case. \label{GenLandCase}}

In this section we consider the evolution of the function $f_{\ep}$ defined in (\ref{eq:exp}) as
$\varepsilon\rightarrow0,$ assuming \eqref{S9E4} and (\ref{S9E5b}).
The latter condition is not sufficient to
obtain a Landau equation, since we need also to have independent deflections on
time scales of order $T_{L}.$ 
Under the conditions yielding
the Landau equation the deflections in times of order $T_{L}$ are gaussian
variables and the independence condition reads (cf.\,\eqref{S9E2})
\begin{equation}
\mathbb{E}\left(  D(0) \, D(\tilde{T}_{L})\right) 
\ll\sqrt{\mathbb{E}\left(   D(0)^2  \right)
\mathbb{E}\left(  D(\tilde{T}_{L})^{2}\right)
}\ \ \text{as\ \ }\varepsilon\rightarrow0 \label{I1E1}%
\end{equation}
where $\tilde{T}_{L}$ is some time scale much smaller than $T_{L}$, and we 
denoted $D\left(  0\right)  $ and $D(  \tilde{T}_{L})  $ the deflections experienced during the time intervals $\left[0,\tilde{T}_{L}\right]  ,\ \left[  \tilde{T}_{L},2\tilde{T}_{L}\right] $ respectively. 
Furthermore, in order to have a well defined probability distribution
for the deflections we need the convergence as $\ep\to 0$ of the
characteristic function $m_{T}^{(\ep)}\left(  \theta \right) $ in (\ref{S4E9}). 
More precisely, restricting for simplicity to the case of one single charge
and assuming $|v|=1$, we have
\begin{equation}
\frac{1}{2}\int_{\mathbb{R}^{3}%
}\left( \theta\cdot\int_0^{\zeta T_L} dt \,\nabla_x\Phi_L\left(
v \zeta T_L -y,\ep\right)\right)  ^{2} dy \rightarrow\kappa\,\zeta\left\vert \theta_{\bot}\right\vert ^{2}%
\ \text{as\ }\varepsilon\rightarrow0 \label{I1E2}%
\end{equation}
for every $\zeta>0$ and for some constant $\kappa>0,$ where $\theta_{\bot}%
=\theta-\frac{\theta\cdot v}{\left\vert v\right\vert }\frac{v}{\left\vert
v\right\vert }.$ In particular,
\begin{equation}
m_{\zeta T_L}^{(\ep)}\left(  \theta \right) \rightarrow\exp\left(
-\kappa\zeta\left\vert \theta_{\bot}\right\vert ^{2}\right)  \ \ \text{as\ }%
\varepsilon\rightarrow 0. \label{I1E3}%
\end{equation}

We will not try to formulate the most general set of conditions under which
(\ref{I1E2}) takes place. However, we can expect this formula to be a
consequence of the smallness of the deflections, the independence condition
(\ref{I1E1}) and the x limit theorem. We will show in
Section \ref{Examples} examples of families of potentials for which the
left-hand side of (\ref{I1E2}) converges to a different
quadratic form. For those families of potentials (\ref{I1E1}) also fails.
Moreover, the following simple argument sheds some light on the relation between
(\ref{I1E1}) and (\ref{I1E2}). Suppose that the deflections of the tagged particle
in the time intervals $\left[  0,\zeta_{1}T_{L}\right]  $ and $\left[  \zeta
_{1}T_{L},\left(  \zeta_{1}+\zeta_{2}\right)  T_{L}\right] $, denoted by
$D_1$ and $D_2$, are independent (at least asymptotically as $\varepsilon
\rightarrow0$). Then the characteristic function
$m_{(\zeta_1+\zeta_2) T_L}^{(\ep)}\left(  \theta \right)$ of the total deflection $D=D_{1}+D_{2}$
is close to a product $m_{\zeta_1 T_L}^{(\ep)}\left(  \theta \right)  m_{\zeta_2 T_L}^{(\ep)}\left(  \theta \right)$.
This is possible only if the function on the right-hand side of
(\ref{I1E2}) is linear in $\zeta$ (cf.\,\eqref{I1E3}).

In the following we assume both (\ref{I1E1}) and (\ref{I1E2}) and we derive a linear Landau equation. 

\begin{claim}
\label{ClaimLand}Assume that (\ref{S9E5b}) holds and suppose that
(\ref{I1E1}), (\ref{I1E2}) are also satisfied. Suppose that the
following limit exists:%
\begin{equation}
f_{\ep}\left(  T_{L}t,T_{L}x,v\right)  \rightarrow f\left(
t,x,v\right)  \text{ as }\varepsilon\rightarrow 0. \label{LimEx}%
\end{equation}
Then%
\begin{equation}
\left(  \partial_{t}f+v\partial_{x}f\right)  \left(  t,x,v\right)  =\kappa\Delta_{v_{\perp}}f\left(  t,x,v \right)  \label{GenLanEq} %v_{\parallel},v_{\perp}\right)  \label{GenLanEq}%
\end{equation}
where $\Delta_{v_{\perp}}$ is the Laplace Beltrami operator on $S^2$ (sphere of radius $|v|=1$) and the diffusion coefficient $\kappa>0$ is defined by \eqref{I1E2}. 
\end{claim}

\begin{remark}
Using Cartesian coordinates, the diffusion term in \eqref{GenLanEq} reads as
\begin{equation}
\sum_{i,j=1}^3 \frac{\partial}{\partial v_i} A_{i,j}(v)  \frac{\partial}{\partial v_j} f\left(t,x,v\right)
\end{equation}
where the diffusion matrix $ A_{i,j}(v)$ is given by
\begin{equation}
 A_{i,j}(v)= \kappa \left(\delta_{ij}-\frac{v_iv_j}{|v|^2}\right)\;. %\frac{\kappa}{2}|v| \left(\delta_{ij}-\frac{v_iv_j}{|v|^2}\right) 
\end{equation}
We refer to \cite{LL2} and \cite{S} for further details.
\end{remark}

\begin{proofof}
[Justification of the Claim \ref{ClaimLand}] 
Using (\ref{I1E3}) and the Fourier inversion
formula, we can compute the probability density for the transition from
$\left(  T_{L}x,v\right)  $ to $\left(  T_{L}y,v+D\right)  $ in a time interval of
length $T_{L}h $ with $D\in\mathbb{R}^{3}:$ 
\begin{align*}
p\left( T_{L}y,v+D;T_{L} x,v;T_{L}h\right)   &  =\frac{\delta\left(  T_{L}y-T_{L}x-v
T_{L}h\right) T_{L}^3 }{\left(  2\pi\right)  ^{3}}\int_{\mathbb{R}^{3}}\exp\left(
-\kappa h \left\vert \theta_{\bot}\right\vert ^{2}\right)  \exp\left(
iD\cdot\theta\right)  d\theta\\
&  =\frac{\delta\left(  T_{L}y-T_{L}x-v T_{L}h\right) \delta\left(  D_{\parallel
}\right) T_{L}^3 }{4\kappa\pi h}\exp\left(  -\frac{\left\vert D_{\perp}\right\vert
^{2}}{\kappa h}\right)\;.
\end{align*}
Here we write $\theta=\left(  \theta_{\parallel},\theta_{\bot}\right)
$ with %
\be
\theta_{\parallel}=\theta\cdot\frac{v}{\left\vert v\right\vert },\ \ \ \ 
\theta_{\bot}=\theta-\left(  \theta\cdot\frac{v}{\left\vert v\right\vert
}\right)  \frac{v}{\left\vert v\right\vert } \label{eq:notTPTP}
\ee
and use a similar
decomposition for $D=\left(  D_{\parallel},D_{\perp}\right)$ and other
vectors appearing later. That is, the probability density yielding the transition
from $\left(T_{L}  x,v\right)  $ to $\left( T_{L} y,w\right)  $ is%
\begin{equation}
p\left( T_{L} y,w;T_{L} x,v;T_{L} h\right)  =\frac{\delta\left(  x-y-v h \right)
\delta\left(  v_{\parallel}-w_{\parallel}\right)  }{4\kappa\pi h}%
\exp\left(  -\frac{\left\vert v_{\perp}-w_{\perp}\right\vert ^{2}}{\kappa
h}\right) \equiv G(y,w; x,v;h) \;. 
\label{I1E4}%
\end{equation}
Using the independence assumption, we obtain the following approximation for $h$ small%
\[
f_{\ep}\left(  T_{L}(t+h) ,T_{L} x,v\right)   \simeq\int_{\mathbb{R}^{3}}%
dy\int_{\mathbb{R}^{3}}dw f_{\ep}\left( T_{L} t,T_{L} y,w;\varepsilon\right)  p\left(
T_{L}y,w;T_{L}x,v;T_{L}h\right)
\]
whence, using (\ref{LimEx}), 
\[
f\left(  t+h,x,v\right)  =\int_{\mathbb{R}^{3}}dy 
\int_{\mathbb{R}^{3}}dw f\left(  t,y,w\right)  G\left( y,w;x,v;h\right)
\]
and by (\ref{I1E4})
\begin{align*}
f\left(  t+h,x,v\right)   &  = 
%\int_{\mathbb{R}^{3}}dy%
%\int_{\mathbb{R}^{3}}dw f\left( t,y,w\right)  \frac{\delta\left(
%x-y-vh\right)  \delta\left(  v_{\parallel}-w_{\parallel
%}\right)  }{4\kappa\pi\zeta}\exp\left(  -\frac{\left\vert v_{\perp}-w_{\perp
%}\right\vert ^{2}}{\kappa\zeta}\right) \\&  =
\frac{1}{4\kappa\pi h}\int_{\mathbb{R}^{2}}dw_{\bot}f\left(  t
,x-vh,v_{\parallel},w_{\bot}\right)  \exp\left(  -\frac{\left\vert
v_{\perp}-w_{\perp}\right\vert ^{2}}{\kappa h}\right)\;.
\end{align*}
Approximating $f\left(  t,y,w\right)  $ by means of its Taylor
expansion up to second order in $w_{\bot}=v_{\bot}$ and to first order in
${y}={x}$ as well as $f\left(  t+h,x,v\right)  $ by its
first order Taylor expansion in $h=0$, we obtain (\ref{GenLanEq}).
\end{proofof}

\subsection{The case of  deflections with correlations of order one.} \label{ss:CorrCase}

If \eqref{S9E4} and $T_{L}\ll T_{BG}$ hold but the condition (\ref{I1E1}) fails, then
the dynamics of the distribution function $f_{\varepsilon}$ cannot be
approximated by means of a Landau equation. We shall not consider this case in
detail in this paper. However it is interesting to formulate the type of mathematical problem
describing the dynamics of the tagged particle. We discuss such formulation in the present
section. 

We denote the deflection experienced by the tagged particle at the point
$x,$ with initial velocity $v$ during a small (eventually infinitesimal) time
$h$ as $D\left(  x,v;h\right) .$ We use here macroscopic variables for space and 
time. As $\ep\to 0$, the characteristic function of $D$ approaches
the exponential of a quadratic function and the
deflections become gaussian variables with zero average. For these variables, the form of the
correlation might be strongly dependent on the family of
potentials considered, but some general features might be expected.

First of all, due to the invariance under translation of the Holtsmark field,
the correlation functions will take the form
\begin{equation}
\mathbb{E}\left[  D\left(  x_{1},v_{1};h\right)  \otimes D\left(
x_{2},v_{2};h\right)  \right]  =K\left(  x_{1}-x_{2},v_{1},v_{2}%
;h\right) \neq 0.  \label{ST1}%
\end{equation}

Furthermore, we will obtain (cf.\,examples in Section \ref{Examples})
\begin{equation}
\int_{0}^{1} K\left(  y(s),v_{1},v_{2};h\right)
ds<\infty\label{ST2}%
\end{equation}
for any curve $y(s)$ of class $C^1$. This integrability condition might be expected if (\ref{I1E1}) fails, because
otherwise one could have large deflections at small distances and $T_L$ would not coincide with the scale of the
mean free path (cf.\,\eqref{S9E4}). 

Finally, the equation yielding the evolution of the tagged particle can be written
as
\begin{equation}
x\left(  \tau+d\tau\right)  -x\left(  \tau\right)  =v(\tau)d\tau\ \ ,\ \ v\left(
\tau+d\tau\right)  -v\left(  \tau\right)  =D\left(  x\left(  \tau\right)
,v\left(  \tau\right)  ;d\tau\right)\;,\label{ST3}%
\end{equation}
where the order of
magnitude of $D$ is not necessarily $d\tau$, but it might
be of order $(d\tau)^{\alpha'}$ for some $0<\alpha' < 1$ (see e.g.\,Section \ref{subsec:KE1}, $\frac{1}{2}<s<1$).

A typical example which can be derived for a family of power law potentials is the following
(cf.\,Theorem \ref{CorrPowLaw}-{\em (ii)}):
\begin{equation}
K\left(  y,v_{1},v_{2};d\tau\right)  =\frac{1}{\left\vert y\right\vert
^{\alpha}}\Lambda\left(  \frac{y}{\left\vert y\right\vert} ,v_{1},v_{2}\right)  \left(  d\tau\right)  ^{2}\ \ ,\ y\neq0\ ,\ \ \ 0<\alpha
<1\ \ ,\ \ \label{ST6}%
\end{equation}%
\begin{equation}
K\left(  0,v_{1},v_{2};d\tau\right)  =\Lambda\left(  v_{1},v_{2}\right)
\left(  d\tau\right)  ^{2-\alpha}\;,\label{ST7}
\end{equation}
where $\Lambda$ is a matrix valued function. Note that, likely, due to the integrability of the factor $\frac{1}{\left\vert y\right\vert^{\alpha}}$ in \eqref{ST6}, 
the condition \eqref{ST7} does not play a relevant role. 

It would be interesting to clarify
if (\ref{ST1})-(\ref{ST7}) forms a well defined mathematical problem which can be
solved for a suitable choice of initial values $x\left(  0\right)
=x_{0},\ v\left(  0\right)  =v_{0}.$\ Note that this is not a standard
stochastic differential equation, but rather
a stochastic differential equation with correlated noise. 
The dynamics \eqref{ST3} has some analogies with fractional Brownian motion \cite{MvN68}.

\section{Examples of kinetic equations derived for different families of
potentials.\label{Examples}}

We now apply the formalism of Section \ref{GenKinEq} to different families
of potentials. First we check if the families of potentials considered have an associated collision length,
then we examine the behaviour of the functions $\sigma\left(  T;\varepsilon
\right)  $ in (\ref{S4E8a}). We compute the time scales
$T_{BG},\ T_{L}$ defined in (\ref{BG}), (\ref{S4E8}) and check if (\ref{S9E4})
and some of the conditions (\ref{S9E5a})-(\ref{S9E6}) and \eqref{I1E1}-\eqref{I1E2} hold. Finally, 
we write the corresponding kinetic equations.

\subsection{Kinetic time scales for potentials with the form $\Phi\left(  x,\varepsilon\right)
=\Psi\left(  \frac{\left\vert x\right\vert }{\varepsilon}\right)
.$\label{Psiovereps}}

We first consider the family of potentials
\begin{equation}
\left\{  \Phi\left(  x,\varepsilon\right)  ;\ \varepsilon>0\right\}  =\left\{
\Psi\left(  \frac{\left\vert x\right\vert }{\varepsilon}\right)
;\ \varepsilon>0\right\}  \label{S9E7}%
\end{equation}
where $\Phi(\cdot,\varepsilon) \in \calC_s$, $s > 1/2$. We have
$\Psi\in C^{2}\left(  \mathbb{R}^{3}\setminus\left\{  0\right\}
\right)  $ and 
\begin{equation}
\Psi\left(  y\right)  \sim\frac{A}{\left\vert y\right\vert ^{s}}%
\ \ \text{as\ \ }\left\vert y\right\vert \rightarrow\infty\ \ ,\ \ \ \nabla
\Psi\left(  y\right)  \sim-\frac{sAy}{\left\vert y\right\vert ^{s+2}%
}\ \ \text{as\ \ }\left\vert y\right\vert \rightarrow\infty\label{PsiAs}%
\end{equation}
for some real number $A\neq0$ (cf.\,\eqref{eq:defCCs}).  
By Definition \ref{LandLength},
the collision length is $\lambda_{\varepsilon}=\varepsilon$ and
\begin{equation}
T_{BG}=\frac{1}{\varepsilon^{2}}\;. \label{TG1}%
\end{equation}

It is possible to state some general result for these potentials which depends
only on the asymptotics of $\Psi\left(  y\right)  $ as a power law as
$\left\vert y\right\vert \rightarrow\infty.$

\begin{theorem}
\label{ProofGenPow} Consider the family of potentials
(\ref{S9E7}) with $\Phi(\cdot,\varepsilon) \in \calC_s$, $s>1/2$. 
Suppose that the corresponding Holtsmark field
defined in Section \ref{Holtsm} is spatially homogeneous. 
Then
\begin{equation}
\lim\sup_{\varepsilon\rightarrow0}\sigma\left(  T_{BG};\varepsilon\right)
\leq\delta\left(  M\right)  \ \ \text{with\ }\lim_{M\rightarrow\infty}%
\delta\left(  M\right)  =0 \quad  \text{if} \quad  s>1 \label{S5E2a}%
\end{equation}
where $\sigma\left(  T;\varepsilon
\right)  $ is defined in (\ref{S4E8a}).
If we define $T_{L}$ by means of (\ref{S4E8}) we have
\begin{equation}
T_{L}\sim\frac{1}{4\pi A^{2}\varepsilon^{2}\log\left(  \frac{1}{\varepsilon
}\right)  }\ \ \text{as\ \ }\varepsilon\rightarrow0 \quad \text{if } \quad s=1
\label{P4E1}
\end{equation}
and
\begin{equation}
T_{L}\sim\left(  \frac{1}{W_{s}A^{2}\varepsilon^{2s}}\right)  ^{\frac{1}%
{3-2s}}\ \ \text{as\ \ }\varepsilon\rightarrow0 \quad  \text{if} \quad \frac{1}%
{2}<s<1 \label{P4E2}
\end{equation}
with 
\begin{equation}
W_{s}=\sup_{|\theta|=1} \int_{\mathbb{R}^{3}}d\xi \left[s^2\left(  \theta_{\bot}\cdot\xi_{\bot}\right)
^{2}\left(  \int_{0}^{1}\frac{d\tau}{\left\vert v\tau-\xi\right\vert ^{s+2}%
}\right)  ^{2}+ \theta_{\parallel}^2\,
\left(  \frac{1 }{\left\vert
    \xi\right\vert ^{s}}-\frac{1  }{\left\vert \xi-v\right\vert ^{s}}\right)  ^{2}
\right] \label{P7E1}
\end{equation}
(cf.\,\eqref{eq:notTPTP}), whence
\[
T_{L}\ll T_{BG}\ \ \text{as } \ \ \varepsilon\rightarrow0 \quad  \text{if }\quad 
s\leq1 .
\]
\end{theorem}

\begin{remark}
We recall that the assumption of spatial homogeneity for the Holtsmark
field means that we need to consider neutral distributions of charges
or distributions with a background charge if $s\leq1.$
\end{remark}

\subsubsection{Proof of Theorem \ref{ProofGenPow}: general strategy.}

We use the splitting
(\ref{S4E4}) which in the case of the family (\ref{S9E7}) becomes
$\Phi\left(  x,\varepsilon\right)  =\Phi_{B}\left(  x,\varepsilon\right)
+\Phi_{L}\left(  x,\varepsilon\right)$ with
\begin{equation}
\Phi_{B}\left(  x,\varepsilon\right)  =\Psi_{B,M}\left(  \frac{x}{\varepsilon
}\right)  \ \ ,\ \ \Phi_{L}\left(  x,\varepsilon\right)  =\Psi_{L,M}\left(
\frac{ x}{\varepsilon}\right) \;,
\ \label{S5E0}%
\end{equation}
\begin{equation}
\Psi_{B,M}\left(  y\right)  =\Psi\left(   y\right)
\eta\left(  \frac{y}{M}\right)  \ \ ,\ \ \Psi
_{L,M}\left(  y\right)  =\Psi\left(   y \right)  \left[
1-\eta\left(  \frac{ y }{M}\right)  \right] . \label{S5E0a}%
\end{equation}
We study the asymptotic
properties of the function%
\begin{equation}
\sigma\left(  T;\varepsilon\right)  =\sup_{\left\vert
\theta\right\vert =1}\int_{\mathbb{R}^{3}}dy\left(  \theta\cdot\int_{0}%
^{T}\nabla_{x}\Phi_{L}\left(  vt-y,\varepsilon\right)  dt\right)  ^{2}
\label{S5E1}%
\end{equation}
separately for the three different ranges $s>1,\ s=1$ and $\frac{1}{2}<s<1$
and assuming $|v|=1$.

\subsubsection{Proof of Theorem \ref{ProofGenPow}: the case $s>1$.%
\label{PowSLarge}}

Our goal is to prove (\ref{S5E2a}). To this end we first notice that
(\ref{S4E8a}) and (\ref{TG1}) imply%

\[
\sigma\left(  T_{BG};\varepsilon\right)  =\frac{1}{\varepsilon^{8}}%
\sup_{\left\vert \theta\right\vert =1}\int_{\mathbb{R}^{3}}d\xi\left(
\frac{\theta}{\varepsilon^{2}}\cdot\int_{0}^{1}\nabla\Psi_{L,M}\left(
\frac{v\tau-\xi}{\varepsilon^{3}}\right)  d\tau\right)  ^{2}\;.%
\]
We use that $\left\vert \nabla\Psi_{L,M}\right\vert \leq\frac
{C\chi_{\left[  M,\infty\right)  }\left(  \left\vert y\right\vert \right)
}{\left\vert y\right\vert ^{s+1}}$ for some $C>0$, where we denote as $\chi_{A}$ the
characteristic function of the set $A.$ Then%
\begin{equation}
\sigma\left(  T_{BG};\varepsilon\right)  \leq\frac{C}{\varepsilon^{8}}%
\int_{\mathbb{R}^{3}}d\xi\left(  \frac{\varepsilon^{3\left(  s+1\right)  }%
}{\varepsilon^{2}}\int_{0}^{1}\frac{\chi_{\left[  M,\infty\right)  }\left(
\frac{\left\vert v\tau-\xi\right\vert }{\varepsilon^{3}}\right)  }{\left\vert
v\tau-\xi\right\vert ^{s+1}}d\tau\right)  ^{2} . \label{P4E4a}%
\end{equation}

We split the integral as $\int_{\mathbb{R}^{3}}\left[  \cdot\cdot
\cdot\right]  d\xi=\int_{\left\{  \left\vert \xi\right\vert \geq2\right\}
}\left[  \cdot\cdot\cdot\right]  d\xi+\int_{\left\{  \left\vert
\xi\right\vert <2\right\}  }\left[  \cdot\cdot\cdot\right]  d\xi$
and notice that the first term is
\begin{equation}
\int_{\left\{  \left\vert \xi\right\vert \geq2\right\}  }\left[  \cdot
\cdot\cdot\right]  d\xi\leq C\varepsilon^{6\left(  s-1\right)  }%
\int_{\left\{  \left\vert \xi\right\vert \geq2\right\}  }\frac{d\xi}{\left\vert \xi\right\vert ^{2\left(
s+1\right)  }}\leq C\varepsilon^{6\left(  s-1\right)  } .\label{P4E4}%
\end{equation}

We split again the second domain as
$$\left\{  \left\vert \xi_{\bot}\right\vert
\geq\frac{M\varepsilon^{3}}{2},\ \left\vert \xi\right\vert <2\right\}
\cup\left\{  \left\vert \xi_{\bot}\right\vert <\frac{M\varepsilon^{3}}%
{2},\ \left\vert \xi\right\vert <2\right\} $$
where $\eta=\left(  \eta_{\parallel},\eta_{\bot}\right)$ (as in \eqref{eq:notTPTP}).
Note that 
$\left\vert v\tau-\xi\right\vert ^{s+1}=\left(
\left(  \tau-\xi_{\parallel}\right)  ^{2}+\left(  \xi_{\bot}\right)^{2}\right)
^{\frac{s+1}{2}}$.  
If $\left\vert \xi_{\bot
}\right\vert \geq\frac{M\varepsilon^{3}}{2}$ we use%
\begin{equation}
\int_{0}^{1}\frac{\chi_{\left[  M,\infty\right)  }\left(  \frac{\left\vert
v\tau-\xi\right\vert }{\varepsilon^{3}}\right)  }{\left\vert v\tau
-\xi\right\vert ^{s+1}}d\tau\leq\int_{-\infty}^{\infty}\frac{d\tau}{\left(
\tau^{2}+\left(  \xi_{\bot}\right)  ^{2}\right)  ^{\frac{s+1}{2}}}\leq
\frac{C}{\left\vert \xi_{\bot}\right\vert ^{s}}. \label{P4E5}%
\end{equation}
Otherwise if $\left\vert \xi_{\bot}\right\vert <\frac{M\varepsilon
^{3}}{2},$ since the integrand is different from zero only if
$\frac{\left\vert v\tau-\xi\right\vert }{\varepsilon^{3}}\geq M,$ we must
have $\left\vert \tau-\xi_{\parallel}\right\vert \geq\frac{M\varepsilon^{3}}{2}$ to
have a nontrivial contribution, hence
\begin{align}
\int_{0}^{1}\frac{\chi_{\left[  M,\infty\right)  }\left(  \frac{\left\vert
v\tau-\xi\right\vert }{\varepsilon^{3}}\right)  }{\left\vert v\tau
-\xi\right\vert ^{s+1}}d\tau &  \leq\int_{\left[  0,1\right]  \cap
\left\{  \left\vert \tau-\xi_{\parallel}\right\vert \geq\frac{M\varepsilon^{3}}%
{2}\right\}  }\frac{d\tau}{\left(  \left(  \tau-\xi_{\parallel}\right)  ^{2}+\left(  \xi_{\bot}\right)^{2}\right)  ^{\frac{s+1}{2}}}
\nonumber\\
&  \leq \frac{1}{\left\vert \xi_{\bot}\right\vert ^{s}}\int_{\left\{  \left\vert
\tau\right\vert \geq\frac{M\varepsilon^{3}}{2\left\vert \xi_{\bot}\right\vert
}\right\}  }\frac{d\tau}{\left(  \tau^{2}+1\right)  ^{\frac{s+1}{2}}}
%\leq\frac{C}{\left\vert \eta_{\bot}\right\vert ^{s}}\frac{1}{\left(\frac{M\varepsilon^{3}}{2\left\vert \eta_{\bot}\right\vert }\right)  ^{s}}%
\leq\frac{C}{M^{s}\varepsilon^{3s}}. \label{P4E6}
\end{align}
Combining (\ref{P4E5}) and (\ref{P4E6}) we get
\[
\int_{0}^{1}\frac{\chi_{\left[  M,\infty\right)  }\left(  \frac{\left\vert
v\tau-\xi\right\vert }{\varepsilon^{3}}\right)  }{\left\vert v\tau
-\xi\right\vert ^{s+1}}d\tau\leq\frac{C}{\left(  \max\left\{  \left\vert
\xi_{\bot}\right\vert ,M\varepsilon^{3}\right\}  \right)  ^{s}}.%
\]
Plugging this estimate into the term $\int_{\left\{  \left\vert \xi
\right\vert <2\right\}  }\left[  \cdot\cdot\cdot\right]  d\xi$ (cf.\,(\ref{P4E4a})) we arrive at 
\begin{align*}
&  \frac{1}{\varepsilon^{8}}\int_{\left\{  \left\vert \xi\right\vert
<2\right\}  }d\xi\left(  \frac{\varepsilon^{3\left(  s+1\right)  }%
}{\varepsilon^{2}}\int_{0}^{1}\frac{\chi_{\left[  M,\infty\right)  }\left(
\frac{\left\vert v\tau-\xi\right\vert }{\varepsilon^{3}}\right)  }{\left\vert
v\tau-\xi\right\vert ^{s+1}}d\tau\right)  ^{2}\\
&  \leq\frac{C}{\varepsilon^{8}}\int_{-2}^{2}d\xi_{\parallel}\int_{\left\{
\left\vert \xi_{\bot}\right\vert \leq M\varepsilon^{3}\right\}  }d\xi_{\bot
}\left(  \frac{\varepsilon^{3\left(  s+1\right)  }}{\varepsilon^{2}}\frac
{1}{\left(  M\varepsilon^{3}\right)  ^{s}}\right)  ^{2}+\\
&  +\frac{C}{\varepsilon^{8}}\int_{-2}^{2}d\xi_{\parallel}\int_{\left\{  \left\vert
\xi_{\bot}\right\vert >M\varepsilon^{3}\right\}  }d\xi_{\bot}\left(
\frac{\varepsilon^{3\left(  s+1\right)  }}{\varepsilon^{2}}\frac{1}{
\left\vert \xi_{\bot}\right\vert   ^{s}}\right)  ^{2}\leq \frac{C}{M^{2\left(
s-1\right)  }}\;.
\end{align*}

Using this and (\ref{P4E4}) we obtain%
\[
\sigma\left(  T_{BG};\varepsilon\right)  \leq C\varepsilon^{6\left(
s-1\right)  }+\frac{C}{M^{2\left(  s-1\right)  }}.
\]
Taking first the limit $\varepsilon\rightarrow0$ (using $s>1$) and then
$M\rightarrow\infty$ we obtain (\ref{S5E2a}). This gives the result of Theorem
\ref{ProofGenPow} for $s>1.$

\subsubsection{Proof of Theorem \ref{ProofGenPow}: the case $s=1$.%
\label{TimeScaleCoulomb}}

Our goal is to prove the existence of $T_{L}\ll T_{BG}$ such that
$\sigma\left(  T_{L};\varepsilon\right)  \sim1$ as $\varepsilon\rightarrow0$,
using (\ref{PsiAs}) with $s=1$. We assume that $A=1,$
since $A$ can be simply absorbed as a rescaling factor.

Changing variables $y=T\xi,\ t=T\tau$ in (\ref{S4E8a}) we can
write:%
\begin{align*}
\sigma\left(  T;\varepsilon\right)   
=\frac{T^3}{\varepsilon^{2}}%
\sup_{\left\vert \theta\right\vert =1}\int_{\mathbb{R}^{3}}d\xi\left(
T\theta \cdot\int_{0}^{1}\nabla\Psi_{L,M}\left(
\frac{T(v\tau-\xi)}{\varepsilon}\right)  d\tau\right)  ^{2}\;.%
\end{align*}
Let $Z = 1 - \eta$ (cf.\,\eqref{S5E0a}). 
Using (\ref{PsiAs}) we obtain%
\be \label{P4E8}
\sigma\left(  T;\varepsilon\right)  =T\varepsilon^{2}\sup
_{\left\vert \theta_{\bot}\right\vert =1}\int_{\mathbb{R}^{3}}d\xi\left(
\left(  \theta_{\bot}\cdot\xi_{\bot}\right)  \int_{0}^{1}\frac{Z\left(
\frac{T\left\vert v\tau-\xi\right\vert }{M\varepsilon}\right) 
}{\left\vert v\tau-\xi\right\vert ^{3}}d\tau\right)  ^{2}\left[  1+\zeta\left(
M;\varepsilon\right)  \right]  %
\ee
where $\lim\sup_{\varepsilon\rightarrow0}|\zeta\left(  M;\varepsilon\right)|
\leq\zeta\left(  M\right)  $ and $\zeta\left(  M\right)  \rightarrow0$ as
$M\rightarrow\infty.$ In $\zeta\left(  M;\varepsilon\right)$ we collect:
(i) the errors coming from the computation of the gradient via \eqref{S5E0a}
with the approximation \eqref{eq:defCCs} (with $s=1$ and $A=1$), which can be
estimated as in the previous section; (ii) the contribution of the longitudinal component
$\theta_\parallel$. Note that the latter yields a term of order
\begin{equation}\label{eq:longcont}
T\varepsilon^{2}\int_{\mathbb{R}^{3}}d\xi\left(  \frac{Z\left(
\frac{T\left\vert \xi\right\vert }{M\varepsilon}\right)  }{\left\vert
\xi\right\vert }-\frac{Z\left(  \frac{T\left\vert \xi-v\right\vert
}{M\varepsilon}\right)  }{\left\vert \xi-v\right\vert }\right)  ^{2}\;,
\end{equation}
which can be estimated by $CT\varepsilon^{2}.$  Since the integral in \eqref{P4E8} will produce an additional contribution of the order $\log\big(\frac{1}{\ep}\big)$, \eqref{eq:longcont} can be absorbed into  $\zeta\left(
M;\varepsilon\right)$.

We decompose the integral $\int_{\mathbb{R}^{3}}\left[  \cdot\cdot
\cdot\right]  d\xi$ in (\ref{P4E8}) as $\int_{\left\{  \left\vert
\xi\right\vert \geq2\right\}  }\left[  \cdot\cdot\cdot\right]  d\xi
+\int_{\left\{  \left\vert \xi\right\vert <2\right\}  }\left[  \cdot
\cdot\cdot\right]  d\xi$ and notice that the first one is bounded, so that
\begin{equation}
\sigma\left(  T;\varepsilon\right)  =T\varepsilon^{2}\sup
_{\left\vert \theta_{\bot}\right\vert =1}\int_{\left\{  \left\vert \xi\right\vert <2\right\} }d\xi\left(
\left(  \theta_{\bot}\cdot\xi_{\bot}\right)  \int_{0}^{1}\frac{Z\left(
\frac{T\left\vert v\tau-\xi\right\vert }{M\varepsilon}\right) 
}{\left\vert v\tau-\xi\right\vert ^{3}}d\tau\right)  ^{2}\left[  1+\zeta\left(
M;\varepsilon\right)  \right]  + O(T\ep^2)\;.\nn%
\end{equation}

Arguing similarly we obtain that the main contribution to the integral is due to a cylinder with principal axis $(\xi_{\parallel}\in [0,1],\xi_{\perp}=0)$ and radius much smaller than $1$. In particular
\be
\label{eq:intEsth}
\sigma\left(  T;\varepsilon\right)  \sim T\varepsilon^{2}\int_{0}%
^{1}d\xi_{\parallel}\int_{\left\{  \left\vert \xi_{\bot}\right\vert \leq1\right\}
}d\xi_{\bot}\left(  \theta_{\bot}\cdot\xi_{\bot}\right)  ^{2}\left(  \int_{0}^{1}\frac{Z\left(
\frac{T\left\vert v\tau-\xi\right\vert }{M\varepsilon}\right) 
}{\left\vert v\tau-\xi\right\vert ^{3}}d\tau\right)^{2}\;,
\ee
where the error is negligible in the limit $\ep \to 0 $ and then $M \to \infty$.

Note that the region with $\left\vert \xi_{\bot}\right\vert \leq\frac{2M\varepsilon}{T}$
yields also a contribution of order $T\varepsilon^{2}$ (as it might be seen
estimating the integral $\int_{0}^{1}\frac{d\tau}{\left\vert v\tau
-\xi\right\vert ^{3}}$ as $\frac{C}{\left(  \frac{M\varepsilon}{T}\right)
^{2}}$\ if $\frac{M\varepsilon}{T}\leq\left\vert \xi_{\bot}\right\vert
\leq\frac{2M\varepsilon}{T}$). 
We then have the approximation%
\[
\sigma\left(  T;\varepsilon\right)  \sim T\varepsilon^{2}\int_{0}%
^{1}d\xi_{\parallel}\int_{\left\{  \frac{2M\varepsilon}{T}\leq\left\vert \xi_{\bot
}\right\vert \leq1\right\}  }d\xi_{\bot}\left(  \theta_{\bot}\cdot\xi_{\bot}\right)
^{2}\left(  \int_{0}^{1}\frac{Z\left(  \frac{T\left\vert v\tau-\xi\right\vert
}{M\varepsilon}\right)  }{\left\vert v\tau-\xi\right\vert ^{3}}d\tau\right)
^{2}\;.
\]

Finally, for any $\xi_{\parallel}\in\left(  0,1\right)  $ we have%
\begin{equation}
\int_{0}^{1}\frac{Z\left(  \frac{T\left\vert v\tau-\xi\right\vert
}{M\varepsilon}\right)  }{\left\vert v\tau-\xi\right\vert ^{3}}d\tau\sim
\frac{1}{\left\vert \xi_{\bot}\right\vert ^{2}}\int_{-\infty}^{\infty}%
\frac{d\tau}{\left(  \tau^{2}+1\right)  ^{\frac{3}{2}}}=\frac{2}{\left\vert
\xi_{\bot}\right\vert ^{2}}\ \ \text{as }\left\vert \xi_{\bot}\right\vert
\rightarrow0,\ \ \left\vert \xi_{\bot}\right\vert \geq \frac{M\varepsilon}{T}\;.
\label{P4E7}%
\end{equation}
The approximation is not uniform in $\xi_{\parallel}$ when $\xi_{\parallel}$ is close to $0$ or $1,$ but the
contributions of those regions (which yield terms bounded by the right-hand
side of (\ref{P4E7})) are negligible compared with those due to the region
$\xi_{\parallel}\in\left(  \varepsilon_{0},1-\varepsilon_{0}\right)
,\ \varepsilon_{0}>0$ small. 
Therefore%
\[
\sigma\left(  T;\varepsilon\right)  \sim4T\varepsilon^{2}\int_{0}^{1}d\xi
_{\parallel}\int_{\left\{  \frac{2M\varepsilon}{T}\leq\left\vert \xi_{\bot}\right\vert
\leq1\right\}  }\frac{\left(  \theta_{\bot}\cdot\xi_{\bot}\right)  ^{2}%
}{\left\vert \xi_{\bot}\right\vert ^{4}}d\xi_{\bot}\]
and, computing the integral in $\xi_{\bot}$ using polar coordinates, we obtain%
\be
\label{eq:frpT41c}
\sigma\left(  T;\varepsilon\right)   \sim 4\pi T\varepsilon^{2}\log\left(
\frac{1}{\varepsilon}\right)  \left[  1+o\left(  1\right)  \right]
\ \ \text{as\ \ }\varepsilon\rightarrow0 \;.
\ee
Using (\ref{S4E8}), Eq.\,(\ref{P4E1}) follows.

\subsubsection{Proof of Theorem \ref{ProofGenPow}: the case $\frac{1}{2}%
<s<1$.} \label{sss:criticaltimescale}

We proceed as in the previous section, assuming 
$\frac{1}{2}<s<1$ and $A=1$.

The analogous of \eqref{P4E8}-\eqref{eq:longcont} is%
\begin{equation}
\sigma\left(  T;\varepsilon\right)  =s^{2}T^{3-2s}\varepsilon^{2s}%
\sup_{\left\vert \theta\right\vert =1}
(J_1({\theta_{\bot}})+J_2(\theta_\parallel)) \left[  1+\zeta\left(
M;\varepsilon\right)  \right]
\label{P4E9}%
\end{equation}
where
\begin{equation}
J_{1}({\theta_{\bot}}):=\int_{\mathbb{R}^{3}}d\xi \left(
\theta_{\bot}\cdot\xi_{\bot}\right)  ^{2}\left(  \int_{0}^{1}\frac{Z\left(
\frac{T\left\vert v\tau-\xi\right\vert }{M\varepsilon}\right)  d\tau
}{\left\vert v\tau-\xi\right\vert ^{s+2}}\right)  ^{2}
\label{P4E9_1}%
\end{equation}
and 
\begin{equation}
J_{2}({\theta_{\parallel}}):=\frac{1}{s^2}\int_{\mathbb{R}^{3}}d\xi \,
\theta_{\parallel}^2\,
\left(  \frac{Z\left(
\frac{T\left\vert \xi\right\vert }{M\varepsilon}\right)  }{\left\vert
\xi\right\vert^s }-\frac{Z\left(  \frac{T\left\vert \xi-v\right\vert
}{M\varepsilon}\right)  }{\left\vert \xi-v\right\vert^s }\right)  ^{2}\;.
\label{P4E9_2}%
\end{equation}

Notice however that, for this range of values of $s$, the
region where $\left\vert v\tau-\xi\right\vert \leq
\frac{M\varepsilon}{T}$ gives a negligible contribution, because the resulting
integral is finite, differently from the previous case. Thus we can replace 
$Z$ by $1$ introducing a negligible error. 

The integral $W_{s}$ in (\ref{P7E1}) is a
numerical constant depending only on $s$ and%
\[
\sigma\left(  T;\varepsilon\right)  \sim W_{s}T^{3-2s}\varepsilon^{2s}%
\]
as $\varepsilon\rightarrow0$, from which 
(\ref{P4E2}) follows. 

This concludes the proof of Theorem \ref{ProofGenPow}.

\subsection{Computation of the correlations for potentials
$\Phi\left(  x,\varepsilon\right)  =\Psi\left(  \frac{\left\vert x\right\vert
}{\varepsilon}\right)  $.}

We now estimate the correlations of the deflections for families of potentials
with the form $\Phi\left(  x,\varepsilon\right)  =\Psi\left(  \frac{\left\vert
x\right\vert }{\varepsilon}\right)  .$ We restrict to the cases where
$T_{L}\ll T_{BG},$ i.e.\,to potentials with the asymptotics (\ref{PsiAs}) with
$s\leq1$. We indicate in this section
the deflection vector during the time interval $\left[  0,\tilde{T}%
_{L}\right]  $ as%
\be
D\left(  x_{0},v;\tilde{T}_{L}\right)  =\int_{0}^{\tilde{T}_{L}}\nabla_x\Phi_{L}\left(
x_{0}+vt,\varepsilon\right)\omega \, dt \;,
\label{eq:DevCorr}
\ee
where $\tilde{T}_{L}=h T_{L},$ $h>0$, $x_{0}\in
\mathbb{R}^{3}$ and $v\in\mathbb{R}^{3}$ with $\left\vert v\right\vert =1$.

\begin{theorem}
\label{CorrPowLaw}Suppose that the assumptions in Theorem \ref{ProofGenPow} hold.
\begin{itemize}
\item[(i)] Let us assume that $s=1$ and $T_{L}$ is as in (\ref{P4E1}). Then (cf.\,\eqref{I1E1})
\begin{align}
&\mathbb{E}\left(  D\left(  x_{0},v;\tilde{T}_{L}\right)  D\left(
x_{0}+v\tilde{T}_{L},v;\tilde{T}_{L}\right)  \right)\nonumber \\&  \ll\sqrt{\mathbb{E}%
\left(  \left(  D\left(  x_{0},v;\tilde{T}_{L}\right)  \right)  ^{2}\right)
\mathbb{E}\left(  \left(  D\left(  x_{0}+v\tilde{T}_{L},v;\tilde{T}%
_{L}\right)  \right)  ^{2}\right)  }\ \ \text{as\ \ }\varepsilon\rightarrow 0.
\label{I1E1a}%
\end{align}

\item[(ii)] Suppose that $\frac{1}{2}<s<1$ and $T_{L}$ is as in (\ref{P4E2}). Let
$x_{1},x_{2}\in\mathbb{R}^{3}$, $\left(  x_{2}-x_{1}\right)
=T_{L}y$ with $y\in\mathbb{R}^{3}$ and %$v_{1}=n_1,v_{2}=n_2$, with 
$v_1,v_2\in S^2$. 
Then (cf.\,Section \ref{ss:CorrCase})%
\[
\mathbb{E}\left(  D\left(  x_{1},v_{1};\tilde{T}_{L}\right)  D\left(
x_{2},v_{2};\tilde{T}_{L}\right)  \right)  \sim K\left(y,v_{1}%
,v_{2};h\right)  \ \ \text{as\ \ }\varepsilon\rightarrow 0\;,
\]
where:
\begin{equation}
K\left(y,v_{1}
,v_{2};h\right)  \sim\frac{\Lambda\left(e\right)  h^{2}}{\left\vert y\right\vert ^{2s-1}}\text{\ \ as\ \ }%
\left\vert y\right\vert \rightarrow\infty \label{I1E1b}%
\end{equation}
with $e = \frac{y}{|y|}$ and
\[
\Lambda\left(e\right)   :=\frac{s^2}{W_s}\int_{\mathbb{R}^{3}}d\eta\frac{\left[
\eta\otimes \left(  \eta-e\right)  \right]}{\left\vert \eta\right\vert ^{s+2}\left\vert \eta-e\right\vert ^{s+2}%
}\;,%\ \ ,\ \ \ \left\vert e\right\vert =1,\ e\in\mathbb{R}^{3}%
\]
and
$K(  0,v_{1},v_{2};h)  =O(h^{3-2s})$. Moreover, as $\ep \to 0$ the correlation matrix is
\be
\label{eq:corrmatex}
\frac{\mathbb{E}\left(  D\left(  x_{1},v_{1};\tilde{T}_{L}\right)  \otimes
D\left(  x_{2},v_{2};\tilde{T}_{L}\right)  \right)  }{\sqrt{\mathbb{E}\left(
\left(  D\left(  x_{1},v_{1};\tilde{T}_{L}\right)  \right)  ^{2}\right)
\mathbb{E}\left(  \left(  D\left(  x_{2},v_{2};\tilde{T}_{L}\right)  \right)
^{2}\right)  }} \sim\frac{K\left( y,v_{1},v_{2};h\right)  }{C_{s}h^{3-2s}}
\ee
where $C_s>0$ is given by \eqref{eq:CsCorr} below.
\end{itemize}
\end{theorem}

\begin{remark}
Notice that we approximate in the case (i) the trajectory of the particle by
rectilinear ones. In an analogous manner, we could prove that
the correlations also tend to zero if we consider particles separated by
distances larger than $\tilde{T}_{L}.$ In case (ii) we obtain that the
correlations between the particles at distances of the order of the mean free
path do not vanish as $\varepsilon\rightarrow0.$
\end{remark}

\subsubsection{Proof of Theorem \ref{CorrPowLaw}: the case $s=1$.}

We assume, without loss of generality, that $x_{0}=0.$ 
The result \eqref{eq:frpT41c} in Section  \ref{TimeScaleCoulomb} shows that the asymptotic
behaviour in the right-hand side of \eqref{I1E1a} is given, up to 
a multiplicative constant, by $\tilde{T}_{L}\varepsilon^{2}\log\left(  \frac
{1}{\varepsilon}\right)  ,$ which is of order $h$ as $\varepsilon
\rightarrow0$\ if $\tilde{T}_{L}=h T_{L}.$ Therefore, we need to prove
that $\mathbb{E}\left(  D\left(  0,v;\tilde{T}_{L}\right)  D\left(  v\tilde
{T}_{L},v;\tilde{T}_{L}\right)  \right)  \ll1$ as $\varepsilon\rightarrow0.$

We get%
\begin{align}
&\mathbb{E}\left(  D\left(  0,v;\tilde{T}_{L}\right)  D\left(  v\tilde{T}%
_{L},v;\tilde{T}_{L}\right)  \right)  \nonumber\\&=\int_{\mathbb{R}^{3}}d\xi\left(
\int_{0}^{\tilde{T}_{L}}\nabla_{x}\Phi_{L}\left(  vt_{1}-\xi,\varepsilon
\right)  dt_{1}\right)  \otimes\left(  \int_{\tilde{T}_{L}}^{2\tilde{T}_{L}%
}\nabla_{x}\Phi_{L}\left(  vt_{2}-\xi,\varepsilon\right)  dt_{2}\right)\;.
\label{A1}%
\end{align}

Arguing as in Section \ref{TimeScaleCoulomb} (cf.\,\eqref{P4E8}) we can prove that the longitudinal contributions (i.e.\,those parallel to $v$) are negligible. Therefore we only consider the components on the plane orthogonal to $v$.

Using the rescaling of variables $t_{1}=T_{L}\tau_1,\ t_{2}=T_{L}\tau_{2},\ \xi=T_{L} y$, we obtain
that 
\eqref{A1} is bounded by $\frac{\left(  T_{L}\right)  ^{5}}{\varepsilon^{2}}$ times
(cf.\,\eqref{S5E0}) 
$$
\int_{0}^{h
}d\tau_{1}\int_{h}^{2h}d\tau_{2}\int_{\mathbb{R}^{3}}dy\left\vert
\Psi_{L,M}^{\prime}\left(  \frac{T_{L}\left(  v\tau_{1}-y\right)
}{\varepsilon}\right)  \right\vert \left\vert \Psi_{L,M}^{\prime}\left(
\frac{T_{L}\left(  v\tau_{2}-y\right)  }{\varepsilon}\right)  \right\vert
\frac{\left\vert y_{\bot}\right\vert }{\left\vert v\tau_{1}-y\right\vert
}\frac{\left\vert y_{\bot}\right\vert }{\left\vert v\tau_{2}-y
\right\vert }%
$$
where $y_{\bot}$ denotes the orthogonal projection of $y$ in the plane
orthogonal to $v.$ That is, for some $C>0$ (cf.\,\eqref{S5E0a}, \eqref{PsiAs}),
\begin{align*}
& \mathbb{E}\left(  D\left(  0,v;\tilde{T}_{L}\right)  D\left(  v\tilde {T}_{L},v;\tilde{T}_{L}\right)  \right)\nonumber\\& \leq 
 CT_{L}\varepsilon^{2}\int_{0}^{h}d\tau_{1}\int_{h}^{2h}%
d\tau_{2}\int_{\mathbb{R}^{3}}dy \left\vert y_{\bot}\right\vert ^{2}
\frac{\chi_{\left\{  \left\vert y\right\vert \geq\frac{\varepsilon}{T_{L}%
}\right\}  }}{\left\vert y\right\vert ^{3}}\frac{\chi_{\left\{  \left\vert
v\left(  \tau_{2}-\tau_{1}\right)  -y\right\vert \geq\frac{\varepsilon
}{T_{L}}\right\}  }}{\left\vert v\left(  \tau_{2}-\tau_{1}\right)
-y\right\vert ^{3}}\\
& = 
 CT_{L}\varepsilon^{2}\int_{-h}^{0}d\tau_{1}\int_{0}^{h}%
d\tau_{2}\int_{\mathbb{R}^{2}}dy_\bot \left\vert y_{\bot}\right\vert ^{2}
\int_{\mathbb{R}} dy_\parallel \frac{\chi_{\left\{  \left\vert y\right\vert \geq\frac{\varepsilon}{T_{L}%
}\right\}  }}{(y^2_\parallel+y^2_\bot)^{\frac{3}{2}}}\frac{\chi_{\left\{  \left\vert
v\left(  \tau_{2}-\tau_{1}\right)  -y\right\vert \geq\frac{\varepsilon
}{T_{L}}\right\}  }}{((y_\parallel-(\tau_2-\tau_1))^2+y^2_\bot)^{\frac{3}{2}}}\;.
\end{align*}

We now use the change of variables $y_\parallel=\left\vert y_{\bot}\right\vert
X.$ Estimating also the characteristic functions by
$1$, we find
\begin{align*}
  \mathbb{E}\left(  D\left(  0,v;\tilde{T}_{L}\right)  D\left(  v\tilde
{T}_{L},v;\tilde{T}_{L}\right)  \right) 
\leq CT_{L}\varepsilon^{2}\int_{-h}^{0}d\tau_{1}\int_{0}^{h}d\tau
_{2}\int_{\mathbb{R}^{2}}\frac{dy_{\bot}}{\left\vert y_{\bot}\right\vert
^{3}}Q\left(  \frac{\tau_{2}-\tau_{1}}{\left\vert y_{\bot}\right\vert
}\right)
\end{align*}
where
\[
Q\left(  s\right)  =\int_{\mathbb{R}}\frac{dX}{\left(  X^{2}+1\right)
^{\frac{3}{2}}}\frac{1}{\left(  \left(  X-s\right)  ^{2}+1\right)  ^{\frac
{3}{2}}}.
\]
We remark that $Q\left(  s\right)  \leq\frac{C}{1+\left\vert s\right\vert
^{3}}.$ Then%
\begin{align*}
\mathbb{E}\left(  D\left(  0,v;\tilde{T}_{L}\right)  D\left(  v\tilde{T}%
_{L},v;\tilde{T}_{L}\right)  \right)   &  \leq CT_{L}\varepsilon^{2}%
\int_{-h}^{0}d\tau_{1}\int_{0}^{h}d\tau_{2}\int_{\mathbb{R}^{2}}%
\frac{dy_{\bot}}{\left\vert y_{\bot}\right\vert ^{3}+\left(  \tau
_{2}-\tau_{1}\right)  ^{3}}\\
&  \leq CT_{L}\varepsilon^{2}\int_{-h}^{0}d\tau_{1}\int_{0}^{h}%
\frac{d\tau_{2}}{\left(  \tau_{2}-\tau_{1}\right)  }
\leq CT_{L}\varepsilon^{2}
\end{align*}
which is negligible as $\varepsilon\rightarrow0,$ by using the formula \eqref{P4E1}.

\subsubsection{Proof of Theorem \ref{CorrPowLaw}: the case $s<1$.} \label{subsec:PrThCPLs<1}

By definition
\begin{align*}
&  \mathbb{E}\left(  D\left(  x_{1},v_{1};\tilde{T}_{L}\right)  \otimes
D\left(  x_{2},v_{2};\tilde{T}_{L}\right)  \right) \\
&  =\int_{0}^{\tilde{T}_{L}}dt_{1}\int_{0}^{\tilde{T}_{L}}dt_{2}%
\int_{\mathbb{R}^{3}}d\xi\,\nabla_{x}\Phi_{L}\left(  \xi-v_{1}t_{1}%
,\varepsilon\right)  \otimes\nabla_{x}\Phi_{L}\left(  \xi-\left(  x_{2}%
-x_{1}\right)  -v_{2}t_{2},\varepsilon\right)
\end{align*}
and we are interested in a situation where $\left(  x_{2}-x_{1}\right)=T_{L}y, y\in\mathbb{R}^{3}$.  
We use again the change of variables $\xi=T_{L}%
\eta,\ t_{j}=T_{L}\tau_{j},\ j=1,2$. 
Then:%
\begin{align*}
&  \mathbb{E}\left(  D\left(  x_{1},v_{1};\tilde{T}_{L}\right)  \otimes
D\left(  x_{2},v_{2};\tilde{T}_{L}\right)  \right) \\
&  =\frac{\left(  T_{L}\right)  ^{5}}{\varepsilon^{2}}\int_{0}^{h}%
d\tau_{1}\int_{0}^{h}d\tau_{2}\int_{\mathbb{R}^{3}}d\eta\,\Psi_{L,M}^{\prime
}\left(  \frac{T_{L}\left(  \eta-v_{1}\tau_{1}\right)  }{\varepsilon}\right)
\Psi_{L,M}^{\prime}\left(  \frac{T_{L}\left(  \eta-y-v_{2}\tau_{2}\right)
}{\varepsilon}\right) \\&
 \quad\left[  \frac{\left(  \eta-v_{1}\tau_{1}\right)  }{\left\vert
\eta-v_{1}\tau_{1}\right\vert }\otimes\frac{\left(  \eta-y-v_{2}\tau
_{2}\right)  }{\left\vert \eta-y-v_{2}\tau_{2}\right\vert }\right]  \;.
\end{align*}

Using (\ref{PsiAs}) (with $A=1$) we obtain that, up to an arbitrarily 
small error  $\zeta\left(
M;\varepsilon\right) $ (as in \eqref{P4E8}), 
\begin{align*}
&  \mathbb{E}\left(  D\left(  x_{1},v_{1};\tilde{T}_{L}\right)  \otimes
D\left(  x_{2},v_{2};\tilde{T}_{L}\right)  \right)  \\
&  \sim s^{2}\left(  T_{L}\right)  ^{3-2s}\varepsilon^{2s}\int_{0}^{h}%
d\tau_{1}\int_{0}^{h}d\tau_{2}\int_{\mathbb{R}^{3}}d\eta\frac{\left[  \left(  \eta-v_{1}\tau_{1}\right)
\otimes\left(  \eta-y-v_{2}\tau_{2}\right)  \right] }{\left\vert
\eta-v_{1}\tau_{1}\right\vert ^{s+2}\left\vert \eta-y-v_{2}\tau
_{2}\right\vert ^{s+2}} \\
&  =\frac{s^2}{W_{s}}\int_{0}^{h}d\tau_{1}\int_{0}^{h}d\tau_{2}\int_{\mathbb{R}^{3}}d\eta\frac{\left[  \left(  \eta-v_{1}\tau_{1}\right)
\otimes\left(  \eta-y-v_{2}\tau_{2}\right)  \right] }{\left\vert
\eta-v_{1}\tau_{1}\right\vert ^{s+2}\left\vert \eta-y-v_{2}\tau
_{2}\right\vert ^{s+2}} \\
& =\frac{s^2}{W_{s}}\int_{0}^{h}d\tau_{1}\int_{0}^{h}d\tau_{2} \int_{\mathbb{R}^{3}} d\eta\frac{\left[  \eta\otimes\left(  \eta-U\right)
\right] }{\left\vert \eta\right\vert ^{s+2}\left\vert \eta-U\right\vert
^{s+2}}\;,
\end{align*}
where we have used that $T_{L}\sim\left(  \frac{1}{W_{s}\varepsilon^{2s}%
}\right)  ^{\frac{1}{3-2s}}$ (cf.\,(\ref{P4E2})) and 
that $U:=\left[  y+v_{2}\tau_{2}-v_{1}\tau_{1}\right]  .$ We remark that the integral in the variable $\eta$
is well defined for each $U\in\mathbb{R}^{3},$ $U\neq0$ since $\frac{1}{2}<s<1.$ 
Let
$e=\frac{U}{\left\vert U\right\vert }$ be a unit vector in the direction of $U.$
Then, rescaling we obtain
\be
\label{eq:proThmCorr-i}
 \mathbb{E}\left(  D\left(  x_{1},v_{1};\tilde{T}_{L}\right)
\otimes D\left(  x_{2},v_{2};\tilde{T}_{L}\right)  \right)   \sim
\int_{0}^{h}d\tau_{1}\int_{0}^{h}d\tau_{2}\frac{\Lambda(e)}{\left\vert
y+v_{2}\tau_{2}-v_{1}\tau_{1}\right\vert ^{2s-1}}
\ee
where
\[
\Lambda\left(  e\right)  =\frac{s^2}{W_{s}}\int_{\mathbb{R}^{3}}d\eta\frac{\left[  \eta
\otimes\left(  \eta-e\right)  \right] }{\left\vert \eta\right\vert
^{s+2}\left\vert \eta-e\right\vert ^{s+2}}\ \ ,\ \ \ \left\vert e\right\vert
=1\;.
\]

We are interested now in taking $h$ much smaller than
$\left\vert y\right\vert.$ Then the following approximation holds:
\[
\mathbb{E}\left(  D\left(  x_{1},v_{1};\tilde{T}_{L}\right)
\otimes D\left(  x_{2},v_{2};\tilde{T}_{L}\right)  \right)  
\sim\frac{h^{2}\Lambda(e)}{\left\vert y\right\vert ^{2s-1}}
\]
with $e =  \frac{y}{|y|}$. Therefore \eqref{I1E1b} is proved. 
Notice also that \eqref{eq:proThmCorr-i} implies
\[
\left\| K\left(y,v_{1}%
,v_{2};h\right)  \right\| \leq\frac
{Ch^{2}}{\left(  \left\vert y\right\vert + h 
\right)  ^{2s-1}}\;.
\]

Similarly, we can compute the typical deflection from a 
given point $x_1,v_1$:
\bea
&&\mathbb{E}\left(  D\left(  x_{1},v_{1};\tilde{T}_{L}\right)  \otimes D\left(
x_{1},v_{1};\tilde{T}_{L}\right)  \right) \\
&& \sim\frac{s^2}{W_{s}}\int_{0}^{h}d\tau_{1}\int_{0}^{h}d\tau_{2}\int_{\mathbb{R}^{3}}d\eta\frac{\left[  \left(  \eta-v_{1}\tau_{1}\right)
\otimes\left(  \eta-v_{1}\tau_{2}\right)  \right] }{\left\vert
\eta-v_{1}\tau_{1}\right\vert ^{s+2}\left\vert \eta-v_{1}\tau
_{2}\right\vert ^{s+2}}\nn\\
&& =\frac{s^2}{W_{s}}\int_{0}^{h}d\tau_{1}\int_{0}^{h}d\tau_{2}\int_{\mathbb{R}^{3}}d\eta\frac{\left[ \eta
\otimes\left(  \eta-v_{1}(\tau_{2}-\tau_1)\right)  \right] }{\left\vert
\eta\right\vert ^{s+2}\left\vert \eta-v_{1}(\tau
_{2}-\tau_1)\right\vert ^{s+2}}\nn\\
&& =\frac{s^2}{W_{s}}\int_{0}^{h}d\tau_{1}\int_{0}^{h}\frac{d\tau_{2}}{|\tau_2-\tau_1|^{2s-1}}\int_{\mathbb{R}^{3}}d\eta\frac{\left[ \eta
\otimes\left(  \eta-v_{1}\right)  \right] }{\left\vert
\eta\right\vert ^{s+2}\left\vert \eta-v_{1}\right\vert ^{s+2}}\;.\nn
\eea
The last integral 
is a matrix and it remains invariant under rotations $v_1\rightarrow Rv_1,$ where
$R\in O\left(  3\right)  ,$ whence it is a multiple of the identity
$\sigma_{s}I$. Since the matrix is positive definite we have $\sigma_{s}>0.$
Then the integral above becomes%
\[
\frac{s^2\sigma_{s}}{W_{s}}\int_{0}^{h}d\tau_{1}\int_{0}^h \frac
{d\tau_{2}}{\left\vert \tau_{2}-\tau_{1}\right\vert ^{2s-1}}=C_{s}h^{3-2s}%
\]
where 
\be
\label{eq:CsCorr}
C_{s}=\frac{s^2\sigma_{s}}{W_{s}}\int_{0}^{1}%
d\tau_{1}\int_{0}^{1}\frac{d\tau_{2}}{\left\vert \tau_{2}-\tau_{1}\right\vert
^{2s-1}}\;.
\ee

We have then obtained 
\begin{align*}
\frac{\mathbb{E}\left(  D\left(  x_{1},v_{1};\tilde{T}_{L}\right)  \otimes
D\left(  x_{2},v_{2};\tilde{T}_{L}\right)  \right)  }{\sqrt{\mathbb{E}\left(
\left(  D\left(  x_{1},v_{1};\tilde{T}_{L}\right)  \right)  ^{2}\right)
\mathbb{E}\left(  \left(  D\left(  x_{2},v_{2};\tilde{T}_{L}\right)  \right)
^{2}\right)  }}\sim \mathcal{C}\left(  y,v_1,v_2;h\right)
\end{align*}
where the correlation function is given by \eqref{eq:corrmatex} and
\[
\| \mathcal{C}\left(  y,v_1,v_2;h\right)  \|
\leq\frac{C}{\left(  1+\frac{\left\vert y\right\vert }{
h}\right)  ^{2s-1}}\;.%
\]

\subsubsection{Kinetic equations.} \label{subsec:KE1}

We can now argue as in Section \ref{GenKinEq} to write the kinetic equations
yielding the evolution for the function $f\left(  t,x,v\right)$, for families of potentials
with the form $\Phi\left(  x,\varepsilon\right)  =\Psi\left(  \frac{\left\vert
x\right\vert }{\varepsilon}\right)$. Recall that the long range behaviour is given by (\ref{PsiAs}).

\paragraph{The case $s>1$: Boltzmann equation.} 
Let us assume that the scattering problem associated to the potential $\Psi$ is well posed for 
every $V$ and almost every impact parameter $b$ (cf.\,Section \ref{ScattPb}). Then, since (\ref{S5E2a}) holds, Claim
\ref{BoltGen} yields that the function defined by means of (\ref{P1E5a})
solves%
\begin{equation}
\left(  \partial_{t}f+v\partial_{x}f\right)  \left(  t,x,v\right)
=\int_{S^{2}}B\left(  v;\omega \right)  \left[  f\left(  t,x,\left\vert
v\right\vert \omega\right)  -f\left(  t,x,v\right)  \right]  d\omega \nn%
\end{equation}
if there are only charges of one type and%
\begin{equation}
\left(  \partial_{t}f+v\partial_{x}f\right)  \left(  t,x,v\right)
=\sum_{j=1}^{L}\mu\left(  Q_{j}\right)\int_{S^{2}}B\left(  v;\omega;Q_j \right)  \left[  f\left(  t,x,\left\vert
v\right\vert \omega\right)  -f\left(  t,x,v\right)  \right]  d\omega\nn
\end{equation}
if the distribution of scatterers contains more than one type of charges. In
these equations the scattering kernel $B$ is given by (\ref{P1E8})-(\ref{P1E8a}).

A particular case is $\Psi\left(  y\right)  :=\frac{1}{\left\vert
y\right\vert ^{s}},\ s>1.$ A rescaling argument allows to
restrict to $V=1$ and the expression for the kernel (with $|v|=1$) is
\[
B\left(  v;\omega\right)  =\frac{b}{|\sin\chi|}\Big|\left(  \frac{\partial\chi\left(  b\right)
}{\partial b}\right)  ^{-1}\Big|%
\]
where the scattering angle $\chi\left(  b\right) =\chi\left(  b,1\right)  $ is a monotone function of
$b$ given by%
\begin{equation}
\chi\left(  b\right)  =\pi-2\int_{r_{\ast}}^{+\infty
}\frac{b\,dr}{r^{2}\sqrt{1-2\Psi_{eff}\left(  r\right)}}\;, \label{eq:SE_0}%
\end{equation}
with%
\[
\Psi_{eff}(r)=\frac{1}{r^s}+\frac{b^{2}}{2r^{2}}
\]
and $r_*$ the unique solution of $2\Psi_{eff}(r_*)=1$.
One finds
\begin{equation}
\frac{\partial\chi\left(  b\right)  }{\partial b}=\frac{2s}{b^{s+1}}\int
_{0}^{\frac{\pi}{2}}\frac{\sin(\xi)u^{s-1}}{\left(  u+\frac{s}{b^{s}}%
u^{s-1}\right)  ^{2}}\left\{  (s-1)\frac{\frac{1}{s-1}+\frac{su^{s-2}}{b^{s}}%
}{\left(  1+\frac{su^{s-2}}{b^{s}}\right)  }-s\right\}  d\xi\;,\label{eq:SE_3}%
\end{equation}
where $u$ and $\xi$ are related by%
\[
\sin^{2}\left(  \xi\right)  =u^{2}+2\left(  \frac{u}{b}\right)  ^{s}\;.%
\]

\paragraph{The case $s=1$: Landau equation.}
Combining Theorems \ref{ProofGenPow} and \ref{CorrPowLaw} we obtain that
the function $f\left( t,x,v\right)  $ in (\ref{LimEx}) satisfies the
linear Landau equation, which in the case of charges of a single type has the form
\begin{equation}
\left(  \partial_{t}f+v\partial_{x}f\right)  \left(  t,x,v\right)  =\kappa\,\Delta_{v_{\perp}}f\left(  t,x,v \right)  \label{P4E3}%v_{\parallel},v_{\perp}\right)  \label{GenLanEq}%
\end{equation}
where $\kappa=\frac{1}{2}$ (since we have absorbed all the numerical constants in the formula for
$T_{L}$, see Section \ref{TimeScaleCoulomb}). If we have charges of different
types (cf.\,\eqref{S4E9}), the same definition of $T_{L}$ in (\ref{P4E1}) leads to
\be
\kappa=\frac{1}{2}\sum_{j=1}^{L}\mu\left(  Q_{j}\right)  Q_{j}^{2}\;.
\ee

Coulombian potentials, i.e.\,$\Psi\left(  y\right)  :=\frac{1}{\left\vert
y\right\vert }$, are particularly relevant in plasma physics and in
astrophysics where kinetic equations are used to
describe the relaxation to equilibrium. The
presence of the logarithmic term in (\ref{P4E1}) is well known in both
fields \cite{BT,LL2}. As explained in \cite{BT}, in systems where
the particles interact by means of Coulombian potentials the scatterers at
distances between $R$ and $2R$ of the trajectory with $R$
larger than the collision length contribute equally to the deflections. 
This is the reason for the onset of the
logarithmic term, and also for the fact that the large amount of small
deflections yields a larger effect than Boltzmann-type collisions with individual scatterers.

\paragraph{The case $\frac{1}{2}<s<1$: correlated deflections in times of the
order of $T_L$.}

In this case we have $T_{L}\ll T_{BG}.$ However, due to Theorem
\ref{CorrPowLaw}, the correlations between the deflections in times of the order of $T_{L}$ are of order one. 
Therefore the dynamics of
the distribution function $f$ cannot be approximated by means of the Landau
equation. In fact the probability distribution for the deflection in a time
$h T_{L}$ is a gaussian distribution with zero average and typical deviation of order $h
^{\frac{3-2s}{2}}$ in the limit $\varepsilon\rightarrow0$, i.e.\,we obtain (cf.\,Section \ref{sss:criticaltimescale})%
\[
m_{h T_L}^{(\ep)}\left(  \theta \right)   \rightarrow\exp\left(
-\kappa \,h^{3-2s}\theta^{2}\right)  \text{ as
}\varepsilon\rightarrow0, \quad \kappa>0\;.
\]

Diffusive processes (in the space of velocities) like the ones given by
the Landau equation are characterized by typical deviations of order
$\sqrt{h}$ which only take place for $s=1.$ Therefore, a diffusive process
cannot be expected if $\frac{1}{2}<s<1$, but rather a stochastic differential equation with correlations as explained in Section  \ref{ss:CorrCase}, see \eqref{ST1}-\eqref{ST7}.

\begin{remark}
\label{RangeInter} The analysis of the function $\sigma\left(  T;\varepsilon\right)$ 
given by \eqref{S5E1} allows to determine the set of scatterers which
influence the dynamics of the tagged particle. We will denote this set
as `domain of influence'. This corresponds to the regions in the $y$ variable which
determine the asymptotics of the function $\sigma\left(  T_{L};\varepsilon
\right)  $ as $\varepsilon\rightarrow0$ if $T_{L}\lesssim T_{BG}.$ Assume 
that $|v|=1$ and that the tagged particle is in the origin at time zero.
For the potentials with the form (\ref{S9E7}) considered in this section,
we obtain that in the case $s=1$ the domain of influence are the scatterers 
located in $\left(  x_{\parallel},x_{\bot}\right)  $ with
$x_{\parallel}\in\left[  0,T_{L}\right]  $ and $k_{1}\leq\left\vert x_{\bot
}\right\vert \leq\frac{T_{L}}{k_{1}},$ where $k_{1}$ is a large number. These
scatterers are responsible for the logarithmic correction which determines the
time scale $T_{L}$ (see Section \ref{TimeScaleCoulomb}). If $s<1,$ the domain of influence is
$x_{\parallel}\in\left[  0,T_{L}\right]  $ and
$\left\vert x_{\bot}\right\vert \leq k_{1}T_L$ (see Section \ref{sss:criticaltimescale}).
\end{remark}

\begin{remark} \label{rem:2D2}
In the two-dimensional case, we may consider families of potentials of the form
$\Phi\left(  x,\varepsilon\right)
=\Psi\left(  \frac{\left\vert x\right\vert }{\varepsilon}\right)$ with
$\Phi(\cdot,\varepsilon) \in \calC_s$ for any $s>0$ and we always obtain
spatially homogeneous Holtsmark fields (see Remark \ref{rem:2D1}). 
Nevertheless, unlike in three dimensions, the Coulombian decay does not correspond 
to the crossing of the Boltzmann and the Landau time scales. Indeed, $T_{BG} = \frac{1}{\ep}$,
\eqref{S5E2a} holds if $s>\frac{1}{2}$ and $T_L$ diverges as $\frac{1}{\ep\log\frac{1}{\ep}}$ if $s=\frac{1}{2}$.
Moreover, $T_{L}\sim\left(  \frac{1}{W_{s}A^{2}\varepsilon^{2s}}\right)  ^{\frac{1}%
{2-2s}}\ \ \text{as\ \ }\varepsilon\rightarrow0\ \ \text{if\ \ }0<s<\frac{1}{2}\;.$ Therefore
$$ T_{L}\ll T_{BG}\ \ \text{as\ \ }\varepsilon\rightarrow0\ \ \text{if\ \ }%
s\leq \frac{1}{2} \ \ \text{in two dimensions}.$$
Moreover, \eqref{I1E1a} is valid for $s = \frac{1}{2}$ and a Landau equation is expected to hold.
Instead, for $0<s<\frac{1}{2}$ the correlations do not vanish on the scale of the mean free path
and a set of equations with memory arises as in \eqref{ST1}-\eqref{ST7} (with $\alpha = 2s$ and $\alpha' = 1-s$).
\end{remark}

\subsection{Potentials with the form $\Phi\left(  x,\varepsilon\right)
=\varepsilon \,G\left(  \left\vert x\right\vert \right)  $.}

We will now consider families of potentials with a form different from
(\ref{S9E7}). We shall see how sensitively the kinetic time scales $T_{BG},\ T_{L}$ 
and the resulting limiting kinetic equation can depend on the
specific details of the interaction. 
Let us consider%
\begin{equation}
\left\{  \Phi\left(  x,\varepsilon\right)  ;\varepsilon>0\right\}  =\left\{
\varepsilon G\left(  \left\vert x\right\vert \right)  ;\varepsilon>0\right\}
\label{T1E1}%
\end{equation}
where $G \in \calC_s$, $s > 1/2$. We have
$G\in C^{2}\left(  \mathbb{R}^{3}\setminus\left\{  0\right\}\right)$ and
\begin{equation}
G\left(  x\right)  \sim\frac{A}{\left\vert x\right\vert ^{s}}%
\ \ \text{as\ \ }\left\vert x\right\vert \rightarrow\infty,\ \ A\neq0\;.
\label{T1E2}%
\end{equation}
Note that these potentials have an intrinsic length scale of order one, i.e.\,the
order of magnitude of the average distance between scatterers. Other types of
potentials with different or additional length scales might be considered
with analogous types of arguments, but we restrict to the present case for simplicity.
Moreover, we restrict to classes of functions satisfying %either 
%or%
\begin{equation}
G\left(  x\right)  \sim\frac{B}{\left\vert x\right\vert ^{r}}\ \ \text{as\ \ }%
\left\vert x\right\vert \rightarrow0\ \ ,\ \ B\neq0,\ \ r \geq 0\;. \label{T1E4}%
\end{equation}
The case $r=0$ corresponds to 
\begin{equation}
G\in C^{2}\left(  \mathbb{R}^{3}\right)  \ \ ,\ \ G\text{ bounded near the
origin}\;. \label{T1E3}%
\end{equation}
We remark that in the case (\ref{T1E3}) the family of potentials (\ref{T1E1})
does not have a collision length (or equivalently $\lambda_{\varepsilon}=0$, $T_{BG} = +\infty$).
On the other hand, in the case (\ref{T1E4}) the collision length is%
\begin{equation}
\lambda_{\varepsilon}=\varepsilon^{\frac{1}{r}} \label{T1E5}%
\end{equation}
and the Boltzmann-Grad time scale is then (cf.\,(\ref{BG}))%
\begin{equation}
T_{BG}=\frac{1}{\varepsilon^{\frac{2}{r}}} . \label{T1E6}%
\end{equation}

\subsubsection{Kinetic time scales.}

We now study the properties of the function $\sigma\left(  T;\varepsilon
\right)  $ in (\ref{S4E8a}) and compare the time scale $T_{L}$ defined 
by means of (\ref{S4E8}) with $T_{BG}$ given by \eqref{T1E6}.
\begin{theorem}
\label{GDiffScales} Consider the family of potentials (\ref{T1E1}) with
$G \in \calC_s$, $s > 1/2$ and satisfying (\ref{T1E4}). 
Suppose that the corresponding Holtsmark field
defined in Section \ref{Holtsm} is spatially homogeneous. 
\begin{itemize}
\item[(i)] \ If $s>1$ and $r>1$, then $\lim\sup_{\varepsilon\rightarrow
0}\sigma\left(  T_{BG};\varepsilon\right)  \leq\delta\left(  M\right)$ with
$\delta\left(  M\right)  \rightarrow0$ as $M\rightarrow\infty.$
\item[(ii)] If $s>1$, then $T_{L}\sim\frac
{1}{4\pi B^2 \varepsilon^{2}\left\vert \log\left(  \varepsilon\right)  \right\vert }$
as $\varepsilon\rightarrow0$ if $r=1$ and $T_{L}\sim\frac{C}{\varepsilon^{2}}$
for some $C>0$
as $\varepsilon\rightarrow0$ if $r<1$. In both cases $T_{L}\ll
T_{BG}$ as $\varepsilon\rightarrow0.$
\item[(iii)] Suppose that $s=1.$ If $r>1$ we have $\lim\sup_{\varepsilon\rightarrow
0}\sigma\left(  T_{BG};\varepsilon\right)  \leq\delta\left(  M\right)$ with
$\delta\left(  M\right)  \rightarrow0$ as $M\rightarrow\infty.$ If $r=1$ we
obtain $T_{L}\sim\frac{C_{1}}{\varepsilon^{2}\left\vert \log\left(  \varepsilon\right)
\right\vert }$ for some $C_{1}>0$ and therefore $T_{L}\ll T_{BG}$ as $\varepsilon\rightarrow0.$ 
If $r<1$ we
obtain $T_{L}\sim\frac
{C_{2}}{\varepsilon^{2}\left\vert\log\left(  \varepsilon \right) \right\vert  }$ for some
$C_{2}>0$ and therefore $T_{L}\ll T_{BG}$
as $\varepsilon\rightarrow0.$
\item[(iv)] Suppose that $s<1.$ If $r+2s>3$ we have $\lim\sup_{\varepsilon
\rightarrow0}\sigma\left(  T_{BG};\varepsilon\right)  \leq\delta\left(
M\right) $ with $\delta\left(  M\right)  \rightarrow0$ as $M\rightarrow
\infty.$ If $r+2s<3$ we obtain $T_{L}\sim\frac{C_{0}}{\varepsilon^{\frac{2}{3-2s}}}$ as
$\varepsilon\rightarrow0$ and then $T_{L}\ll T_{BG}=\frac{1}{\varepsilon
^{\frac{2}{r}}}$ as $\varepsilon\rightarrow0.$ If $r+2s=3$ we obtain that
$T_{L}$ and $T_{BG}$ are comparable as $\varepsilon\rightarrow0.$
\end{itemize}
\end{theorem}

\begin{proof}
We will assume in all the proof that $v=\left(  1,0,0\right)  .$ 
We use the splitting (\ref{S4E4}) which becomes here
$\Phi\left(  x,\varepsilon\right)  =\Phi_{B}\left(  x,\varepsilon\right)
+\Phi_{L}\left(  x,\varepsilon\right)$ with
\begin{equation}
\Phi_{B}\left(  x,\ep\right)  =\ep G\left(   |x|\right)
\eta\left(  \frac{|x|}{M \lambda_\ep}\right)  \ \ ,\ \ \Phi
_{L}\left(  x,\ep\right)  =\ep G\left(   |x|\right)  \left[
1-\eta\left(  \frac{ |x| }{M\lambda_\ep}\right)  \right] \;. \label{T1E7}%
\end{equation}

\medskip
\noindent
{\em Proof of (i).}
\smallskip

\noindent
Suppose that $s>1.$ Then using \eqref{S4E8a} and the fact that $\left\vert \theta_{\bot}\right\vert \leq1$ 
we have (in a similar way as in the proof of Theorem \ref{ProofGenPow})%
\be\sigma\left(  T;\varepsilon\right)   =\sup_{\left\vert \theta\right\vert
=1}\int_{\mathbb{R}^{3}}d\xi\left(  \theta\cdot\int_{0}^{T}\nabla_{x}\Phi
_{L}\left(  vt-\xi,\varepsilon\right)  dt\right)  ^{2} \leq J_1 + J_2\label{P5E1}
\ee
where
\begin{align}
& J_1 =  C\varepsilon^{2}\int_{\mathbb{R}^{3}}d\xi
\left\vert \xi_{\bot}\right\vert
^{2}\left(  \int_{0}^{T}dt\,\frac{\chi_{\left\{  \left\vert vt-\xi\right\vert
>1\right\}  }}{\left\vert vt-\xi\right\vert ^{s+2}}\right)  ^{2}\;,\nonumber \\%
& J_2 = C\varepsilon^{2}\int_{\mathbb{R}^{3}}d\xi
\left\vert \xi_{\bot}\right\vert
^{2}\left(  \int_{0}^{T}dt\,\frac{\chi_{\left\{  M\varepsilon^{\frac{1}{r}%
}\leq\left\vert vt-\xi\right\vert \leq1\right\}  }}{\left\vert
vt-\xi\right\vert ^{r+2}}\right)  ^{2}\nonumber
\end{align}
for some $C>0$.

We estimate $J_{1}$ as%
\begin{align}
J_{1}  & \leq C\varepsilon^{2}\int_{\mathbb{R}^{3}}d\xi\left\vert \xi_{\bot
}\right\vert ^{2}\left(  \int_{0}^{T}dt\frac{\chi_{\left\{  \left\vert
\xi-vt\right\vert >1\right\}  }}{\left(  \left(  \xi_{1}-t\right)
^{2}+\left\vert \xi_{\bot}\right\vert ^{2}\right)  ^{\frac{s_{\ast}+2}{2}}}\right)
^{2}\nonumber\\
&  \leq C\varepsilon^{2}\int_{-T}^{T}dt_{1}\int_{-T}^{T}H_{1}\left(  t\right)
dt=CT\varepsilon^{2}\int_{-T}^{T}H_{1}\left(  t\right)  dt \label{P5E2}%
\end{align}
where $s_{\ast}:=\min\{s ,2\}$ and 
\begin{align}
H_{1}\left(  t\right)   &  :=\int_{\mathbb{R}^{2}}d\xi_{\bot}\left\vert \xi_{\bot
}\right\vert ^{2}\int_{-\infty}^{\infty}d\xi_{1}\frac{\chi_{\left\{
\left\vert \xi\right\vert >1\right\}  }\chi_{\left\{  \left\vert
\xi-vt\right\vert >1\right\}  }}{\left(  \left(  \xi_{1}\right)
^{2}+\left\vert \xi_{\bot}\right\vert ^{2}\right)  ^{\frac{s_{\ast}+2}{2}}\left(
\left(  \xi_{1}-t\right)  ^{2}+\left\vert \xi_{\bot}\right\vert ^{2}\right)
^{\frac{s_{\ast}+2}{2}}}\nonumber\\
%&  =\int_{\mathbb{R}^{2}}\frac{d\xi_{\bot}}{\left\vert \xi_{\bot}\right\vert
%^{2s+1}}\int_{-\infty}^{\infty}\frac{\chi_{\left\{  \left(  \eta_{1}\right)
%^{2}+1>\frac{1}{\left\vert \xi_{\bot}\right\vert }\right\}  }\chi_{\left\{
%\left(  \eta_{1}-\frac{t}{\left\vert \xi_{\bot}\right\vert }\right)
%^{2}+1>\frac{1}{\left\vert \xi_{\bot}\right\vert }\right\}  }d\eta_{1}%
%}{\left(  \left(  \eta_{1}\right)  ^{2}+1\right)  ^{\frac{s+2}{2}}\left(
%\left(  \eta_{1}-\frac{t}{\left\vert \xi_{\bot}\right\vert }\right)
%^{2}+1\right)  ^{\frac{s+2}{2}}}\nonumber\\
%&  =\frac{1}{\left\vert t\right\vert ^{2s-1}}\int_{\mathbb{R}^{2}}\frac
%{d\eta_{\bot}}{\left\vert \eta_{\bot}\right\vert ^{2s+1}}\int_{-\infty
%}^{\infty}\frac{\chi_{\left\{  \left(  \eta_{1}\right)  ^{2}+1>\frac
%{1}{\left\vert t\right\vert \left\vert \eta_{\bot}\right\vert }\right\}  }%
%\chi_{\left\{  \left(  \eta_{1}-\frac{1}{\left\vert \eta_{\bot}\right\vert
%}\right)  ^{2}+1>\frac{1}{\left\vert t\right\vert \left\vert \eta_{\bot
%}\right\vert }\right\}  }d\eta_{1}}{\left(  \left(  \eta_{1}\right)
%^{2}+1\right)  ^{\frac{s+2}{2}}\left(  \left(  \eta_{1}-\frac{1}{\left\vert
%\eta_{\bot}\right\vert }\right)  ^{2}+1\right)  ^{\frac{s+2}{2}}} . \label{P5E3}%
%\end{align}
&=\frac{1}{\left\vert t\right\vert ^{2s_{\ast}-1}}\int_{\mathbb{R}^{2}}\left\vert
\xi_{\bot}\right\vert ^{2}d\xi_{\bot}\int_{-\infty}^{\infty}\frac
{\chi_{\left\{  \left\vert \xi\right\vert >\frac{1}{t}\right\}  }%
\chi_{\left\{  \left\vert \xi-v\right\vert >\frac{1}{t}\right\}  }d\xi_{1}%
}{\left(  \left(  \xi_{1}\right)  ^{2}+\left\vert \xi_{\bot}\right\vert
^{2}\right)  ^{\frac{s_{\ast}+2}{2}}\left(  \left(  \xi_{1}-1\right)  ^{2}+\left\vert
\xi_{\bot}\right\vert ^{2}\right)  ^{\frac{s_{\ast}+2}{2}}}. \label{P5E3}
\end{align}

In the case $t\geq \frac 1 2$ we estimate the characteristic functions by one. 
%In the limit $t\rightarrow\infty$ we obtain the integral%
%\[
%\int_{\mathbb{R}^{2}}\left\vert \xi_{\bot}\right\vert ^{2}d\xi_{\bot}%
%\int_{-\infty}^{\infty}\frac{d\xi_{1}}{\left(  \left(  \xi_{1}\right)
%^{2}+\left\vert \xi_{\bot}\right\vert ^{2}\right)  ^{\frac{s+2}{2}}\left(
%\left(  \xi_{1}-1\right)  ^{2}+\left\vert \xi_{\bot}\right\vert ^{2}\right)
%^{\frac{s+2}{2}}}%
%\]

Suppose first that $\left\vert \xi_{\bot}\right\vert \geq1.$ We split the
integral in $\xi_{1}$ in the regions $\left\vert \xi_{1}\right\vert
\leq\left\vert \xi_{\bot}\right\vert $ and $\left\vert \xi_{1}\right\vert
>\left\vert \xi_{\bot}\right\vert .$ The resulting contribution to the
integral in $\xi_{1}$ would be of order %$\frac{\left\vert \xi_{\bot}\right\vert }{\left\vert \xi_{\bot}\right\vert ^{2\left(  s+2\right)  }%}=
$\frac{1}{\left\vert \xi_{\bot}\right\vert ^{2s_{\ast}+3}}$ in the first region and
similar in the second region. This gives an integrable contribution in the
region $\left\vert \xi_{\bot}\right\vert \geq1.$ Suppose now that $\left\vert
\xi_{\bot}\right\vert <1.$ We first estimate the integral%
\[
\int_{-\infty}^{\infty}\frac{d\xi_{1}}{\left(  \left(  \xi_{1}\right)
^{2}+\left\vert \xi_{\bot}\right\vert ^{2}\right)  ^{\frac{s_{\ast}+2}{2}}\left(
\left(  \xi_{1}-1\right)  ^{2}+\left\vert \xi_{\bot}\right\vert ^{2}\right)
^{\frac{s_{\ast}+2}{2}}}%
\]
for $\left\vert \xi_{\bot}\right\vert $ small. We separate the regions close
to $\xi_{1}=0$ and $\xi_{1}=1.$ The rest gives a bounded contribution. The
contributions near these two points are similar and can be bounded by:%
\[
\int_{-\infty}^{\infty}\frac{d\xi_{1}}{\left(  \left(  \xi_{1}\right)
^{2}+\left\vert \xi_{\bot}\right\vert ^{2}\right)  ^{\frac{s_{\ast}+2}{2}}}\leq
\frac{C}{\left\vert \xi_{\bot}\right\vert ^{s_{\ast}+1}}.
\]
Then, the contribution to the integral $\int_{|\xi_{\bot}|< 1} (\dots) \left\vert
\xi_{\bot}\right\vert ^{2}d\xi_{\bot}$ is bounded by $\int_{|\xi_{\bot}|< 1}\frac
{d\xi_{\bot}}{\left\vert \xi_{\bot}\right\vert ^{s_{\ast}-1}}<\infty $. 
Therefore, for $t \geq \frac 1 2$, we have that:
\[
0\leq H_{1}\left(  t\right)  \leq\frac{C}{\left\vert t\right\vert ^{2s_{\ast}-1}}.
\]

Suppose now that $t <\frac 1 2.$ Then, since $\left\vert
\xi\right\vert >1$ and $\left\vert \xi-vt\right\vert >1$ we obtain that:%
\[
\left(  \left(  \xi_{1}\right)  ^{2}+\left\vert \xi_{\bot}\right\vert
^{2}\right)  ^{\frac{s_{\ast}+2}{2}}\left(  \left(  \xi_{1}-t\right)  ^{2}+\left\vert
\xi_{\bot}\right\vert ^{2}\right)  ^{\frac{s_{\ast}+2}{2}}\geq C\left[  1+\left(
\left(  \xi_{1}\right)  ^{2}+\left\vert \xi_{\bot}\right\vert ^{2}\right)
^{s_{\ast}+2}\right]
\]
whence we obtain the following estimate:
\begin{align*}
H_{1}\left(  t\right)    & \leq C\int_{\mathbb{R}^{2}}d\xi_{\bot} \int_{\R} d\xi_1 \frac{\left\vert \xi_{\bot}\right\vert ^{2}%
}{\left[  1+\left(  \left(  \xi_{1}\right)  ^{2}+\left\vert \xi_{\bot
}\right\vert ^{2}\right)  ^{s_{\ast}+2}\right]  }\leq C<\infty
\end{align*}
since $s_{\ast}>1>\frac{1}{2}$. Hence, we obtain
\[
0\leq H_{1}\left(  t\right)  \leq\frac{C}{1+\left\vert t\right\vert
^{2s_{\ast}-1}}
\]
and, using (\ref{P5E2}), we get
\begin{equation}
J_{1}\leq CT\varepsilon^{2}\int_{-T}^{T}\frac{dt}{1+\left\vert t\right\vert
^{2s_{\ast}-1}}\leq CT\varepsilon^{2}\label{P5E4}
\end{equation}
if $s>1$.

We have several possibilities for $J_{2}$ depending on the values of $r.$
Here we suppose that $r>1.$ A computation similar to the one yielding
(\ref{P5E2}) gives%
\begin{equation}
J_{2}\leq CT\varepsilon^{2}\int_{-T}^{T}H_{2}\left(  t\right)  dt
\label{P5E4a}%
\end{equation}
with%
\[
H_{2}\left(  t\right)  :=\int_{\mathbb{R}^{2}}d\xi_{\bot}\left\vert \xi_{\bot}\right\vert
^{2}\int_{-\infty}^{\infty}d\xi_{1}\frac{\chi_{\left\{  M\varepsilon
^{\frac{1}{r}}\leq\left\vert \xi\right\vert \leq1\right\}  }\chi_{\left\{
M\varepsilon^{\frac{1}{r}}\leq\left\vert \xi-vt\right\vert \leq1\right\}  }%
}{\left(  \left(  \xi_{1}\right)  ^{2}+\left\vert \xi_{\bot}\right\vert
^{2}\right)  ^{\frac{r+2}{2}}\left(  \left(  \xi_{1}-t\right)  ^{2}+\left\vert
\xi_{\bot}\right\vert ^{2}\right)  ^{\frac{r+2}{2}}}\;.
\]
If $\left\vert t\right\vert \geq2$
we have 
\begin{align*}
&  \frac{\chi_{\left\{  M\varepsilon^{\frac{1}{r}}\leq\left\vert
\xi\right\vert \leq1\right\}  }\chi_{\left\{  M\varepsilon^{\frac{1}{r}}%
\leq\left\vert \xi-vt\right\vert \leq1\right\}  }}{\left(  \left(  \xi
_{1}\right)  ^{2}+\left\vert \xi_{\bot}\right\vert ^{2}\right)  ^{\frac
{r+2}{2}}\left(  \left(  \xi_{1}-t\right)  ^{2}+\left\vert \xi_{\bot
}\right\vert ^{2}\right)  ^{\frac{r+2}{2}}}\\
&  \leq C\frac{\chi_{\left\{  M\varepsilon^{\frac{1}{r}}\leq\left\vert
\xi\right\vert \leq1\right\}  }\chi_{\left\{  M\varepsilon^{\frac{1}{r}}%
\leq\left\vert \xi-vt\right\vert \leq1\right\}  }}{\left\vert t\right\vert
^{r+2}}\left[  \frac{1}{\left(  \left(  \xi_{1}\right)  ^{2}+\left\vert
\xi_{\bot}\right\vert ^{2}\right)  ^{\frac{r+2}{2}}}+\frac{1}{\left(  \left(
\xi_{1}-t\right)  ^{2}+\left\vert \xi_{\bot}\right\vert ^{2}\right)
^{\frac{r+2}{2}}}\right]
\end{align*}
whence the contribution of this term to the integral in (\ref{P5E4a}) can be
estimated as 
\begin{align*}
&  \int_{\left[  -T,T\right]  \setminus\left[  -2,2\right]  }H_{2}\left(
t\right)  dt\\
&  \leq C\int_{\left[  -T,T\right]  \setminus\left[  -2,2\right]  }\frac
{dt}{\left\vert t\right\vert ^{r+2}}\int_{\mathbb{R}^{2}}d\xi_{\bot}\left\vert \xi_{\bot
}\right\vert ^{2}\int_{-\infty}^{\infty}\frac{\chi_{\left\{
M\varepsilon^{\frac{1}{r}}\leq\left\vert \xi\right\vert \leq1\right\}  }%
d\xi_{1}}{\left(  \left(  \xi_{1}\right)  ^{2}+\left\vert \xi_{\bot
}\right\vert ^{2}\right)  ^{\frac{r+2}{2}}}\leq C\int_{\left\{  M\varepsilon
^{\frac{1}{r}}\leq\left\vert \xi\right\vert \leq1\right\}  }\frac{d^{3}\xi
}{\left\vert \xi\right\vert ^{r}}%
\end{align*}
which implies%
\begin{equation}
\int_{\left[  -T,T\right]  \setminus\left[  -2,2\right]  }H_{2}\left(
t\right)  dt\leq C\left[  \left(  \frac{1}{M\varepsilon^{\frac{1}{r}}}\right)
^{r-3+\delta}+1\right]  \label{P5E7a}%
\end{equation}
(where the coefficient $\delta>0$ has been introduced to include the case $r=3$
in which the integral diverges logarithmically).
On the other hand the contribution to the integral in (\ref{P5E4a}) due to the
region $\left\{  \left\vert t\right\vert <2\right\}  $ can be estimated as%
\begin{align*}
&  \int_{\left[  -T,T\right]  \cap\left[  -2,2\right]  }H_{2}\left(  t\right)
dt\\
&  \leq C\int_{\left\{  M\varepsilon^{\frac{1}{r}}\leq\left\vert
\xi\right\vert \leq 1\right\}  }\frac{d\xi_{\bot}}{\left\vert \xi_{\bot
}\right\vert ^{r}}\int_{-\infty}^{\infty}\frac{d\xi_{1}}{\left(  \left(
\xi_{1}\right)  ^{2}+\left\vert \xi_{\bot}\right\vert ^{2}\right)
^{\frac{r+2}{2}}}\leq \frac{C}{\left(  M\varepsilon^{\frac{1}{r}}\right)
^{2\left(  r-1\right)  }}\;.
\end{align*}
Combining this inequality with (\ref{P5E7a}) we obtain%
\[
\int_{-T}^{T}H_{2}\left(  t\right)  dt\leq C\left[  \frac{1}{\left(
M\varepsilon^{\frac{1}{r}}\right)  ^{2\left(  r-1\right)  }}+1\right] .
\]
Using now (\ref{P5E4a}), as well as $r>1$, we obtain%
\begin{equation}
J_{2}\leq 
\frac{CT\varepsilon^{2}}{\left(  M\varepsilon^{\frac{1}{r}}\right)  ^{2\left(
r-1\right)  }}=\frac{CT\varepsilon^{\frac{2}{r}}}{M^{2\left(  r-1\right)  }}.
\label{P5E8}
\end{equation}

Thanks to (\ref{P5E1}), (\ref{P5E4}), (\ref{P5E8}) as well as (\ref{T1E6}) we obtain%

\[
\sigma\left(  T_{BG};\varepsilon\right)  \leq CT_{BG}\varepsilon^{2}%
+\frac{CT_{BG}\varepsilon^{\frac{2}{r}}}{M^{2\left(  r-1\right)  }}=
C\varepsilon^{2\left(  1-\frac{1}{r}\right)  }+\frac{C}{M^{2\left(
r-1\right)  }}
\]
so that item (i) in Theorem \ref{GDiffScales} is proved.

\medskip
\noindent
{\em Proof of (ii).}
\smallskip

\noindent
Suppose now that $s>1$ and $r\leq1.$ We can then use formula (\ref{P5E4}), but the above
estimate for $J_{2}$ is not enough. Actually we need to approximate the
integral%
\[
\int_{\mathbb{R}^{3}}d\xi\left(  \theta\cdot\int_{0}^{T}\nabla_{x}\Phi
_{L}\left(  vt-\xi,\varepsilon\right)  \chi_{\left\{ M\varepsilon
^{\frac{1}{r}}\leq\left\vert vt-\xi\right\vert \leq1\right\}  }dt\right)
^{2}%
\]
as $\varepsilon\rightarrow0$ for $T$ large. 
Proceeding as above, we obtain that the integral is approximated by%
\begin{equation}
\ep^2B^2\int_{0}%
^{T}dt_{1}\int_{-t_{1}}^{T-t_{1}}W\left(  t\right)  dt \label{P5E9}%
\end{equation}
where%
\begin{align*}
W\left(  t\right)  :=\int_{\mathbb{R}^{2}}d\xi_{\bot}(\theta_\perp \cdot \xi_\perp)^2\int_{-\infty}^{\infty
}d\xi_{1}\frac{\chi_{\left\{  M\varepsilon
^{\frac{1}{r}}\leq\left\vert \xi\right\vert \leq1\right\}  }\,\chi_{\left\{
M\varepsilon^{\frac{1}{r}}\leq\left\vert \xi-vt\right\vert
\leq1\right\}  }}{\left(  \left(  \xi
_{1}\right)  ^{2}+\left\vert \xi_{\bot}\right\vert ^{2}\right)  ^{\frac
{r+2}{2}}\left(  \left(  \xi_{1}-t\right)  ^{2}+\left\vert \xi_{\bot
}\right\vert ^{2}\right)  ^{\frac{r+2}{2}}}
\;.
\end{align*}

We are interested in computing (\ref{P5E9}) for $T\gg1.$ Note then that
 most of the contribution is due to the region $t_{1}\in\left[
L,T-L\right]  $ with $L$ large:%
\begin{equation}
\int_{0}^{T}dt_{1}\int_{-t_{1}}^{T-t_{1}}W\left(  t\right)  dt=\int_{L}%
^{T-L}dt_{1}\int_{-t_{1}}^{T-t_{1}}W\left(  t\right)  dt+O\left(  L\right)
\int_{-\infty}^{\infty}W\left(  t\right)  dt\;. \label{P6E2}%
\end{equation}
Moreover, it turns out that $\int_{-\infty}^{\infty}\left\vert
W\left(  t\right)  \right\vert dt<\infty.$ The main contribution to
the first integral on the right-hand side of (\ref{P6E2}) is due to the strip $\left\{
\left\vert t\right\vert <L'\right\}  $ where we can assume that $L'<L.$ We then
get
\begin{equation}
\int_{0}^{T}dt_{1}\int_{-t_{1}}^{T-t_{1}}W\left(  t\right)  dt=T\left[
1+o\left(  1\right)  \right]  \int_{-\infty}^{\infty}W\left(  t\right)  dt
\label{P6E2a}%
\end{equation}
where $o\left(  1\right)  \rightarrow0$ as $T\rightarrow\infty.$ We have%
\begin{align}
&  \int_{-\infty}^{\infty}W\left(  t\right)  dt\nonumber\\
&  =\int_{\mathbb{R}^{2}}d\xi_{\bot
}(\theta_\perp \cdot \xi_\perp)^2\int_{-\infty}^{\infty}d\xi_1\frac{\chi_{\left\{ M\varepsilon^{\frac{1}{r}}
\leq\left\vert \xi\right\vert \leq1\right\}  }}{\left(  \left(  \xi_{1}\right)  ^{2}+\left\vert
\xi_{\bot}\right\vert ^{2}\right)  ^{\frac{r+2}{2}}}\int_{-\infty}^{\infty
}dt\frac{\chi_{\left\{ M^2\varepsilon^{\frac{2}{r}}%
\leq t^2+\left\vert \xi_\perp\right\vert^2 \leq1\right\}  }}{\left(  t^{2}+\left\vert \xi_{\bot}\right\vert ^{2}\right)
^{\frac{r+2}{2}}}\;.\label{P6E3}%
\end{align}

We must deal now separately with the cases $r=1$ and $r<1.$ In the first case, 
we are going to show that \eqref{P6E3}
diverges logarithmically as $\varepsilon\rightarrow0.$ 
The main contribution is due to the
region where $\left\vert \xi_{\bot}\right\vert \rightarrow0$. 
The characteristic function $\chi_{\left\{ M^2\varepsilon^{2}%
\leq t^2+\left\vert \xi_\perp\right\vert^2 \leq1\right\}  }$
can be replaced by $ \chi_{\left\{  \left\vert t\right\vert \leq1\right\}}$ by making an error of order $O\left(  \frac{M\varepsilon}{\left\vert \xi_{\bot}\right\vert ^{3}}\right)$ in the integral $\int_{-\infty}^{+\infty} dt \cdots$, and the resulting
contribution to the whole integral is bounded by a constant:
\begin{align*}
 \int_{-\infty}^{\infty}W\left(  t\right)  dt
=\int_{\mathbb{R}^{2}}d\xi_{\bot
} \frac{(\theta_\perp \cdot \xi_\perp)^2}{|\xi_\perp|^2}
\int_{-\infty}^{\infty}d\xi_1\frac{\chi_{\left\{ M\varepsilon
\leq\left\vert \xi\right\vert \leq1\right\}  }}{\left(  \left(  \xi_{1}\right)  ^{2}+\left\vert
\xi_{\bot}\right\vert ^{2}\right)  ^{\frac{3}{2}}}\int_{-\frac{1}{\left\vert \xi_{\bot}\right\vert }}^{\frac
{1}{\left\vert \xi_{\bot}\right\vert }}\frac{dt}{\left(  t^{2}+1\right)
^{\frac{3}{2}}}+ O(1)\;.
\end{align*}
For $\left\vert \xi_{\bot}\right\vert \rightarrow0$, using $\int_{-\infty}^{\infty}\frac{dt}{\left(  t^{2}+1\right)
^{\frac{3}{2}}}=2$, we can approximate with
\begin{align*}
 & 2\int_{\mathbb{R}^{2}}d\xi_{\bot}
 \frac{(\theta_\perp \cdot \xi_\perp)^2}{|\xi_\perp|^2}
\int_{-\infty}^{\infty}\chi_{\left\{ M\varepsilon\leq\left\vert
\xi\right\vert \leq1\right\}  }\frac{d\xi_{1}}{\left(  \left(  \xi_{1}\right)
^{2}+\left\vert \xi_{\bot}\right\vert ^{2}\right)  ^{\frac{3}{2}}}\\
& = 2\int_{\mathbb{R}^{2}}d\xi_{\bot}
 \frac{(\theta_\perp \cdot \xi_\perp)^2}{|\xi_\perp|^4}
\int_{-\infty}^{\infty}%
\chi_{\left\{  \frac{M\varepsilon}{\left\vert \xi_{\bot}\right\vert }%
\leq\sqrt{X^{2}+1}\leq\frac{1}{\left\vert \xi_{\bot}\right\vert }\right\}
}\frac{dX}{\left(  X^{2}+1\right)  ^{\frac{3}{2}}} \\
& \sim 4 \pi\left\vert \theta_{\bot}\right\vert ^{2}\log\left(  \frac{1}{\varepsilon
}\right)
\end{align*}
as $\ep \to 0$. We have then obtained
\begin{equation}
\int_{\mathbb{R}^{3}}d\xi\left(  \theta\cdot\int_{0}^{T}\nabla_{x}\Phi
_{L}\left(  vt-\xi,\varepsilon\right)  \chi_{\left\{ M\varepsilon
\leq\left\vert vt-\xi\right\vert \leq1\right\}  }dt\right)
^{2}\sim4\pi B^2 T\left\vert \theta_{\bot}\right\vert ^{2}\ep^2\log\left(  \frac
{1}{\varepsilon}\right)  \ \ \text{as\ }\varepsilon\rightarrow0 \label{P6E4a}%
\end{equation}
whence the asymptotics in Theorem \ref{GDiffScales}-(ii) for $r=1.$

In the case $r<1$, (\ref{P6E3}) converges to a
finite integral as $\varepsilon\rightarrow0.$ The regions
$\left\{  \left\vert vt-\xi\right\vert >1\right\}  $ and $\left\{  \left\vert
vt-\xi\right\vert \leq1\right\}  $ give contributions of the same order of
magnitude and 
\begin{align}
&  \int_{\mathbb{R}^{3}}d\xi\left(  \theta\cdot\int_{0}^{T}\nabla_{x}\Phi
_{L}\left(  vt-\xi,\varepsilon\right)  dt\right)  ^{2}\nonumber\\
&  \sim \ep^2 B^2 T\int_{\mathbb{R}^{2}%
}d\xi_{\bot}(\theta_\perp \cdot \xi_\perp)^2
\int_{-\infty}^{\infty}d\xi_1
\frac{\chi_{\left\{  \left\vert \xi\right\vert
\leq1\right\}  }}{\left(
\left(  \xi_{1}\right)  ^{2}+\left\vert \xi_{\bot}\right\vert ^{2}\right)
^{\frac{r+2}{2}}}\int_{-\infty}^{\infty}dt\frac{\chi_{\left\{ \sqrt{t^2+|\xi_\perp|^2}\leq1\right\}  }}{\left(  t^{2}+\left\vert
\xi_{\bot}\right\vert ^{2}\right)  ^{\frac{r+2}{2}}}+\nonumber\\
&  +\ep^2 A^2 T\int_{\mathbb{R}^{2}%
}d\xi_{\bot}(\theta_\perp \cdot \xi_\perp)^2
\int_{-\infty}^{\infty}d\xi_1
\frac{\chi_{\left\{  \left\vert \xi\right\vert
>1\right\}  }}{\left(
\left(  \xi_{1}\right)  ^{2}+\left\vert \xi_{\bot}\right\vert ^{2}\right)
^{\frac{s+2}{2}}}\int_{-\infty}^{\infty}dt\frac{\chi_{\left\{ \sqrt{t^2+|\xi_\perp|^2} >1\right\}  }}{\left(  t^{2}+\left\vert
\xi_{\bot}\right\vert ^{2}\right)  ^{\frac{s+2}{2}}} \label{P6E4}%
\end{align}
as $\varepsilon\rightarrow0$ for $T$ large. In particular this implies the asymptotics of
$T_{L}$ in Theorem \ref{GDiffScales}-(ii) for $r<1.$

\medskip
\noindent
{\em Proof of (iii).}
\smallskip

\noindent
We assume that $s=1$ and $r>1.$ We can use the decomposition
(\ref{P5E1}) and bound $J_{2}$ using (\ref{P5E8}). Concerning $J_{1},$ we notice that (\ref{P5E2}), (\ref{P5E3}), (\ref{P5E3a}) are
valid for $s\leq1.$ Then 
\[
\sigma\left(  T;\varepsilon\right)  \leq CT\log\left(  T\right)
\varepsilon^{2}+\frac{CT\varepsilon^{\frac{2}{r}}}{M^{2\left(  r-1\right)  }}\;,
\]
therefore \eqref{T1E6} implies
$$
\sigma\left(  T;\varepsilon\right)  \leq
C\log\left(  \frac{1}{\varepsilon^{\frac{2}{r}}%
}\right)  \varepsilon^{2\left(  1-\frac{1}{r}\right)  }
+\frac{C}{M^{2\left(
r-1\right)  }}
$$
which proves the first
statement of (iii) in Theorem \ref{GDiffScales}.

Suppose now that $s=1$ and $r\leq1.$ The case $r=1$ is already
included in the results of Theorem \ref{ProofGenPow}. For $r<1,$ we consider the asymptotics of the quadratic form 
appearing in the definition of $\sigma\left(  T;\varepsilon\right)  .$
Computing the contribution due to the region $\left\{  M\varepsilon
^{\frac{1}{r}}\leq\left\vert vt-\xi\right\vert \leq1\right\}  $ we can
argue as in the proof of the point (ii) above and we obtain a term identical
to the first one on the right-hand side of (\ref{P6E4}). We are left with the contribution due to the region $\left\{  \left\vert vt-\xi
\right\vert >1\right\}  .$ We remark as before that%
\begin{align}
\int_{\mathbb{R}^{3}}d\xi\left(  \theta\cdot\int_{0}^{T}\nabla_{x}\Phi
_{L}\left(  vt-\xi,\varepsilon\right)  \chi_{\left\{  \left\vert
vt-\xi\right\vert >1\right\}  }dt\right)  ^{2} \sim \ep^2A^2\int_{0}^{T}dt_{1}\int
_{-t_{1}}^{T-t_{1}}\tilde{W}\left(  t\right)  dt \label{P6E5}%
\end{align}
where, using that $s=1$, we get
\[
\tilde{W}\left(  t\right)  :=\int_{\mathbb{R}^{2}}d\xi_{\bot}(\theta_\perp \cdot \xi_\perp)^2\int_{-\infty}^{\infty
}d\xi_{1}\frac{\chi_{\left\{  \left\vert \xi\right\vert >1\right\}  }\,\chi_{\left\{
\left\vert \xi-vt\right\vert
>1\right\}  }}{\left(  \left(  \xi
_{1}\right)  ^{2}+\left\vert \xi_{\bot}\right\vert ^{2}\right)  ^{\frac
{3}{2}}\left(  \left(  \xi_{1}-t\right)  ^{2}+\left\vert \xi_{\bot
}\right\vert ^{2}\right)  ^{\frac{3}{2}}}
\;.
\]
Arguing as in the derivation of (\ref{P6E2a}) we see that the main
contribution to (\ref{P6E5}) as $T\rightarrow\infty$ is due to the regions
where $t_{1}\gg1$ and $\left(  T-t_{1}\right)  \gg1.$ However we cannot
derive an approximation like (\ref{P6E2a}) because the integral of $\tilde{W}$
is not finite. Indeed, proceeding as in the proof of (\ref{P5E3a}) we obtain the asymptotics%
\begin{equation}
\tilde{W}\left(  t\right)  \sim\frac{C'}{\left\vert t\right\vert
}\ \ \text{as\ }\left\vert t\right\vert \rightarrow\infty\label{P6E6}%
\end{equation}
with%
\[
C'=\int_{\mathbb{R}^{2}}d\eta_{\bot}\frac
{(\eta_\perp\cdot\theta_\perp)^2}{\left\vert \eta_{\bot}\right\vert ^{5}}\int_{-\infty}^{\infty
}\frac{d\eta_{1}}{\left(  \left(  \eta_{1}\right)  ^{2}+1\right)  ^{\frac
{3}{2}}\left(  \left(  \eta_{1}-\frac{1}{\left\vert \eta_{\bot}\right\vert
}\right)  ^{2}+1\right)  ^{\frac{3}{2}}}\;.
\]
Moreover, $\tilde{W}\left(  t\right)  $ is bounded if $\left\vert t\right\vert
$ is bounded. 
%This provides a cutoff yielding a logarithmic divergence whose form does not affect the leading asymptotic term. 
Writing
\[
\int_{0}^{T}dt_{1}\int_{-t_{1}}^{T-t_{1}}\tilde{W}\left(  t\right)
dt=T^{2}\int_{0}^{1}d\tau_{1}\int_{-\tau_{1}}^{1-\tau_{1}}\tilde{W}\left(
T\tau\right)  d\tau\;,
\]
we obtain that the region where $\left\vert \tau\right\vert \leq\frac{L}{T}$ with $L$ large
(independent of $T$) yields a contribution of order $O(LT).$ On the
other hand, in the region where $\left\vert \tau\right\vert >\frac{L}{T}$ we
can use (\ref{P6E6}). It follows that
\begin{align}
&T^{2}\int_{0}^{1}d\tau_{1}\int_{-\tau_{1}}^{1-\tau_{1}}\tilde{W}\left(
T\tau\right)  d\tau \nonumber\\&
\quad =C'T\int_{0}^{1}d\tau_{1}\int_{-\tau_{1}}^{1-\tau_{1}%
}\chi_{\left\{  \left\vert \tau\right\vert >\frac{L}{T}\right\}  }\frac{d\tau
}{\left\vert \tau\right\vert }+O\left(  T\right)  \sim 2C'T\log\left(
T\right)  +O\left(  T\right)  \text{ as }T\rightarrow\infty\;.
\end{align}
Therefore $\sigma(T;\ep)$ is, up to a multiplicative constant, asymptotic to  $\ep^2 T \log T$
for $T$ large, whence the last statement of (iii) in Theorem \ref{GDiffScales} follows.

\medskip
\noindent
{\em Proof of (iv).}
\smallskip

\noindent
We finally consider the case $\frac{1}{2}<s<1.$ Suppose first that $r+2s>3.$
Then $r>3-2s>1.$ We use the splitting (\ref{P5E1}) and bound $J_{1}$
using (\ref{P5E2}), (\ref{P5E3a}) which are valid for $\frac{1}{2}<s<1.$ We
then obtain%
\[
J_{1}\leq  C\varepsilon^{2}T^{3-2s}\;. 
\]
We estimate $J_{2}$ using (\ref{P5E8}) with $r>1$, which yields
$J_{2}\leq\frac{CT\varepsilon^{\frac{2}{r}}}{M^{2\left(  r-1\right)  }}$.
Using (\ref{T1E6}) we arrive to 
\[
\sigma\left(  T_{BG};\varepsilon\right)  \leq C\varepsilon^{\frac{2\left(
r+2s-3\right)  }{r}}+\frac{C}{M^{2\left(  r-1\right)  }}%
\]
which proves the first statement of case (iv).

\smallskip
Suppose next that $\frac{1}{2}<s<1$ and $r+2s \leq 3.$ If $r>1$ we can use
(\ref{P5E8}) to prove $J_{2}\leq\frac{CT\varepsilon^{\frac{2}{r}}}{M^{2\left(
r-1\right)  }}.$ If $r=1$, (\ref{P6E4a}) yields $J_{2}\leq
CT\varepsilon^{2}\log\left(  \frac{1}{\varepsilon}\right)  .$ If $r<1$ we
argue as in the derivation of (\ref{P6E4}) to obtain $J_{2}\leq CT\varepsilon
^{2}.$  We combine all those estimates in a single formula:%
\begin{equation}
J_{2}\leq CT\left(  \frac{\varepsilon^{\frac{2}{r}}}{M^{2\left(  r-1\right)
}}+\varepsilon^{2}\log\left(  \frac{1}{\varepsilon}\right)  \right)\;.
\label{P6E7}
\end{equation}
On the other hand, we can compute the contribution of the region $\left\{
\left\vert vt-\xi\right\vert >1\right\}  $ to $\sigma\left(  T;\varepsilon
\right)  $ using (\ref{P6E5}) with $\tilde{W}\left(  t\right)  $ replaced by%
\[
\bar{W}\left(  t\right)  :=\int_{\mathbb{R}^{2}}d\xi_{\bot}(\theta_\perp \cdot \xi_\perp)^2\int_{-\infty}^{\infty
}d\xi_{1}\frac{\chi_{\left\{  \left\vert \xi\right\vert >1\right\}  }\,\chi_{\left\{
\left\vert \xi-vt\right\vert
>1\right\}  }}{\left(  \left(  \xi
_{1}\right)  ^{2}+\left\vert \xi_{\bot}\right\vert ^{2}\right)  ^{\frac
{s+2}{2}}\left(  \left(  \xi_{1}-t\right)  ^{2}+\left\vert \xi_{\bot
}\right\vert ^{2}\right)  ^{\frac{s+2}{2}}}
\;.
\]
We have the asymptotics%
\[
\bar{W}\left(  t\right)  \sim\frac{C''}{\left\vert t\right\vert ^{2s-1}%
}\ \ \text{as\ \ }\left\vert t\right\vert \rightarrow\infty
\]
with
\[
C''=\int_{\mathbb{R}^{2}}d\eta_{\bot}\frac
{(\eta_\perp\cdot\theta_\perp)^2}{\left\vert \eta_{\bot}\right\vert ^{2s+3}}\int_{-\infty}^{\infty
}\frac{d\eta_{1}}{\left(  \left(  \eta_{1}\right)  ^{2}+1\right)  ^{\frac
{s+2}{2}}\left(  \left(  \eta_{1}-\frac{1}{\left\vert \eta_{\bot}\right\vert
}\right)  ^{2}+1\right)  ^{\frac{s+2}{2}}}\;,%
\]
whence%
\[
\int_{0}^{T}dt_{1}\int_{-t_{1}}^{T-t_{1}}\bar{W}\left(  t\right)  dt\sim
C''T^{3-2s}\int_{0}^{1}d\tau_{1}\int_{-\tau_{1}}^{1-\tau_{1}}\frac{d\tau
}{\left\vert \tau\right\vert ^{2s-1}}\;.%=C_{0}T^{3-2s}.
\]
Then, (\ref{P6E5}) and (\ref{P6E7}) imply
\begin{equation}
\sigma\left(  T;\varepsilon\right)  \sim \varepsilon^{2}\left(\frac{T}{C_0}\right)^{3-2s}+O\left(
\frac{T\varepsilon^{\frac{2}{r}}}{M^{2\left(  r-1\right)  }}\right)  +O\left(
T\varepsilon^{2}\log\left(  \frac{1}{\varepsilon}\right)  \right).
\label{P6E8}%
\end{equation}
Using that $3-2s>1$ and $r+2s \leq 3$ we obtain that the solution of the equation
$\sigma\left(  T_{L};\varepsilon\right)  =1$ satisfies%
\[
T_{L}\sim\frac{C_{0}}{\varepsilon^{\frac{2}{3-2s}}}\ \ \text{as\ \ }\varepsilon\rightarrow0\;,
\]
whence case (iv) in Theorem \ref{GDiffScales} follows. 
\end{proof}

\subsubsection{Computation of the correlations.}

We now discuss under which conditions the correlations of deflections
in times of order of $T_{L}$ are negligible. We restrict our analysis
to the case $T_{L}\ll T_{BG}$.
We shall use the notation \eqref{eq:DevCorr} for the deflection vector.

\begin{theorem}
\label{CorrG}Suppose that the assumptions of the
Theorem \ref{GDiffScales} hold. 
\begin{itemize}
\item[(i)] Suppose that $s>1$ and $r\leq1.$ Let $T_{L}$ be as in Theorem
\ref{GDiffScales}, case (ii). Then%
\begin{align}
&\mathbb{E}\left(  D\left(  x_{0},v;\tilde{T}_{L}\right)  D\left(
x_{0}+v\tilde{T}_{L},v;\tilde{T}_{L}\right)  \right) \nonumber \\&
 \ll  \sqrt{\mathbb{E}%
\left(  \left(  D\left(  x_{0},v;\tilde{T}_{L}\right)  \right)  ^{2}\right)
\mathbb{E}\left(  \left(  D\left(  x_{0}+v\tilde{T}_{L},v;\tilde{T}%
_{L}\right)  \right)  ^{2}\right)  }\ \ \text{as\ \ }\varepsilon\rightarrow0.
\label{A1E1}%
\end{align}

\item[(ii)] Suppose that $s=1$ and $r\leq1.$ Let $T_{L}$ be as in Theorem
\ref{GDiffScales}, case (iii). Then (\ref{A1E1}) holds.

\item[(iii)] Suppose that $s<1$ and $r+2s\leq3.$ Let $T_{L}$ be as in Theorem
\ref{GDiffScales}, case (iv). Then%
\[
\lim\inf_{\varepsilon\rightarrow0}\frac{\mathbb{E}\left(  D\left(
x_{0},v;\tilde{T}_{L}\right)  D\left(  x_{0}+v\tilde{T}_{L},v;\tilde{T}%
_{L}\right)  \right)  }{\sqrt{\mathbb{E}\left(  \left(  D\left(
x_{0},v;\tilde{T}_{L}\right)  \right)  ^{2}\right)  \mathbb{E}\left(  \left(
D\left(  x_{0}+v\tilde{T}_{L},v;\tilde{T}_{L}\right)  \right)  ^{2}\right)  }%
}>0
\]
for each fixed $h$.
\end{itemize}
\end{theorem}

\begin{proof}
It is similar to the proof of Theorem \ref{CorrPowLaw} and we just
sketch the details. In the case (i) the main contribution to the deflections
is due to the region where $\left\vert \xi-vt\right\vert \leq1$ or at least
the region where $\left\vert \xi-vt\right\vert $ is bounded. On the other hand
the computation of the correlations requires to take into account the
contribution to the integrals of regions where $\left\vert \xi-vt\right\vert $ is large.
The latter are negligible and (\ref{A1E1}) follows. The case (ii) with $r=1$
is similar to the case (i) of Theorem \ref{CorrPowLaw}. If
$r<1$ we can use analogous arguments to show that (\ref{A1E1}) holds, since the contribution of the regions $\left\vert
\xi-vt\right\vert \leq1$ is negligible. Finally, in the case
(iii) we argue as in the case (ii) of Theorem \ref{CorrPowLaw}, since the
largest contribution to the deflections is due to the region $\left\vert
\xi-vt\right\vert >1.$ 
\end{proof}

\subsubsection{Kinetic equations.}
Combining Theorems \ref{GDiffScales} and \ref{CorrG} we can write the kinetic
equations yielding the evolution of the distribution $f$ using the arguments
in Section \ref{GenKinEq}.
We then obtain the following list of cases, assuming that $\Phi\left(  x,\varepsilon\right)
=\varepsilon \,G\left(  \left\vert x\right\vert \right)  $ with $G$ satisfying
(\ref{T1E2})-(\ref{T1E4}) and restricting for simplicity to the case of one single charge.

\begin{itemize}
\item If $s>1$ and $r>1$ we claim that%
\begin{equation}
f_{\ep}\left(  T_{BG}t,T_{BG}x,v\right)  \rightarrow f\left(
t,x,v\right)  \text{ as }\varepsilon\rightarrow0 \label{A2}%
\end{equation}
with $T_{BG}=\frac{1}{\varepsilon^{\frac{2}{r}}}$ where $f$ solves the linear
Boltzmann equation%
\begin{equation}
\left(  \partial_{t}f+v\partial_{x}f\right)  \left(  t,x,v\right)
=\int_{S^{2}}B\left(  v;\omega \right)  \left[  f\left(  t,x,\left\vert
v\right\vert \omega\right)  -f\left(  t,x,v\right)  \right]  d\omega \label{A3}%
\end{equation}
with $B$ as in (\ref{P1E8})-(\ref{P1E8a}).

\item If $s>1$ and $r\leq1$ we have%
\begin{equation}
f_{\ep}\left(  T_{L}t,T_{L}x,v\right)  \rightarrow f\left(
t,x,v\right)  \text{ as }\varepsilon\rightarrow0 \label{A4}%
\end{equation}
where $f$ solves the Landau equation%
\begin{equation}
\left(  \partial_{t}f+v\partial_{x}f\right)  \left(  t,x,v\right)  =\kappa\,\Delta_{v_{\perp}}f\left(  t,x,v \right) 
\label{A5} 
\end{equation}
with $\kappa$ as in \eqref{I1E2} and where $T_L$ is as in Theorem \ref{GDiffScales},
case (ii). 

\item If $s=1$ and $r>1$ we obtain (\ref{A2}) with $T_{BG}=\frac
{1}{\varepsilon^{\frac{2}{r}}}$ where $f$ solves (\ref{A3}) and $B$ is as in
(\ref{P1E8})-(\ref{P1E8a}).

\item If $s=1$ and $r\leq1$ we obtain (\ref{A4}) where $f$ solves (\ref{A5})
and $T_L$ is as in Theorem \ref{GDiffScales},
case (iii). 

\item Suppose that $s<1$ and $r+2s>3.$ Then $r>1$ and we obtain (\ref{A2})
with $T_{BG}=\frac{1}{\varepsilon^{\frac{2}{r}}}$ where $f$ solves the linear
Boltzmann equation (\ref{A3}) with the kernel $B$ given by (\ref{P1E8})-(\ref{P1E8a}).

\item Suppose that $s<1$ and $r+2s<3.$ Then the paths yielding the
trajectories are given by a probability measure with
correlations as described in Section  \ref{ss:CorrCase}. 
If $r+2s=3,$ the trajectories would be given by a measure with correlations plus pointwise large deflections described by the Boltzmann equation.
\end{itemize}

Notice that in all the examples mentioned above where the dynamics is given by
the Boltzmann equation we only need to compute the collision kernel for cross
sections associated to the potential $\frac{1}{\left\vert x\right\vert ^{r}}$
with $r>1.$

We could use similar methods to study more complicated classes of potentials,
for instance 
$\Phi\left(  x,\varepsilon\right)    =\varepsilon^{a_{1}}\Psi_{1}\left(  x\right)
+\varepsilon^{a_{2}}\Psi_{2}\left(  x\right)  $ or
$\Phi\left(  x,\varepsilon\right)   =
\Psi_{1}\left(  \frac{x}{\lambda_{1,\varepsilon}}\right)  +\varepsilon\Psi
_{2}\left(  \frac{x}{\lambda_{2,\varepsilon}}\right)
.$ However, we will not continue with a further analysis of these cases or similar generalizations.

\begin{remark}
\label{MixEq} We have found examples of families of potentials for
which the resulting kinetic equation contains both Boltzmann and terms associated to correlations over
distances of order $T_{L}$. This phenomenon takes place when the
time scale $T_{L}$ associated to the deflections produced by the collective
effect of many scatterers is similar to the Boltzmann-Grad time scale $T_{BG}.$
In the family of potentials (\ref{T1E1})-(\ref{T1E4}) we have obtained that
all the potentials satisfying $r+2s=3,\ \frac{1}{2}< s<1$ yield such evolution. In the case
of Coulomb potentials the Landau term is the only one which appears in the
limit $\varepsilon\rightarrow0,$ due to the presence of a logarithmically
small factor in the Boltzmann type term. It might be possible to
modify the Coulomb potential, replacing it by terms like $\frac{1}{\left\vert
x\right\vert \left(  \log\left(  1+\left\vert x\right\vert \right)  \right)
^{\alpha}}$ to obtain an evolution equation for the tagged particle described
by an equation containing both Boltzmann and Landau terms.
\end{remark}

\begin{remark}
\label{BoSD} 
It is interesting to remark that the fact that the dynamics is described by a
Landau or a Boltzmann equation does not depend only on the decay properties of
the potential but also on the size of the coefficients describing such decay.
Suppose for instance that we consider the family of potentials $\Phi\left(
x,\varepsilon\right)  =\frac{\varepsilon^{r}}{\left\vert x\right\vert ^{r}%
}+\frac{\varepsilon^{\alpha}}{\left\vert x\right\vert ^{s}}$ where
$r>1,\ \alpha+2s>3,\ \frac{1}{2}<s<1.$ Then, $T_{BG}=\frac{1}{\varepsilon^{2}%
}.$ Arguing as in the derivation of (4.4) it might be seen that the possible
contribution to the Landau time of $\frac{\varepsilon^{r}}{\left\vert
x\right\vert ^{r}}$\ would be much larger than $\frac{1}{\varepsilon^{2}}.$ On
the other hand the Landau time scale associated to $\frac{\varepsilon^{\alpha
}}{\left\vert x\right\vert ^{s}}$ is of order $\frac{1}{\left(  \varepsilon
\right)  ^{\frac{2\alpha}{3-2s}}}$ which is much larger than $\frac
{1}{\varepsilon^{2}}$ since $\frac{2\alpha}{3-2s}>2.$ Therefore the dynamics
of $f\left(  t,x,v\right)  $ is given by the linear Boltzmann equation with
the cross section associated to the scattering potential $\frac{1}{\left\vert
x\right\vert ^{r}}$ in spite of the fact that the potential $\Phi\left(
x;\varepsilon\right)  $ behaves for large values of $\left\vert x\right\vert $
as $\frac{\varepsilon^{\alpha}}{\left\vert x\right\vert ^{s}}$ with $s<1.$

\end{remark}

\begin{remark}
Let us consider the domain of influence as defined in Remark \ref{RangeInter}.
For the potentials with the form (\ref{T1E1}) considered in this section, assuming
$v = (1,0,0)$ and that the tagged particle is in the origin at time zero,
we obtain that in the case $s=1,\ r<1$ the domain of influence are the scatterers 
located in $\left(  x_{1},x_{\bot}\right)  $ with
$x_{1}\in\left[  0,T_{L}\right]  $ and $k_{1}\leq\left\vert x_{\bot
}\right\vert \leq\frac{T_{L}}{k_{1}},$ where $k_{1}$ is a large number. If $s>1,\ r<1,$ the domain of influence is given by
$x_{1}\in\left[  0,T_{L}\right]  $ and
$\left\vert x_{\bot}\right\vert \leq k_{1}$. Finally
if $s<1$ and $r<3-2s,$  then $\left\vert x_{\bot
}\right\vert \leq k_{1}T_{L}.$
\end{remark}

\begin{remark}
In the two-dimensional case, similar computations
lead to kinetic equations as established in this section, with the following differences. 
The Boltzmann-Grad time scale is $T_{BG} = \frac{1}{\ep^{\frac{1}{r}}}$. The critical value
of $r$ separating the Boltzmann and the Landau behaviour is $r= \frac{1}{2}$
(instead of $1$) for $s \geq \frac{1}{2}$. For $0<s<\frac{1}{2}$ and $r>1-s$, a linear Boltzmann
equation is expected to hold, while for $r \leq 1-s$ we find that $T_L$
grows as $\ep^{-\frac{1}{1-s}}$ and the correlations do not vanish on the macroscopic scale.
\end{remark}

\subsection{Two different ways of deriving Landau equations with finite range
potentials.\label{DiffLandau}}

We now discuss two classes of potentials which can be studied with the
formalism developed above yielding in both cases Landau kinetic equations, but
for which the interaction has a very different form. 
We consider 
\begin{equation}
\Phi\left(  x,\varepsilon\right)  =\varepsilon\Psi\left(  \frac{x}
{L_{\varepsilon}}\right)  \ \ \text{with\ }L_{\varepsilon}\geq 1,
\label{S5E3}%
\end{equation}%
\begin{equation}
\Phi\left(  x,\varepsilon\right)  =\varepsilon\Psi\left(  \frac{x}%
{L_{\varepsilon}}\right)  \ \ \text{with\ }L_{\varepsilon}%
\rightarrow0\text{ as }\varepsilon\rightarrow 0 \label{S5E4}
\end{equation}
and we assume that the functions $\Psi$ are bounded and smooth in
$\mathbb{R}^{3}$, so that the collision length associated to
$\left\{  \Phi\left(  x,\varepsilon\right)  ;\varepsilon
>0\right\}  $ is $\lambda_{\varepsilon}=0$ and $$T_{BG}=\infty\;.$$ For the sake
of definiteness we assume also that the potentials $\Psi\left(  y\right)
=\Psi\left(  \left\vert y\right\vert \right)  $ are compactly supported or
decay very fast (say exponentially) as $\left\vert y\right\vert \rightarrow
\infty,$ but the results described in the following would remain valid if
$\Psi\left(  \left\vert y\right\vert \right)  \sim\frac{1}{\left\vert
y\right\vert ^{s}}$ at least if $s>1$. Rigorous derivations of a Landau
equation in the case (\ref{S5E3}) have been obtained in \cite{Pi81} if $L_{\varepsilon}\rightarrow\infty$, as $\varepsilon\rightarrow 0$, and in \cite{DGL,KP, KR} for $L_{\varepsilon}$ of order $1$.
The case (\ref{S5E4}) has been considered in \cite{DR}. 

Notice that there is a difference between the dynamics of the
tagged particle in the cases (\ref{S5E3}),\ (\ref{S5E4}). Indeed, in the case
(\ref{S5E3}) the tagged particle interacts at any time with a large number of
scatterers (of the order of $L_{\varepsilon}^{3}$). These
interactions are very weak, but the randomness in the distribution of
scatterers has as a consequence that the force acting on the tagged particle
at a given time is a random variable and this yields, under suitable
assumptions on $\varepsilon,\ L_{\varepsilon}$ a diffusive dynamics for the
velocity, or more precisely a Landau equation for $f.$
In the case of potentials as in (\ref{S5E4}) the tagged particle does not
interact with any scatterer during most of the time, but meets one scatterer
in times of order $\frac{1}{L_{\varepsilon}^{2}}$ much in
the same manner as in the Boltzmann-Grad limit. The main difference with the
Boltzmann-Grad case is that in these collisions the velocity of the tagged
particle is deflected a very small amount. The accumulation of many
independent random deflections yields also a diffusive behaviour for the
velocity of the tagged particle due to the central limit theorem.
In spite of these differences, we obtain the same type of Landau
equation in both cases. This is due to the fact that the relevant variable is the deflection of the
particle velocity. In the case in which these deflections
are small, they are additive, and there is no important difference if many
scatterers act on the particle at a given time or if only one of them
acts rarely.

\begin{theorem} We have the following cases.
\label{LandTimes}
\begin{itemize}
\item[(i)] Suppose that we consider potentials with the form
(\ref{S5E3}) with $\varepsilon^{-\frac{2}{5}}\ll L_{\varepsilon}%
\ll\varepsilon^{-\frac{2}{3}}\ \ $as\ \ $\varepsilon\rightarrow0.$ Then
$T_{L}\sim\frac{1}{\varepsilon\left(  L_{\varepsilon
}\right)  ^{\frac{3}{2}}}\rightarrow\infty$ as $\varepsilon\rightarrow0.$
\item[(ii)] Suppose that we consider potentials with the form (\ref{S5E3}) with
$\varepsilon\left(  L_{\varepsilon}\right)  ^{\frac{5}{2}}\rightarrow
C_{\ast}\in\left(  0,\infty\right)  .$ Then $T_{L}%
\sim\frac{1}{\varepsilon\left(  L_{\varepsilon}\right)  ^{\frac{3}{2}}%
}=\frac{L_{\varepsilon}}{C_{\ast}}$ as $\varepsilon\rightarrow0.$
\item[(iii)] Suppose that we consider potentials with the form (\ref{S5E3}) or
(\ref{S5E4}) and that $L_{\varepsilon}\ll\varepsilon^{-\frac{2}{5}}$
(notice that this includes the case (\ref{S5E4})). Then $T_{L}\sim\frac{1}{\varepsilon^{2}\left( L_{\varepsilon}\right)  ^{4}}$ as
$\varepsilon\rightarrow0.$
\end{itemize}
\end{theorem}

\begin{remark}
The condition $L_{\varepsilon}\ll\varepsilon^{-\frac{2}{3}}$ which is
assumed in Theorem\ \ref{LandTimes} is required in order to have a kinetic
limit. It is possible to have potentials for which this condition fails and where $T_{L}\leq C$. Then (\ref{KinLim})
would also fail.
\end{remark}

\begin{proof}
Using (\ref{S5E3}),\ (\ref{S5E4}) and the fact that $\Phi = \Phi_L$, we can write the function $\sigma\left(
T;\varepsilon\right)  $ in (\ref{S4E8a}) as%
\bea
\sigma\left(  T;\varepsilon\right) & =&\varepsilon^{2}\sup_{\left\vert
\theta\right\vert =1}\int_{\mathbb{R}^{3}}d\xi\left(  \theta\cdot\int_{0}%
^{T}\nabla_{x}\Psi\left(  \frac{vt-\xi}{L_{\varepsilon}}\right)  dt\right)
^{2} \nn\\
& = &\varepsilon^{2}\left(  L_{\varepsilon
}\right)  ^{3}\sup_{\left\vert \theta\right\vert =1}\int_{\mathbb{R}^{3}}%
d\eta\left(  L_{\varepsilon}\theta\cdot\int_{0}^{\frac{T}{L
_{\varepsilon}}}\nabla_{x}\Psi\left(  v\tau-\eta\right)  d\tau\right)  ^{2}
\label{B1} \;.
\eea
As usual we assume $v=\left(  1,0,0\right)$ and study parallel and longitudinal components separately.  
We have
\begin{equation}
\theta_{1}\int_{0}^{\frac{T}{L_{\varepsilon}}%
}\frac{\partial\Psi}{\partial x_{1}}\left(  v\tau-\eta\right)  d\tau
=\theta_{1}\left[  \Psi\left(  \frac{T}{L_{\varepsilon}}v-\eta\right)
-\Psi\left(  -\eta\right)  \right] . \label{B2}%
\end{equation}
The longitudinal contribution is
\begin{equation}
\theta_{\bot}\cdot\int_{0}^{\frac{T}{L_{\varepsilon}}}\nabla_{x}\Psi\left(
v\tau-\eta\right)  d\tau=-\left(  \theta_{\bot}\cdot\eta_{\bot}\right)
\int_{0}^{\frac{T}{L_{\varepsilon}}}\frac{\partial\Psi}{\partial\left(
\left\vert x\right\vert \right)  }\left(  \left\vert v\tau-\eta\right\vert
\right)  \frac{d\tau}{\left\vert v\tau-\eta\right\vert } \;. \label{B3}%
\end{equation}
The integral on the right of (\ref{B3}) can be approximated in the form
$\frac{T}{L_{\varepsilon}}Q\left(  \eta\right)  $ where $Q\left(
\eta\right)  $ decreases fast as $\left\vert \eta\right\vert \rightarrow
\infty$ if $\frac{T}{L_{\varepsilon}}\leq1.$ If $\frac{T}{L
_{\varepsilon}}\geq1$ we can approximate the integral by a function 
$W=W\left(  \eta_{\bot}\right)  $ which is rather independent of $\frac
{T}{L_{\varepsilon}}$ for the values of $\eta_{1}$ in the interval $\left[
0,\frac{T}{L_{\varepsilon}}\right]  .$ On the other hand the right-hand
side of (\ref{B2}) can be estimated by a function 
decreasing fast if $\frac{T}%
{\ell_{\varepsilon}}\geq1$ and as $\frac{T}{L_{\varepsilon}}$ if $\frac{T}{L_{\varepsilon}}\leq1.$ 

Suppose first that $\frac{T}{L_{\varepsilon}}\leq1.$ Then $\sigma\left(  T;\varepsilon\right)  $ can be approximated as
$C\varepsilon^{2}\left(  L_{\varepsilon}\right)  ^{3}T^{2}.$ In order to
have a kinetic limit we need to have $T_{L}\gg1,$ therefore
$1=\sigma\left(  T_{L};\varepsilon\right) $ yields $L_{\varepsilon}\ll\varepsilon^{-\frac{2}{3}}$. 
Moreover we have $T_{L}\sim\frac{C}{\varepsilon\left(  L_{\varepsilon
}\right)  ^{\frac{3}{2}}}$ if $\frac{T}{L_{\varepsilon}}\sim\frac{C%
}{\varepsilon\left(  L_{\varepsilon}\right)  ^{\frac{5}{2}}}\leq1.$ This
gives the results in the cases (i) and (ii) of the theorem.

If
$\frac{T}{L_{\varepsilon}}\geq1,$ then the contribution to
$\sigma\left(  T;\varepsilon\right)  $ of the term proportional to $\theta
_{1}$ in (\ref{B2}) is negligible and we obtain, using also (\ref{B3}), the
following approximation for $\sigma\left(  T;\varepsilon\right)  $:%
\[
\varepsilon^{2}\left(  L_{\varepsilon}\right)  ^{5} \sup_{\left\vert
\theta\right\vert =1}  \int
_{0}^{\frac{T}{L_{\varepsilon}}}d\eta_{1}\int d\eta_{\bot}(\theta_{\bot}
\cdot\eta_{\bot})^2\left(  Q\left(  \eta_{\bot}\right)  \right)
^{2}=C_{1}\varepsilon^{2}\left(  L_{\varepsilon}\right)  ^{5}\frac{T}%
{L_{\varepsilon}}=C_{1}\varepsilon^{2}\left(  L_{\varepsilon}\right)
^{4}T
\]
whence $T_{L}\sim\frac{C_{2}}{\varepsilon^{2}\left(  L_{\varepsilon
}\right)  ^{4}}$ as $\varepsilon\rightarrow0.$ Notice that 
$\frac{T_{L}}{L_{\varepsilon}}\geq1$ if $\frac{1}{\varepsilon^{2}\left(
L_{\varepsilon}\right)  ^{5}}\geq c_{0}>0,$ whence case (iii) follows. 
\end{proof}

We can now examine in which cases we can derive a Landau equation for
$f.$ This is not possible if $\frac
{T_{L}}{L_{\varepsilon}}\leq C$ because in that case the form of the
interaction potential in (\ref{S5E3}) implies that deflections separated by
times of order $T_{L}$ have correlations of order one and (\ref{I1E1}) would
not be satisfied. In this case a correlated model in the spirit
of the one discussed in Section \ref{ss:CorrCase}
would allow to describe the trajectories of the tagged particle.
We then restrict our attention to the case in which $L_{\varepsilon
}\ll\varepsilon^{-\frac{2}{5}}.$ Suppose that $T=h T_{L},$ $h>0$.
In this case we have the approximation %
\begin{align*}
& \varepsilon^{2}\left(  L_{\varepsilon}\right)  ^{3}\int_{\mathbb{R}^{3}%
}d\eta\left(  L_{\varepsilon}\theta\cdot\int_{0}^{\frac{T}{L
_{\varepsilon}}}\nabla_{x}\Psi\left(  v\tau-\eta\right)  d\tau\right)  ^{2}\\
& \sim\varepsilon^{2}\left(  L_{\varepsilon}\right)  ^{3}\int_{0}%
^{\frac{h}{\varepsilon^{2}\left(  L_{\varepsilon}\right)  ^{5}}}%
d\eta_{1}\int_{\mathbb{R}^{2}}d\eta_{\bot}(L_\ep\theta_{\bot}
\cdot\eta_{\bot})^2\left(  Q\left(  \eta_{\bot}\right)  \right)
^{2}\\
& \rightarrow h \int_{\mathbb{R}^{2}}d\eta_{\bot}(\theta_{\bot}
\cdot\eta_{\bot})^2\left(  Q\left(  \eta_{\bot}\right)  \right)
^{2}=2\kappa h\left\vert \theta_{\bot
}\right\vert ^{2}%
\end{align*}
where $\kappa>0.$ Then, using the Claim \ref{ClaimLand} we
obtain that $f$ satisfies%
\[
\left(  \partial_{t}f+v\partial_{x}f\right)  \left(  t,x,v\right)  =\kappa\Delta_{v_{\perp}}f\left(  t,x,v \right) \;.
\]

\begin{remark}   %is the case of grazing collision
Clearly it is possible to describe the difference between the
cases $L_{\varepsilon}\ll1$ and $L_{\varepsilon}\gtrsim1$ in terms of the domain of influence as
introduced in Remark \ref{RangeInter}. In the case of the potentials
having the forms (\ref{S5E3}), (\ref{S5E4}) this domain of
influence are the points $x=\left(  x_{1},x_{\bot
}\right)  $ with $x_{1}\in\left[  0,T_{L}\right]  ,$ $\left\vert x_{\bot
}\right\vert \leq C_{1}L_{\varepsilon}.$ 
In the
first case the tagged particle interacts at any given time with a large number
of scatterers. On the contrary if we assume (\ref{S5E4}) the
tagged particle at a given time would interact typically with zero scatterers,
and occasionally would interact with one scatterer. 
These rare interactions are weak
collisions, and the accumulation of many of them yields
the deflection of the particle velocity.
\end{remark}

\begin{remark}
The impossibility to obtain a Landau equation if $\varepsilon^{-\frac{2}{5}%
}\lesssim L_{\varepsilon}\ll\varepsilon^{-\frac{2}{3}}$ can be seen also in
the fact that the characteristic function for the deflections takes the form
\[
m_{hT_L}^{(\ep)}\left(  \theta \right)
  \rightarrow\exp\left(  -h^{2}%
\int_{\mathbb{R}^{3}}d\eta\left(  \theta\cdot\nabla_{x}\Psi\left(
\eta\right)  \right)  ^{2}\right)  \text{ as }\varepsilon\rightarrow0\;.
\]
The characteristic size of the deflections in this case is $h$
instead of the parabolic rescaling $\sqrt{h}$ which takes place in the
diffusive limit.
\end{remark}

\section{Spacial nonhomogeneous distribution of scatterers} \label{Nonhomog}

\subsection{Dynamics of a tagged particle in a spherical scatterer cloud with
Newtonian interactions.}\label{ss:clouds}

We have seen in Theorem \ref{eq:NtiCp} that it is not possible to have
spatially homogeneous generalized Holtsmark fields for Newtonian scatterers, i.e.\,having just
one sign for the charges and generating potentials of the form $\Phi\left(
x,\ep\right)  =\frac{\ep}{\left\vert x\right\vert }$. In this case it is natural to examine the dynamics of
a tagged particle in the field generated by random scatterer distributions in
bounded clouds. One of the simplest examples that we might consider is the
dynamics of a tagged particle in a spherical cloud of scatterers uniformly distributed. Theorem \ref{NewtonEquation}
ensures that in this case we cannot ignore the macroscopic average force
acting on the tagged particle due to the overall mass distribution on the
sphere. Let us assume that the cloud has a radius $R$ and that the tagged
particle moves inside this cloud in an orbit with a characteristic semiaxis of
order $\frac{R}{2}.$ Then we shall argue that a kinetic
description for the dynamics of the tagged particle is not possible, due to the
onset of correlations between the forces in macroscopic times. 

Assume that $N=\frac{4\pi R^{3}}{3}$ scatterers are distributed
independently and uniformly in the ball $B_{R}\left(
0\right)  ,$ with density one. A scatterer located in $x_{j}$ yields a potential
$\frac{\varepsilon}{\left\vert x-x_{j}\right\vert }.$ Since all the forces are
attractive, it follows that in the limit $N\rightarrow\infty$ there is a mean
force at each point $x\in B_{R}\left(  0\right),$ directed towards
the centre of the sphere and proportional to $\left\vert x\right\vert =r.$ Let the tagged particle 
have unit mass and move 
in an orbit around the center with characteristic length $r$ of order $R$, say $\frac{R}{2}$. The orbit can be expected to
experience random deflections due to the discreteness of the scatterer
distribution. To estimate the time scale in which these deflections take place,
we first remark that the potential energy is of order $C_{1}\varepsilon r^{2}$
where $C_{1}$ is just a numerical constant. 
Since the kinetic energy is $\frac
{V^{2}}{2}$, it follows that the velocity of the tagged particle is of order
$C_{2}\sqrt{\varepsilon}r.$ We recall that the Landau time scale $T_L$ is defined as the time in which
the velocity experiences a deflection comparable to itself. In the case of Coulomb potentials,
we have seen that $T_L$ differs from the
Boltzmann-Grad time scale $T_{BG}$ only by a logarithmic factor. The collision length
$\lambda_{\varepsilon}$ for a particle with velocity $V$ is given by
$\lambda_{\varepsilon}=\frac{\varepsilon}{V^{2}}=\frac{C_{3}}{r^{2}},$ 
therefore 
$T_{BG}=\frac{1}{V\left(  \lambda_{\varepsilon
}\right)  ^{2}}$. Including the effect of the Coulombian logarithm we would
obtain $T_{L}=\frac{1}{V\left(  \lambda_{\varepsilon}\right)  ^{2}}%
\log\left(  \frac{R}{\lambda_{\varepsilon}}\right)  .$ The mean free path is
then approximated if $R\rightarrow\infty$ as
$\ell_{\varepsilon}\simeq\frac{1}{\left(  \lambda_{\varepsilon}\right)  ^{2}%
}\log\left(  \frac{R}{\lambda_{\varepsilon}}\right)  =C_{4}R^{4}\log\left(
R\right) $ and
\[
\frac{\ell_{\varepsilon}}{R}\simeq C_{4}R^{3}\log\left(  R\right)  \simeq
C_{5}N\log\left(  N\right)  \ \ \text{as\ \ }N\rightarrow\infty\;,
\]
where $C_{5}$ is just a numerical constant (see \cite{BT} for a similar
estimate). Namely the mean free path is much larger than the length of the orbit.  

If the deviations in one orbit were described by a Landau equation, then they should be approximated 
by the sum of $\frac{\ell_{\varepsilon}}{R}\sim N\log\left(  N\right) $ independent random deflections with zero average.  
Denoting by $\sigma$ the relative change of velocity in each deflection, 
we would need $N\log\left(  N\right)  \sigma^{2}\sim 1$, whence the typical deviation in one
orbit would be $\frac{1}{\sqrt{N\log\left(  N\right)  }}$ 
and the corresponding change of velocity $\frac
{V}{\sqrt{N\log\left(  N\right)  }}\simeq\frac{\sqrt{\varepsilon}r}{\sqrt
{N\log\left(  N\right)  }}.$ But the period of the orbit is $\frac{1}%
{\sqrt{\varepsilon}}$ and then the typical change of position of the tagged particle
in one period would be$\frac
{R}{\sqrt{R^{3}\log\left(  R\right)  }}\rightarrow 0$.

We finally remark that the onset of large correlations in the forces acting on
the tagged particle in times of the order of the orbit period would also take place if the scatterers move, since they would
move in elliptical orbits with the same period, given that the density of
scatterers is constant in the cloud. The situation could change, however, for a
nonspherical cloud, but we will not pursue this analysis here.

\subsection{On the derivation of kinetic equations with a Vlasov term.}

If we assume that the scatterers are not distributed in a
spatially homogeneous way, then it is possible to obtain limit equations for $f$ containing
Vlasov terms. 

We will restrict here the analysis to particles interacting by means
of Coulomb potentials $\Phi\left(  x,\varepsilon\right)  =\frac{\varepsilon}{\left\vert x\right\vert
}$. We have shown that, in order to have random force fields which are spatially
homogeneous, we need to assume electroneutrality. Let us restrict ourselves to
the case in which there are only two types of charges $+1$ and $-1$ and
that the scatterers have these charges with probability $\frac{1}{2}$,
independently on the probability measure yielding their spatial distribution.
We will assume that the two types of scatterers are distributed in space according to inhomogeneous
Poisson measures with densities 
\begin{equation}
\rho_{+}\left(  x\right)  =\frac{1}{2}+\delta_{\varepsilon}F_{+}\left(
\frac{x}{\ell_{\varepsilon}}\right)  \ ,\ \ \rho_{-}\left(  x\right)
=\frac{1}{2}+\delta_{\varepsilon}F_{-}\left(  \frac{x}{\ell_{\varepsilon}%
}\right)  \label{Ch1}%
\end{equation}
respectively. Here $\ell_{\varepsilon}$ denotes the corresponding mean free path for one tagged particle moving
in the field of scatterers with density approximately equal to one. By \eqref{S9E4}, \eqref{TG1} and \eqref{P4E1} we 
may assume
\[
\ell_{\varepsilon}=\frac{1}{\varepsilon^{2}\log\left(  \frac{1}{\varepsilon
}\right)  }.
\]
Moreover, the parameter $\delta_{\varepsilon}>0$ converges to zero as $\varepsilon
\rightarrow0.$ Its precise dependence on $\varepsilon$ will be fixed below. Finally we
assume also that the functions $F_{+}\left(  y\right)  ,\ F_{-}\left(
y\right)  $ decay fast for large values of $\left\vert y\right\vert $ (we
could assume for instance that these functions are compactly supported).

The force produced by a given configuration has the form%
\[
-\frac{\varepsilon}{2}\sum_{j,k}\left[  \frac{\left(  x-x_{j}^{+}\right)
}{\left\vert x-x_{j}^{+}\right\vert ^{3}}-\frac{\left(  x-x_{k}^{-}\right)
}{\left\vert x-x_{k}^{-}\right\vert ^{3}}\right]\;,
\]
where $x_j^+, x_k^-$ are the locations of scatterers with charges $+1$, $-1$.
These forces yield deflections described by a Landau equation. In addition, the slight fluctuations of the
density yield a nonvanishing mean field which can be approximated by 
\begin{align*}
& -\varepsilon\delta_{\varepsilon}\int_{\mathbb{R}^{3}}\left[
F_{+}\left(  \frac{y}{\ell_{\varepsilon}}\right)  -F_{-}\left(  \frac{y}%
{\ell_{\varepsilon}}\right)  \right]  \frac{\left(  x-y\right)  }{\left\vert
x-y\right\vert ^{3}}dy\\
& =-\varepsilon\delta_{\varepsilon}\ell_{\varepsilon}\int
_{\mathbb{R}^{3}}\left[  F_{+}\left(  \xi\right)  -F_{-}\left(  \xi\right)
\right]  \frac{\left(  \frac{x}{\ell_{\varepsilon}}-\xi\right)  }{\left\vert
\frac{x}{\ell_{\varepsilon}}-\xi\right\vert ^{3}}d\xi \;.
\end{align*}
The mean field variation is of order one in regions with
macroscopic size $\ell_{\varepsilon}.$ The macroscopic time scale is also
$\ell_{\varepsilon}.$ Therefore, the change induced by these terms in the
macroscopic time scale is of order $\varepsilon\delta_{\varepsilon}%
\ell_{\varepsilon}^{2}.$ 

We select $\delta_{\varepsilon}$ in order to make
this quantity of order one, i.e. $\varepsilon\delta_{\varepsilon}%
\ell_{\varepsilon}^{2}=1,$ whence:%
\begin{equation}
\delta_{\varepsilon}=\varepsilon^{3}\left(  \log\left(  \frac{1}{\varepsilon
}\right)  \right)  ^{2} \;.\label{Ch2}%
\end{equation}
We then obtain that $f_{\varepsilon}\left(  \ell_{\varepsilon}t,\ell
_{\varepsilon}x,v\right)  \rightarrow f\left(  t,x,v\right)  $ where $f$
solves the Vlasov-Landau equation:%
\begin{align}
\partial_{t}f+v\partial_{x}f+g\partial_{v}f  & =\kappa\Delta_{v^{\bot}%
}f\ ,\ \kappa>0\label{Ch3}\\
g\left(  x\right)    & :=-\int_{\mathbb{R}^{3}}\left[  F_{+}\left(
\xi\right)  -F_{-}\left(  \xi\right)  \right]  \frac{\left(  x-\xi\right)
}{\left\vert x-\xi\right\vert ^{3}}d\xi \;.\label{Ch4}%
\end{align}
Notice that the distributions of charges $F_{+},F_{-}$ must be chosen in such
a way that we do not have periodic orbits for the tagged particle, since then
there would be correlations and the Landau
equation would fail as in Section \ref{ss:clouds}.

It is possible to derive Vlasov-Boltzmann or Vlasov-Landau equations for other types of
long range potentials like the ones considered in this paper. One obtains
mean field forces of the same order of magnitude as the
Landau or Boltzmann terms if the size of the inhomogeneities is chosen in a
suitable way as indicated above. An attempt for a rigorous derivation of this type of equations
is provided in \cite{DSS17}.

\section{Concluding remarks.} \label{sec:CR}

We have developed a formalism which allows to obtain the kinetic equation
describing the evolution of a tagged particle moving in a field of fixed
scatterers (Lorentz gas) distributed in the whole three-dimensional space according to
a Poisson measure with density of order one. Each scatterer is the centre of an interaction potential
which decays at infinity as a power law $\frac{1}{\left\vert x\right\vert
^{s}}$ with $s>\frac{1}{2}.$ 

We have first studied the properties of the random force field generated by
the scatterers and, in particular, the conditions under which this field is invariant under translations. In the case of potentials decreasing
for large $\left\vert x\right\vert $ as $\frac{1}{\left\vert x\right\vert
^{s}}$ with $s\leq1$ some ``electroneutrality'' of the system must be imposed,
either by means of the addition of a background with opposite charge density
or using charges with positive and negative signs. 

We have then studied the conditions in the interactions which allow to
obtain a kinetic description for the dynamics of the tagged particle. To this
end, the interaction between the tagged particle and the scatterers must be
weak enough to ensure that the mean free path
%, i.e. the length that the tagged
%particle must travel before having a change on its velocity comparable to its
%initial value, 
is much larger than the typical distance among the scatterers.
Under this assumption, we have three main
possibilities. If the fastest process yielding particle
deflections are binary collisions with single
scatterers, the resulting equation is the linear Boltzmann equation. If, on the contrary, the deflections
due to the accumulation of a large number of small interactions yield a
relevant change in the direction of the velocity before
a binary collision takes place, then we can have a Landau type dynamics. 
We have denoted as Landau time $T_{L}$ the time scale in
which such macroscopic deflections take place. In order to be able
to describe the evolution of the tagged particle by means of a Landau equation,
we have shown also that deflections experienced by the
particle over times of order $T_{L}$ must be uncorrelated. We have provided
examples of potentials for which this lack of correlations does not take
place. In such cases, we cannot expect to have a single PDE describing the
probability distribution in the particle phase space. Instead, the correlations between 
macroscopic deflections must be taken into account.

\bigskip
\noindent
\textbf{Acknowledgment.}
We thank Chiara Saffirio and Mario Pulvirenti for interesting discussions on the topic. The authors acknowledge support through the CRC 1060
\textit{The mathematics of emergent effects}
at the University of Bonn that is funded through the German Science
Foundation (DFG). S.\,acknowledges hospitality at HCM Bonn as well as support from the DFG grant 269134396.


\begin{thebibliography}{999}                                                                                              %

\bibitem{Al78}
Alder, B. J.: Computer results on transport properties. In: Garrido L., Seglar P., Shepherd P.J.,
Stochastic Processes in Nonequilibrium Systems,
\emph{Lecture Notes in Physics} {\bf 84}, Springer, Berlin, Heidelberg, 1978

\bibitem{vB82}
van Beijeren, H.: Transport properties of stochastic Lorentz models, \emph{Rev. Mod. Phys.} {\textbf{54}}, 1982

\bibitem{BNP}
Basile, G., Nota, A. and Pulvirenti, M.:
A diffusion limit for a test particle in a random distribution of scatterers.
\emph{J. Stat. Phys.},  \textbf{155}(6), 1087--1111, 2014

\bibitem {BT}Binney, J., Tremaine S.: \emph{Galactic Dynamics} (Second
edition). Princeton University Press, Princeton, 2008

\bibitem {CH}Chandrashekar, S.: Stochastic Problems in Physics and Astronomy.
\emph{Rev. Mod. Phys.} \textbf{15}, 1, 1943

\bibitem {CH1}Chandrasekhar, S., Von Neumann, J.: The Statistics of the
Gravitational Field Arising from a Random Distribution of Stars. I. The Speed
of Fluctuations. \emph{Astrophysical J.}, \textbf{95}, 489--531, 1941

\bibitem {CH2}Chandrasekhar, S., Von Neumann, J.: The Statistics of the
Gravitational Field Arising from a Random Distribution of Stars. II. The speed
of fluctuations; dynamical friction; spatial correlations. \emph{Astrophysical
J.}, \textbf{97}, 1--27, 1943

\bibitem {CH3}Chandrasekhar, S.: The Statistics of the Gravitational Field
Arising from a Random Distribution of Stars. III. The Correlations in the
Forces Acting at Two Points Separated by a Finite Distance.
\emph{Astrophysical J.}, \textbf{99}, 25--46, 1944

\bibitem {CH4}Chandrasekhar, S.: The Statistics of the Gravitational Field
Arising from a Random Distribution of Stars. IV. The Stochastic Variation of
the Force Acting on a Star. \emph{Astrophysical J.}, \textbf{99}, 47--58, 1944

\bibitem {DR}Desvillettes, L. and Ricci, V.: A rigorous derivation of a linear
kinetic equation of Fokker--Planck type in the limit of grazing collisions.
\emph{J. Stat. Phys.} \textbf{104}, 1173--1189, 2001

\bibitem{DSS17} Desvillettes, L., Saffirio, S. and Simonella, S.:
A derivation of the linear Boltzmann-Vlasov equation from particle systems.
\emph{In preparation}

\bibitem {DGL}D\"urr, D., Goldstein, S. and Lebowitz, J.: Asymptotic motion of
a classical particle in a random potential in two dimensions: Landau model.
\emph{Comm. Math. Phys.} \textbf{113}, 209--230, 1987

\bibitem {DP}Desvillettes, L. and Pulvirenti, M.: The linear Boltzmann
equation for long--range forces: a derivation from particle systems.
\emph{Models Methods Appl. Sci. } {\textbf{9}}, 1123--1145, 1999

\bibitem {F} Feller, W.: An Introduction to Probability Theory and Its Applications. Vol. 2. (Second Edition), Wiley Series in Probability and Statistics, 1971

\bibitem {G}Gallavotti, G.: Grad-Boltzmann limit and Lorentz's Gas.
\emph{Statistical Mechanics. A short treatise.} Appendix 1. A2. Springer,
Berlin, 1999

\bibitem{Ha74}
Hauge, E. H.: What can one learn from Lorentz models? in: Transport Phenomena, \emph{Lecture Notes in Physics} {\textbf{31}}, 1974

\bibitem{H}Holtsmark, J.: Uber die Verbreiterung von Spektrallinien. \emph{Annalen der Physik}. \textbf{363}(7), 577--630, 1919

\bibitem {KP}Kesten, H., Papanicolaou, G.: A limit theorem for stochastic
acceleration. \emph{Comm. Math. Phys.} \textbf{78}, 19--63, 1981

\bibitem{Ki03} Kiessling, M. K.-H.: Mathematical Vindications of the ``Jeans Swindle''.
\emph{Adv.\,Appl.\,Math.} {\bf 31}, 132-149, 2003

\bibitem {K} Kolokoltsov, V. N.: Markov processes, semigroups and generators. DeGruyter Studies in Mathematics {\bf 38}, DeGruyter, 2011

\bibitem {KR}Komorowski, T. and Ryzhik, L.: Diffusion in a weakly random
Hamiltonian flow. \emph{Comm. Math. Phys.} \textbf{263}, 273--323, 2006

\bibitem {LL1}Landau, L.D., Lifshitz, E.M.: Mechanics. \emph{Course of
Theoretical Physics}. vol.\textbf{1}. Pergamon press, Oxford, 1960

\bibitem {LL2}Landau, L.D., Lifshitz, E.M.: Physical Kinetics. \emph{Course of
Theoretical Physics}. vol.\textbf{10}. Pergamon press, Oxford--Elmsford, N.Y., 1981

\bibitem {L}Lorentz, H.A.: The motion of electrons in metallic bodies.
\emph{Proc. Acad. Amst.} \textbf{7}, 438--453, 1905

\bibitem{MvN68}
Mandelbrot B., van Ness J.: Fractional brownian motions, fractional noises and applications.
{\em SIAM Rev.} {\bf 10}, 4, 1968

\bibitem {MN} Marcozzi, M., Nota, A.: Derivation of the linear Landau equation
and linear Boltzmann equation from the Lorentz model with magnetic field.
\emph{J. Stat. Phys.} \textbf{162}, 1539--1565, 2016

\bibitem{Pi81}
Piasecki, J.: Self Diffusion in Fluids with Weak Long-Range Forces.
\emph{J. Stat. Phys.} \textbf{26}(2), 375--396, 1981

\bibitem {S}Spohn, H.: The Lorentz flight process converges to a random flight
process. \emph{Comm. Math. Phys.} \textbf{60}, 277--290, 1978

\bibitem {S1}Spohn, H.: Large scale dynamics of interacting particles. Texts
and Monographs in Physics, Springer Verlag, Heidelberg, 1991

\bibitem {S2}Spohn, H.: Kinetic equations from Hamiltonian dynamics: Markovian
limits. \emph{Rev. Mod. Phys.} \textbf{53}, 569--615, 1980


\end{thebibliography}
\end{document}